\documentclass[aps,pra,11pt,onecolumn,nofootinbib,tightenlines]{revtex4-2}

\usepackage{mathrsfs}
\usepackage{physics}
\usepackage{graphicx}
\usepackage{comment}
\usepackage{amsmath}
\usepackage{amsfonts}
\usepackage{amssymb}
\usepackage{braket}
\usepackage{amsthm}
\usepackage{caption}
\usepackage{lipsum}
\usepackage{subcaption}
\usepackage{mwe}
\usepackage{enumerate}
\usepackage{appendix}
\usepackage{tikz}
\usepackage{enumitem}
\usepackage{booktabs}
\usepackage{array}
\usepackage{forest}
\usepackage{mathtools}
\usepackage[colorlinks=true,urlcolor=blue,citecolor=blue,linkcolor=blue]{hyperref}

\usepackage[most,breakable]{tcolorbox}

\newtheorem{theorem}{Theorem}[section]
\newtheorem{lemma}[theorem]{Lemma}
\newtheorem{proposition}[theorem]{Proposition}
\newtheorem{definition}[theorem]{Definition}

\newtheorem{result}{Result}

\newtheorem{remark}[theorem]{Remark}
\let\oldremark\remark
\renewcommand{\remark}{\oldremark\upshape}

\newtheorem{corollary}[theorem]{Corollary}

\newcommand{\bigket}[1]{\Big| #1 \Big\rangle}
\newcommand{\bigbra}[1]{\Big \langle #1 \Big |}

\def \d {\mathrm{d}} 
\newcommand{\D}{\mathcal{D}}
\newcommand{\N}{\mathcal{N}}
\newcommand{\M}{\mathcal{M}}

\newcommand{\widebar}[1]{\overline{\mkern-4mu#1\mkern-1mu}}

\DeclareMathOperator*{\opt}{opt}
\DeclareMathOperator*{\Trm}{Tr}
\DeclareMathOperator*{\Op}{Op \!\!\,}

\begin{document}

\title[Quantum conditional entropies]{\Large Quantum conditional entropies from convex trace functionals}

\begin{abstract}
    We study geometric properties of trace functionals that generalize those in [Zhang, Adv. Math. 365:107053 (2020)], arising from a novel family of conditional entropies with applications in quantum information theory. Building on new convexity results for these functionals, we establish data-processing inequalities and additivity properties for our entropies, demonstrating their operational significance. We further prove completeness under duality, chain rules, and various monotonicity properties for this family. Our proofs draw on tools from complex interpolation theory, multivariate Araki--Lieb--Thirring inequalities, variational characterizations of trace functionals, and spectral pinching techniques.

\end{abstract}

\author{Roberto Rubboli}
\email{roberto.rubboli@u.nus.edu}
\affiliation{Centre for Quantum Technologies, National University of Singapore, Singapore 117543, Singapore}
\affiliation{Department of Mathematical Sciences, University of Copenhagen, Universitetsparken 5, 2100 Denmark}

\author{Milad M. Goodarzi}
\affiliation{Centre for Quantum Technologies, National University of Singapore, Singapore 117543, Singapore}

\author{Marco Tomamichel}
\affiliation{Department of Electrical and Computer Engineering,
National University of Singapore, Singapore 117583, Singapore}
\affiliation{Centre for Quantum Technologies, National University of Singapore, Singapore 117543, Singapore}

\maketitle


\section{Introduction}
Quantum conditional entropies quantify the uncertainty about the state of a quantum system $A$ from the perspective of an observer with access to a potentially correlated quantum system $B$, often referred to as side information. These measures of conditional uncertainty play a central role in quantum information theory, with applications ranging from quantum cryptography~\cite{pirandola20_review} and quantum Shannon theory~\cite{wilde13_book, hayashi17_book, holevo12_book} to resource theories~\cite{chitambar19_review}.
While von Neumann entropy often takes a prominent role in such applications, other notions of conditional entropy are necessary to give a more fine-grained picture when we move beyond asymptotic scenarios. In cryptography, conditional min-entropy quantifies how much secret key can be extracted from a random variable $A$ when an eavesdropper has access to correlated quantum side information $B$ that may assist in estimating the secret key~\cite{bennett95_privacyamp}. A larger family of quantum R\'enyi entropies was employed recently in~\cite{arqand2025generalized} to give the tightest known security analysis for quantum key distribution.
Applications of general conditional entropies in information theory include recent applications in the analysis of error exponents for source compression and channel coding~\cite{renes2025_reliability,li2025_reliability}.

Most prominent notions of quantum conditional entropy in the literature arise from an underlying quantum R\'enyi divergence and admit two canonical constructions. To make the discussion concrete, consider the family of $\alpha$-$z$ relative entropies~\cite{jaksic2012entropic,audenaert13_alphaz}\footnote{In~\cite{jaksic2012entropic}, this quantity was introduced in the context of quantum dynamical systems and entropic fluctuations in non-equilibrium statistical mechanics, while in~\cite{audenaert13_alphaz}, it was studied from an information-theoretic perspective.} (see Definition~\ref{alpha-z relative entropy}) that unifies the most prominent quantum divergences. The two canonical conditional entropies of $A$ conditioned on $B$ are then defined as follows:
\begin{equation}\label{dualdefinition}
H_{\alpha,z}^\downarrow(A|B)_\rho = -D_{\alpha,z}(\rho_{AB}\|I_A \otimes \rho_B), \qquad H_{\alpha,z}^\uparrow(A|B)_\rho = -\inf_{\sigma_B}D_{\alpha,z}(\rho_{AB}\|I_A \otimes \sigma_B),
\end{equation}
where the infimum is taken over all quantum states $\sigma_B$ on system $B$ and $\rho_B$ denotes the marginal of $\rho_{AB}$ on $B$. In particular, these include the two subfamilies of Petz-type entropies~\cite{tomamichel08_aep} for $z = 1$, as well as the sandwiched-type entropies~\cite{lennert13_renyi} for $z = \alpha$. 

In this work, we go beyond this framework and define a three-parameter family of entropies
\begin{equation}
\label{eq:hazlambda}
    H_{\alpha,z}^\lambda(A|B)_\rho = \frac{1}{1-\alpha} \opt_{\sigma_B} \log \Trm\left(\rho_{AB}^\frac{\alpha}{2z} \left(I_A \otimes \rho_B^{\frac{(1-\lambda)(1-\alpha)}{2z}}\sigma_B^{\frac{\lambda(1-\alpha)}{z}} \rho_B^{\frac{(1-\lambda)(1-\alpha)}{2z}}\right) \rho_{AB}^\frac{\alpha}{2z}\right)^z,
\end{equation}
where $\opt$ denotes an infimum (resp., supremum) when the exponent of $\sigma_B$ is negative (resp., non-negative). This definition recovers the entropies in~\eqref{dualdefinition} for $\lambda \in \{0,1\}$ and can therefore be viewed as an interpolation (or extrapolation) between these two cases. (Its precise definition and connection to previously studied quantities are discussed in Section~\ref{sec:notation} and illustrated in Figure~\ref{tree of conditional entropies}.)

Our investigation is motivated by both structural and operational considerations. On the structural side, the previously studied Petz-type and sandwiched-type conditional R\'enyi entropies arise from distinct constructions and parameter regimes, yet they are connected via entropic duality relations~\cite{tomamichel13_duality}, suggesting the presence of a common underlying structure that our framework reveals. From this viewpoint, our new conditional entropies provide a natural complement to the previously studied families, completing the picture and unifying several arguments. On the other hand, in the commutative setting, it was already known that certain families of conditional entropies with operational meaning (see, e.g.,~\cite{hayashi2016equivocations}) do not arise from the canonical constructions. Our framework recovers these as commutative special cases. See Section~\ref{sec:classical} for further discussion.

\medskip
\paragraph*{Note added:} After the electronic preprint of the present work appeared, the new conditional entropies were already used to characterize strong converse exponents for randomness extraction, both in the commutative~\cite{li2025two} and in the general quantum~\cite{rubboli2026strong} setting. Additional applications of our quantities can therefore be expected, further motivating their study here.

\medskip
The goal of this work is to establish several fundamental mathematical properties of the entropies in~\eqref{eq:hazlambda}, including the data-processing inequality (DPI), additivity under tensor products, duality, and chain rules, which are essential for their use in quantum information theory. These results are summarized in Section~\ref{sec:results}. We then discuss the relationship between properties of entropic quantities and geometric properties of certain trace functionals in Section~\ref{Multivariate trace functionals}.

\subsection{Main results}
\label{sec:results}
The following results hold for all parameter triples $(\alpha,z,\lambda)$ in the DPI region of $H_{\alpha,z}^\lambda$, denoted $\mathcal{D}$ (cf.~Definition~\ref{def:DPIregion} and Figure~\ref{Figure alpha,z,lambda}), consisting of the union of the two sets $\mathcal{D}_1 = \{0 < \alpha < 1, \ 1-z \leq \alpha \leq  z, \ 1-\frac{z}{1-\alpha} \leq \lambda \leq 1 \}$ and $\mathcal{D}_2 = \{1<\alpha<\infty, \ \alpha-1 \leq  z \leq \alpha \leq 2z, \ 1+\frac{z}{1-\alpha} \leq \lambda \leq 1 \}$. With this notation, we can state our main results.

\medskip
\paragraph*{Data-processing inequality:}
A fundamental requirement for conditional entropies as measures of conditional uncertainty is monotonicity under bipartite channels of the form $\mathcal{M} \otimes \mathcal{N}$, where $\mathcal{M}$ is a convex combination of unitary channels acting on system~$A$ (a mixing channel), and $\mathcal{N}$ is an arbitrary channel acting on system~$B$. This is necessary to be consistent with the interpretation as a measure of conditional uncertainty: uncertainty cannot decrease when mixing $A$ or acting on the side information $B$. In addition, one must allow for the possibility that the channel on~$A$ depends on the outcome of a measurement performed on~$B$. The reader is referred to Section~\ref{sec:DPI} for a detailed discussion of the operational significance of these requirements.

\begin{figure}
\centering

\begin{subfigure}[t]{0.3\textwidth}
\centering
\begin{tikzpicture}
\node[inner sep=0] (img1)
{\includegraphics[width=0.9\linewidth]{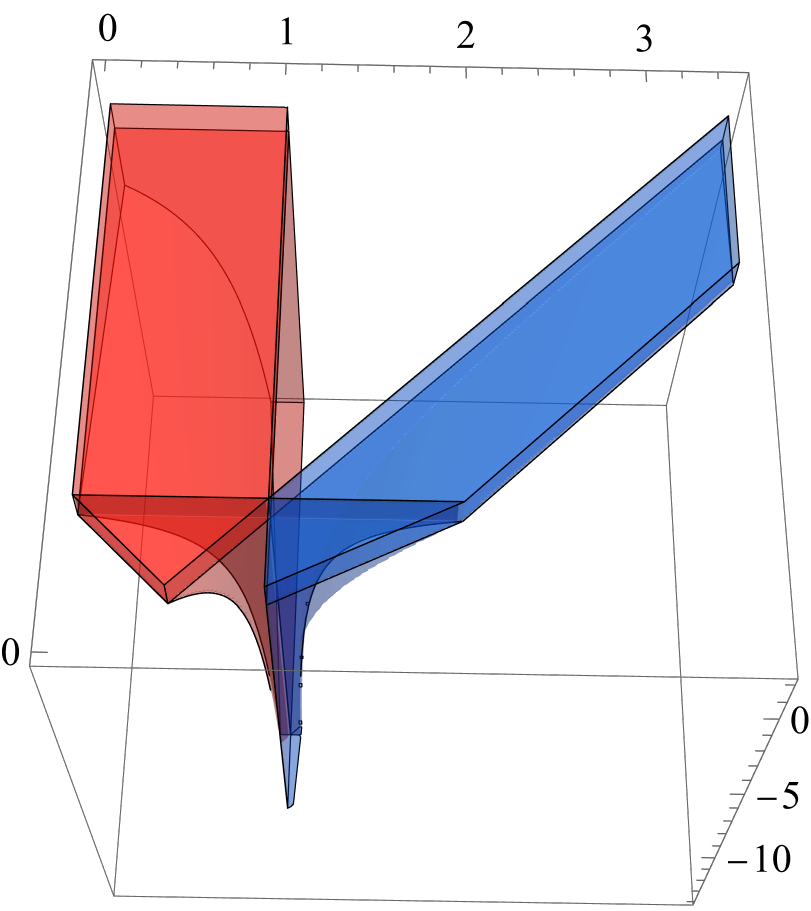}};

\node[anchor=north] at ([yshift=0pt]img1.south) {\ensuremath{\alpha}};

\node[anchor=east, rotate=0] at ([xshift=10pt,yshift=22pt]img1.west) { $z$};

\node[anchor=west, rotate=0] at ([xshift=-3pt,yshift=-60pt]img1.east) { $\lambda$};

\end{tikzpicture}
\end{subfigure}
\hspace{0.1\textwidth}
\begin{subfigure}[t]{0.4\textwidth}
\centering
\begin{tikzpicture}
\node[inner sep=0] (img2)
{\includegraphics[width=.7\linewidth]{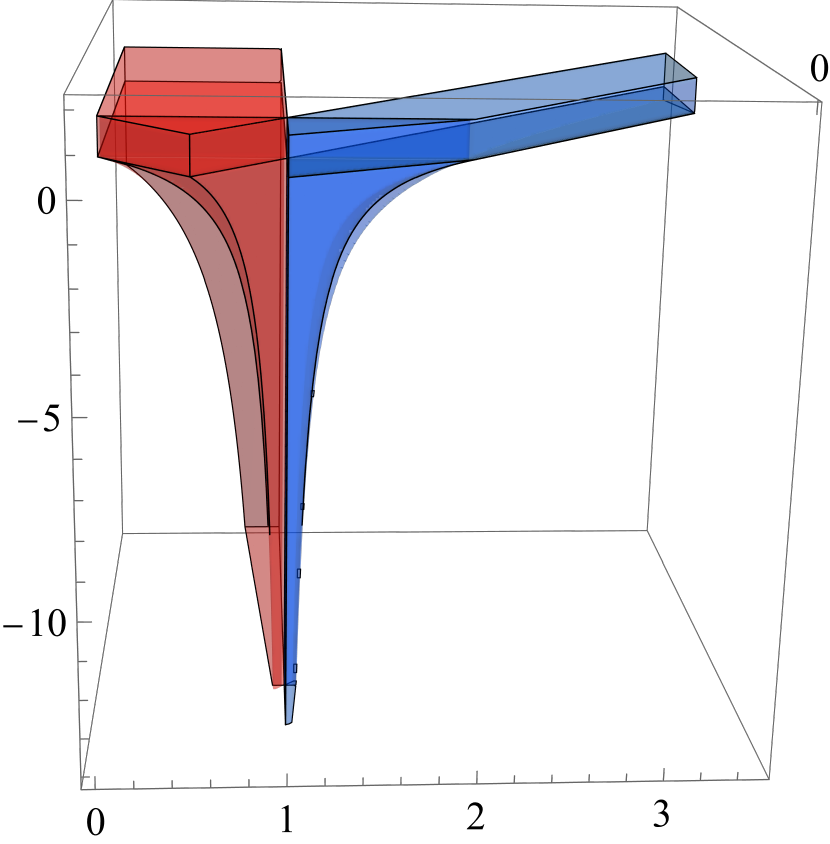}};

\node[anchor=north] at ([yshift=7pt]img2.south) {\ensuremath{\alpha}};

\node[anchor=east, rotate=0] at ([xshift=10pt,yshift=60pt]img2.west) {\ensuremath{z}};

\node[anchor=west, rotate=0] at ([xshift=-5pt,yshift=0pt]img2.east) {\ensuremath{\lambda}};

\end{tikzpicture}
\end{subfigure}

\caption{The figure illustrates the DPI region $\mathcal{D}$ in the parameter space $(\alpha, z, \lambda)$. (See Definition~\ref{def:DPIregion}.) The red region corresponds to $\mathcal{D}_1$, while the blue region corresponds to $\mathcal{D}_2$. The cross-section of $\mathcal{D}$ at $\lambda = t$ is independent of $t$ for $t \in [0,1]$ and corresponds to the DPI region of $D_{\alpha,z}$.}
\label{Figure alpha,z,lambda}
\end{figure}

\begin{result}[cf.\ Theorem~\ref{DPItheorem}]
Let $(\alpha,z,\lambda) \in \D$, and let $\rho_{AB}$ be a quantum state and
$\sigma_{A'B'} = \sum_i \mathcal{M}^i_{A\rightarrow A'}\otimes \mathcal{N}^i_{B\rightarrow B'}(\rho_{AB})$,
where each $\mathcal{M}^i_{A\rightarrow A'}$ is a sub-unital channel and $\{\mathcal{N}^i_{B\rightarrow B'}\}_i$ is a quantum instrument. Then,
\begin{equation}
    H^{\lambda}_{\alpha, z}(A|B)_{\rho} \leq H^{\lambda}_{\alpha, z}(A'|B')_{\sigma}.
\end{equation}

\end{result}

We remark that most studies in the existing literature establish the DPI for conditional entropies in canonical form by leveraging the DPI of the underlying divergence. However, this approach is not applicable in our case, as the underlying trivariate trace functional fails to satisfy monotonicity under quantum channels for the relevant parameter ranges (see Remark~\ref{counterexample}). To overcome this difficulty, our strategy is to first establish the convexity or concavity of the function
$\rho \mapsto \exp\big((1-\alpha) H^\lambda_{\alpha,z}(A|B)_{\rho}\big)$ in certain parameter ranges using Theorem \ref{log-concavity2}, and then apply standard arguments to promote this property to monotonicity under quantum channels.

\medskip
\paragraph*{Additivity:}
A key property is additivity with respect to tensor products. We prove
\begin{result}[cf.~Theorem~\ref{Additivity}]
Let $(\alpha,z,\lambda) \in \D$, and let $\rho_{AB}$ and $\sigma_{A'B'}$ be two quantum states. Then,
    \begin{equation}\label{additivityintro}
		H^\lambda_{\alpha,z}(AA'|BB')_{\rho \otimes \sigma} = H^\lambda_{\alpha,z}(A|B)_{\rho} + H^\lambda_{\alpha,z}(A'|B')_{\sigma}.
\end{equation}
\end{result}
To prove this, we first characterize the optimizers $\sigma_B$ using a fixed-point equation (cf.\ Theorem~\ref{fixed-point}), which we can then leverage to show additivity.

\medskip
\paragraph*{Duality:}

Another feature of conditional entropies is their duality relations.
 Let $\rho_{ABC}$ be a pure quantum state with marginals $\rho_{AB}$ and $\rho_{AC}$. The following relations hold~\cite{tomamichel08_aep,lennert13_renyi,tomamichel13_duality}:
\begin{align}
 	\label{relation 1}
 	\widetilde{H}^\uparrow_{\alpha}(A|B)_\rho + \widetilde{H}^\uparrow_{\hat{\alpha}}(A|C)_\rho &= 0, \quad \alpha^{-1} + \hat{\alpha}^{-1} = 2;\\ 
 	\label{relation 2}
 	\widebar{H}^\downarrow_{\alpha}(A|B)_\rho + \widebar{H}^\downarrow_{\hat{\alpha}}(A|C)_\rho &= 0, \quad \alpha + \hat{\alpha} = 2;\\ 
 	\label{relation 3}
 	\widetilde{H}^\downarrow_{\alpha}(A|B)_\rho + \widebar{H}^\uparrow_{\hat{\alpha}}(A|C)_\rho &= 0, \quad \alpha \hat{\alpha} = 1.
\end{align}
Duality relations have found many applications in cryptography as they relate the relative knowledge two parties can have about a quantum system.
Such relations can be used to define a dual to any conditional entropy, and we may ask if a particular family of conditional entropies is complete under duality, that is, if the duality relation is bijective. It is easy to see that the entropies $H_{\alpha,z}$ defined in \eqref{dualdefinition} are not complete. On the other hand, we show that this is the case for our entropies.

\begin{result}[cf.~Theorem \ref{Duality}]
Let $(\alpha,z,\lambda),(\hat{\alpha},\hat{z},\hat{\lambda}) \in \D$, and let $\rho_{ABC}$ be a pure quantum state. Then 
\begin{equation}\label{duality1intro}
    H^{\lambda}_{\alpha,z}(A|B)_\rho + H^{\hat{\lambda}}_{\hat{\alpha},\hat{z}}(A|C)_\rho = 0,
\end{equation}
whenever the parameters satisfy
\begin{equation}
    \frac{z}{1-\alpha}+\frac{\hat{z}}{1-\hat{\alpha}}=0, \qquad \frac{1-z}{1-\alpha}=\hat{\lambda}, \qquad \frac{1-\hat{z}}{1-\hat{\alpha}}=\lambda.
\end{equation}
\end{result}
    
This result has several implications for quantum conditional entropy. Computing the dual of a given conditional entropy is generally difficult, and the duals of $H_{\alpha,z}^{\uparrow}$ and $H_{\alpha,z}^{\downarrow}$ for general $z$ have remained unknown. We provide a resolution to this problem: If $(\alpha,z,0) \in \mathcal{D}$ (or equivalently, if $(\alpha,z,1) \in \mathcal{D}$), then
\begin{align}
&H_{\alpha,z}^{\uparrow}(A|B)_{\rho} + H^{\hat{\lambda}}_{\hat{\alpha},\hat{\alpha}}(A|C)_{\rho} = 0, \qquad \text{where} \qquad \hat{\alpha}=\frac{z}{z+\alpha-1},\qquad \hat{\lambda}=\frac{z-1}{\alpha-1};\\
&H_{\alpha,z}^{\downarrow}(A|B)_{\rho} + H^{\hat{\lambda}}_{\hat{\alpha},1}(A|C)_{\rho} = 0, \qquad \text{where} \qquad \hat{\alpha}=\frac{z-\alpha+1}{z},\qquad \hat{\lambda}=\frac{z-1}{\alpha-1}.
\end{align}
This result surprisingly reveals that the duals of $H_{\alpha,z}^{\uparrow}$ and $H_{\alpha,z}^{\downarrow}$ do not admit the canonical forms given in \eqref{dualdefinition} as $\hat{\lambda} \not\in \{0, 1\}$, in stark contrast to the established dualities for $\widebar{H}_\alpha$ and $\widetilde{H}_\alpha$. This implies that not all conditional entropies arise in canonical forms from an underlying divergence and resolves an open question posed in~\cite{ji2025fundamental}.

\medskip
\paragraph*{Chain rules:}
Chain rules reflect the compositional structure of multipartite systems. They generalize the equality $H(AB|C) = H(A|BC) + H(B|C)$ for the von Neumann entropy to R\'enyi entropies; however, in the latter case, only inequalities can be shown.
\begin{result}[cf.~Theorem~\ref{Final chain rule}]
    Let $(\alpha,z,\lambda), (\beta,w,\mu), (\gamma,v,\lambda) \in \D$, and let $\rho_{ABC}$ be a tripartite quantum state. Suppose the parameters satisfy the relations
	\begin{equation}\label{conditions intro}
		\frac{z}{1-\alpha}=\frac{w}{1-\beta}+\frac{v}{1-\gamma}, \qquad \frac{1-z}{1-\alpha} = \frac{1-w}{1-\beta}, \qquad \mu=\frac{1-v}{1-\gamma}.
	\end{equation}
	If $(\alpha-1)(\beta-1)(\gamma-1) > 0$, the following chain rule holds:
	\begin{equation}
		H^\lambda_{\alpha,z}(AB|C)_\rho\geq	H_{\beta,w}^\mu(A|BC)_\rho+	H^\lambda_{\gamma,v}(B|C)_\rho.
	\end{equation}
	The inequality is reversed when $(\alpha-1)(\beta-1)(\gamma-1)<0$.
\end{result}
The chain rules for $\widetilde{H}_\alpha^\uparrow$ and $\widetilde{H}_\alpha^\downarrow$ were first established in~\cite{dupuis2015chain} (see also \cite{beigi2023operator} for an alternative proof). In contrast, chain rules for $\widebar{H}_\alpha^\uparrow$ and $\widebar{H}_\alpha^\downarrow$, though desirable, have remained unknown. The general relation reveals several new chain rules that involve both entropies as special cases. If $(\alpha-1)(\beta-1)(\gamma-1) > 0$, we have the following new relations:
\begin{align}
		&\widebar{H}^\downarrow_{\alpha}(AB|C)_{\rho}\geq	\widebar{H}^\downarrow_{\beta}(A|BC)_{\rho}+	\widebar{H}^\downarrow_{\gamma}(B|C)_{\rho}, \ \ \ \text{when} \ \ \ \frac{1}{1-\alpha}=\frac{1}{1-\beta}+\frac{1}{1-\gamma};\\
	\label{eq:intro/petzchainrule}
		&\widebar{H}^\uparrow_{\alpha}(AB|C)_{\rho}\geq	\widebar{H}^\uparrow_{\beta}(A|BC)_{\rho}+	\widetilde{H}^\uparrow_{\gamma}(B|C)_{\rho}, \ \ \ \text{when} \ \ \ \frac{1}{1-\alpha}=\frac{1}{1-\beta}+\frac{\gamma}{1-\gamma};\\
        &\widetilde{H}^\downarrow_{\alpha}(AB|C)_{\rho}\geq	\widetilde{H}^\downarrow_{\beta}(A|BC)_{\rho}+	\widebar{H}^\downarrow_{\gamma}(B|C)_{\rho}, \ \ \ \text{when}  \ \ \ \frac{\alpha}{1-\alpha}=\frac{\beta}{1-\beta}+\frac{1}{1-\gamma}.
\end{align}
In each of the preceding inequalities, the arrows on the entropies indexed by $\alpha$ and $\gamma$ may simultaneously be reversed under the same parameter relation. The inequalities themselves are reversed when $(\alpha - 1)(\beta - 1)(\gamma - 1) < 0$ (see Corollaries~\ref{col:chain1} and~\ref{col:chain2}).

\subsection{Relation to multivariate trace functionals}\label{Multivariate trace functionals}

As will become clear below, the analysis of the three-parameter family $H_{\alpha,z}^\lambda$ is closely related to structural properties of certain optimized trace functionals. We will briefly overview this connection and its history here.

\smallskip
Let $p$, $q$, and $s$ be real parameters, and let $K$ be an arbitrary but fixed matrix. Consider trace functionals of the form
\begin{equation}\label{maineq}
\Psi_{p,q,s}(A, B) = \Tr \left(A^{\frac{p}{2}} K B^q K^\dagger A^{\frac{p}{2}}\right)^s,
\end{equation}
where $A$ and $B$ are positive semidefinite matrices. These trace functionals arise in the formulation of various entropic quantities in mathematical physics and quantum information, often in reparametrized or disguised forms. Their geometric properties have inspired the development of numerous methods and techniques in matrix analysis.

For instance, the celebrated Lieb concavity theorem \cite{lieb1973convex} asserts that the map $(A, B) \mapsto \Psi_{p,q,1}(A, B)$ is jointly concave whenever $p, q \geq 0$ and $p + q \leq 1$. Lieb's proof relies on an analytic extension of the functional and the Hadamard three-line theorem. As an application, he resolved the Wigner--Yanase--Dyson conjecture \cite{wigner1963information,lieb1973convex}, and, in collaboration with Ruskai ~\cite{lieb73subadd}, established the strong subadditivity of quantum entropy. Ando~\cite{ando79convex} later gave an alternative proof based on integral representations of operator convex and concave functions, and also showed joint convexity when $-1 \leq p \leq 0$ and $p + q = 1$.

Motivated by the introduction of sandwiched Rényi divergences in quantum information~\cite{lennert13_renyi,Wilde3}, further special cases with $s \neq 1$ were explored by Frank and Lieb~\cite{frank13_sandwiched}, using variational characterizations, and by Beigi~\cite{beigi13_sandwiched}, employing complex interpolation techniques. Extending these results, Zhang~\cite{zhang20_alphaz} confirmed a conjecture in~\cite{carlen18_alphaz} and fully characterized the parameter regimes for which the map $(A, B)\mapsto \Psi_{p,q,s}(A, B)$ is jointly convex or concave. Moreover, observe that $D_{\alpha,z}(\rho \| \sigma) = \frac{1}{\alpha-1}\log{\Psi_{p,q,s}(\rho, \sigma)}$, where $K=I$, and $p=\frac{\alpha}{z}, q=\frac{1-\alpha}{z}, s=z$. Using his joint convexity/concavity characterization of $\Psi$, Zhang determined the exact parameter ranges for which $D_{\alpha,z}$ is monotone under completely positive trace-preserving (CPTP) maps.

In a closely related but distinct direction, Carlen, Frank, and Lieb~\cite{carlen2016some} also studied trivariate trace functionals of the form
\begin{equation}\label{eq:trivariate}
\Phi_{p,q,r,s}(A, B, C) = \Tr \left(A^{\frac{p}{2}} K B^{\frac{q}{2}} C^r B^{\frac{q}{2}} K^\dagger A^{\frac{p}{2}}\right)^s,
\end{equation}
where $A$, $B$, and $C$ are positive semidefinite matrices. They showed in \cite[Corollary 3.3]{carlen2016some} that when $K=I$ and $p,q,r \neq 0$, the trivariate functional $(A, B, C) \mapsto \Phi_{p,q,r,1}(A, B, C)$ is never jointly concave, and is jointly convex only under the conditions $q = 2$, $p, r < 0$, and $-1 \leq p + r \leq 0$.

These pathological results for $\Phi$ suggest that, although $\Phi$ and $\Psi$ are formally similar, their behavior with respect to joint convexity/concavity is quite different.
Nonetheless, as we shall see (see Corollary~\ref{Psi-Phi} and Remark~\ref{Remark:PsiPhi}), $\Psi$ can be formulated as a variational expression involving $\Phi$, namely, for $t > s > 0$, we may write
\begin{equation}
\label{equivalence optimization}
        \Psi_{p,q,t}(A,B) = \opt_C \Phi_{p,q,\frac{1}{s}-\frac{1}{t},s}^{\frac{t}{s}}(A,B,C),
\end{equation}
where the optimization, denoted by $\opt$, is taken to be a supremum or an infimum over density matrices (i.e., positive matrices with unit trace), depending on whether the problem is concave or convex in $C$. This indicates that $\Phi$---despite its pathological behavior as a trivariate functional---exhibits desirable convexity properties once optimized over $C$.

Motivated by the applications considered in this paper, namely, the formulation and analysis of quantum conditional entropies, we adopt a more general framework. Specifically, we examine a setting in which the trace functional $\Phi$ acts on matrices defined on finite-dimensional bipartite systems, with the optimization carried out locally on only one of the subsystems, i.e., we study the functional
$(A,B) \mapsto \opt_{C} \Phi_{p,q,r,s}(A,I \otimes B, I \otimes C)$,
where $C$ is a density matrix on the second subsystem. In contrast to~\eqref{equivalence optimization}, when both subsystems are non-trivial, the optimization over $C$ no longer admits a closed-form solution. In this more general setting, we establish convexity/concavity properties for this new bivariate trace functional.
Our conditional entropies can then be expressed in terms of $\Phi$ as 
\begin{equation}\label{eq:intro-cond}
\exp\big( (1-\alpha) H_{\alpha,z}^\lambda(A|B)_\rho \big) = \opt_{\sigma_B} \Phi_{ \frac{\alpha}{z},\frac{\bar{\lambda}\bar{\alpha}}{z},\frac{\lambda \bar{\alpha}}{z},z}(\rho_{AB}, I_A \otimes \rho_B, I_A \otimes \sigma_B),
\end{equation}
where $\bar{\alpha} = 1-\alpha$ and $\bar{\lambda} = 1-\lambda$. Hence, several properties of the conditional entropies are connected to those of the trace functionals.

\medskip
\textbf{Layout of the paper.} In Section~\ref{sec:notation}, we formally define our new family of quantum conditional entropies, and explore its relationship with related quantities in the literature. In Section~\ref{Main technical results}, we formulate and prove our results concerning the convexity properties of trivariate trace functionals. Section~\ref{sec:additivity} examines properties of the optimizer and proves additivity of the entropy for tensor product states. In Section~\ref{sec:DPI}, we prove the data processing inequality over the relevant parameter ranges. In Section~\ref{sec:duality}, we establish the relevant duality relations. Section~\ref{sec:reparametrization} gives a reparameterization that allows us to define the DPI region and duality using linear equations. Section~\ref{sec:characterization} presents alternative formulations of $H^\lambda_{\alpha,z}$ as an asymptotic limit, which are essential for further analysis. In Section~\ref{sec:chainrules}, we derive a general chain rule, and in Section~\ref{sec:mono}, we investigate monotonicity with respect to the parameters. Section~\ref{sec:classical} treats the commutative case, and Section~\ref{section limits} considers pointwise limits of our entropies.
We conclude in Section~\ref{sec:conclusion} with a discussion of open problems.

\section{A new concept of conditional entropy} \label{sec:notation}

In this section, we formally define the new conditional entropy that forms the central focus of this work, and explore its relationship with related quantities in the literature.

\medskip
Throughout, $A$ and $B$ are finite-dimensional systems of possibly different dimensions. We denote by $\mathcal{L}(A,B)$ the set of all matrices acting from $A$ to $B$, and write $\mathcal{L}(A) := \mathcal{L}(A,A)$. We denote by $\mathcal{P}(A)$ the set of positive semidefinite matrices on $A$, and by $\mathcal{S}(A) \subset \mathcal{P}(A)$ the subset of unit-trace matrices, representing quantum states. We use the terms \textit{density matrix} and \textit{quantum state} interchangeably. For a bipartite matrix $X_{AB}$, we define the marginals $X_A = \Trm_B(X_{AB})$ and $X_B = \Trm_A(X_{AB})$ via partial trace over the respective subsystems. We adopt the convention that identity matrices are inserted implicitly in tensor products. In particular, for matrices $X_{AB}$ and $Y_B$, the product $X_{AB}Y_B$ is shorthand for $X_{AB}(I_A \otimes Y_B)$, where $I_A$ denotes the identity on $A$.

\subsection{A three-parameter family}

In the following definition, we define the negative powers in the sense of generalized inverses; i.e., when we take the negative powers of positive operators, we simply ignore the kernel. Moreover, $\rho^0$ equals the projector on the support of $\rho$.
\begin{definition}
\label{trivariate}
    Let $\alpha \in (0,1) \cup (1,\infty)$, $z \in (0,\infty)$, $\lambda \in \mathbb{R}$. Let $\rho \in \mathcal{S}(A)$ and $\tau, \sigma \in \mathcal{P}(A)$ such that $\rho \ll \tau$. If $(\lambda(1-\alpha)\geq 0 \wedge \tau \not \perp \sigma) \vee \rho \ll \sigma$, then 
\begin{equation}
\Upsilon_{\alpha,z}^\lambda(\rho, \tau,\sigma):=
\frac{1}{\alpha-1}\log{\Trm\Big(\rho^\frac{\alpha}{2z}\tau^{\frac{(1-\lambda)(1-\alpha)}{2z}}\sigma^{\frac{\lambda(1-\alpha)}{z}}\tau^{\frac{(1-\lambda)(1-\alpha)}{2z}}\rho^\frac{\alpha}{2z}\Big)^z}\,.
\end{equation}
In all other cases, $\Upsilon_{\alpha,z}^\lambda(\rho, \tau,\sigma) = + \infty$.
\end{definition}
With this in mind, and noting that $\rho_{AB} \ll I_A \otimes \rho_B$, we proceed to define the new conditional entropy as follows 
\begin{definition}
\label{new entropy}
Let $\alpha \in (0,1) \cup (1,\infty)$, $z \in (0,\infty)$, and $\lambda \in \mathbb{R}$. For a bipartite quantum state \(\rho_{AB}\), we define the following quantity:
\begin{equation}
	H^\lambda_{\alpha,z}(A|B)_{\rho} := 
    \begin{cases}
\sup_{\sigma_B} - \Upsilon^\lambda_{\alpha,z}\big(\rho_{AB},I_A \otimes \rho_B, I_A \otimes \sigma_B\big) & \lambda > 0 \\
\inf_{\sigma_B} -\Upsilon^\lambda_{\alpha,z}\big(\rho_{AB},I_A \otimes \rho_B, I_A \otimes \sigma_B\big) & \lambda < 0 
    \end{cases}
\end{equation}
For $\lambda = 0$, the optimization is taken to be a supremum if $\alpha<1$ and an infimum if $\alpha>1$.
\end{definition}
In the following, for compactness, we often denote the optimization with ``$\opt$" as
\begin{equation}
       H^\lambda_{\alpha,z}(A|B)_{\rho} =\opt_{\sigma_B}- \Upsilon^\lambda_{\alpha,z}\big(\rho_{AB},I_A \otimes \rho_B, I_A \otimes \sigma_B\big).
\end{equation}
Here, the optimization denoted by ``$\opt$" is defined as the supremum over the set of all quantum states \(\sigma_B\) when \(\lambda > 0\), and as the infimum when \(\lambda < 0\). The choice between infimum and supremum depends on the sign of the exponents in the states being optimized: if the optimization is concave, it represents a supremum; otherwise, it represents an infimum.

\begin{remark}
    Note that, based on a standard argument for lower/upper semicontinuity, in the region $\D$, the optimization over the states $\sigma_B$ in Definition~\ref{new entropy} is always achieved, allowing the infimum and supremum to be replaced by the minimum and maximum, respectively. We refer to Appendix~\ref{lower semicontinuity} for further details. Moreover, throughout the paper, when proving properties of $H_{\alpha,z}^\lambda(A|B)_\rho$, we typically restrict attention to full-rank states $\rho$. This entails no loss of generality, since the extension to arbitrary states follows from the continuity of $H_{\alpha,z}^\lambda(A|B)_\rho$ in $\rho$ (see Proposition~\ref{Continuity}).
\end{remark}
Finally, we introduce some definitions that will be used in the sequel.
We define
\begin{equation}\label{Q}
Q_{\alpha,z}^\lambda(\rho_{AB}|\sigma_B):=\exp{(\alpha-1)\Upsilon(\rho_{AB},I_A \otimes \rho_B, I_A \otimes \sigma_B)},
\end{equation}
and denote its optimized version by
\begin{equation}\label{Q2}
Q_{\alpha,z}^\lambda(A|B)_\rho := \opt_{\sigma_B} Q_{\alpha,z}^\lambda(\rho_{AB}|\sigma_B). \end{equation}
In this case, the symbol ``$\opt$" denotes the infimum over quantum states $\sigma_B$ when $\lambda(1-\alpha) < 0$, and the supremum otherwise.

\subsection{Connections to known quantities}
We now demonstrate that the new conditional entropy in Definition~\ref{new entropy} includes all previously studied forms of quantum conditional entropies, organizing them within a cohesive mathematical framework. To this end, we first demonstrate that the studied forms are specific instances of the upward or downward \(\alpha\)-\(z\) conditional entropies. As these, in turn, are special cases of \(H_{\alpha,z}^\lambda\) for \(\lambda = 0,1\), it follows that the new conditional entropy \(H_{\alpha,z}^\lambda\) encompasses all the studied forms.

Let us first introduce the $\alpha$-$z$ R\'enyi relative entropy~\cite{audenaert13_alphaz,zhang20_alphaz}. If 
\begin{definition}[$\alpha$-$z$ R\'enyi relative entropy]
\label{alpha-z relative entropy}
    Let $ \alpha \in (0,1) \cup(1,\infty) $, $z \in (0,\infty)$, $\rho \in \mathcal{S}(A)$ and $ \sigma \in \mathcal{P}(A)$. If $(\alpha<1 \wedge \rho \not \perp \sigma) \vee \rho \ll \sigma$, then 
\begin{equation}
\label{definition}
D_{\alpha,z}(\rho \| \sigma):=
\frac{1}{\alpha-1}\log{\Trm\big(\rho^\frac{\alpha}{2z}\sigma^\frac{1-\alpha}{z}\rho^\frac{\alpha}{2z}\big)^z}\,.
\end{equation}
In all other cases, $D_{\alpha,z}(\rho \| \sigma) = + \infty$.
\end{definition}
With this definition in place, we can define the optimized (arrow up $\uparrow$) and non-optimized (arrow down $\downarrow$) $\alpha$-$z$ conditional entropies
\begin{equation} 
\label{alpha-z conditional}
H_{\alpha,z}^\uparrow(A|B)_\rho := \sup_{\sigma_B} -D_{\alpha,z}(\rho_{AB}\|I_A \otimes \sigma_B) \,,\quad 
H_{\alpha,z}^\downarrow(A|B)_\rho := -D_{\alpha,z}(\rho_{AB}\|I_A \otimes \rho_B).
\end{equation}
Next, we demonstrate that several studied forms are specific instances of these conditional entropies.
For $z=\alpha$ and $z=1$, the $\alpha$-$z$ R\'enyi relative entropy reduces to the \textit{sandwiched R\'enyi divergence} $\tilde{D}_{\alpha}(\rho \| \sigma)$~\cite{lennert13_renyi, Wilde3} and the \emph{Petz--R\'enyi divergence} $\widebar{D}_{\alpha}(\rho \| \sigma)$~\cite{petz1986quasi,tomamichel16_book}.

Accordingly, the arrow up and down Petz conditional entropies~\cite{petz1986quasi,tomamichel08_aep} and sandwiched conditional entropies~\cite{lennert13_renyi,beigi13_sandwiched} are defined in terms of the Petz and sandwiched R\'enyi divergences. We have
\begin{align}
&\widetilde{H}_{\alpha}^\uparrow(A|B)_\rho := \sup_{\sigma_B}- \widetilde{D}_{\alpha}(\rho_{AB}\|I_A \otimes \sigma_B) \,,\qquad 
\widetilde{H}_{\alpha}^\downarrow(A|B)_\rho := -\widetilde{D}_{\alpha}(\rho_{AB}\|I_A \otimes \rho_B) \,,\\
&\widebar{H}_{\alpha}^\uparrow(A|B)_\rho := \sup_{\sigma_B} - \widebar{D}_{\alpha}(\rho_{AB}\|I_A \otimes \sigma_B) \,, \qquad 
\widebar{H}_{\alpha}^\downarrow(A|B)_\rho := -\widebar{D}_{\alpha}(\rho_{AB}\|I_A \otimes \rho_B) \,.
\end{align}
Therefore, it is clear that the arrow up and down Petz and sandwiched conditional entropies are special cases of the arrow up and down $\alpha$-$z$ conditional entropies in~\eqref{alpha-z conditional}. Explicitly, 
\begin{align}
    &\widetilde{H}_{\alpha }^\uparrow(A|B)_{\rho}=H_{\alpha,\alpha}^\uparrow(A|B)_{\rho}, \qquad \widetilde{H}_{\alpha }^\downarrow(A|B)_{\rho}=H_{\alpha,\alpha}^\downarrow(A|B)_{\rho}, \\
    & \widebar{H}_{\alpha }^\uparrow(A|B)_{\rho}=H_{\alpha,1}^\uparrow(A|B)_{\rho}, \qquad \widebar{H}_{\alpha }^\downarrow(A|B)_{\rho}=H_{\alpha,1}^\downarrow(A|B)_{\rho}\,. 
\end{align}
We now discuss some limiting points of the $\alpha$-$z$ conditional entropies. 
We have the following limits
\begin{align}
\label{Fidelity}
& D_{\min}(\rho \| \sigma) := D_{1/2,1/2}(\rho\|\sigma) = -\log{F(\rho,\sigma)} \, , \\
& \widebar{D}_{0}(\rho \| \sigma) := \lim_{\alpha \rightarrow 0}D_{\alpha,1}(\rho\|\sigma)= -\log{\Trm(\Pi(\rho)\sigma)} \, , \\
\label{D}
&D(\rho \| \sigma) := \lim_{\alpha\rightarrow 1}D_{\alpha,z}(\rho\|\sigma) = \Trm[\rho (\log{\rho} - \log{\sigma})]\,, \,  \text{and}\\
\label{Dmax}
&D_{\max}(\rho \| \sigma) := \lim_{\alpha \rightarrow \infty} D_{\alpha,\alpha}(\rho\|\sigma) = \inf \{ \lambda \in \mathbb{R} : \rho \leq 2^{\lambda}\sigma \} \, ,
\end{align}
The quantity \( D_{\min} \) corresponds to the (log) fidelity, with \( F(\rho, \sigma) = \|\rho^{\frac{1}{2}} \sigma^{\frac{1}{2}}\|_1^2 \) denoting Uhlmann's fidelity. The Umegaki relative entropy \( D \) emerges as the limit \( \alpha \to 1 \) for any \( z \neq 0 \)~\cite{lin2015investigating}. Meanwhile, \( D_{\max} \) represents the \textit{max-relative entropy}~\cite{tomamichel16_book, Datta_rob2, Renner}.

Using these limits, the corresponding conditional entropies can be defined. Specifically, the von Neumann conditional entropy, the min-entropy~\cite{Renner}, and the max-entropy~\cite{konig2009operational, tomamichel16_book} are given as
\begin{align} 
&H(A|B)_{\rho}=-D(\rho_{AB}\|I_A \otimes \rho_B)=\sup_{\sigma_B}-D(\rho_{AB}\|I_A \otimes \sigma_B) &,\\
&H_{\min}(A|B)_{\rho} = \sup_{\sigma_B}-D_{\max}(\rho_{AB}\|I_A \otimes \sigma_B),\;\;\;\;H_{\max}(A|B)_{\rho} = \sup_{\sigma_B}-D_{\min}(\rho_{AB}\|I_A \otimes \sigma_B) \\
&H^{\downarrow}_{\min}(A|B)_{\rho} = -D_{\max}(\rho_{AB}\|I_A \otimes \rho_B),\;\; \quad \;\;\;\; H^{\uparrow}_{\max}(A|B)_{\rho} = \sup_{\sigma_B}-\widebar{D}_{0}(\rho_{AB}\|I_A \otimes \sigma_B)  \;  \,.
\end{align}

\begin{figure}
\centering
\resizebox{0.8\linewidth}{!}{%
\begin{forest}
    for tree={
        draw,
        rounded corners,
        align=center,
        parent anchor=south,
        child anchor=north,
        s sep=25mm,
        l=2cm
    }
    [$H_{\alpha,z}^{\lambda}$
        [$H_{\alpha,z}^{\downarrow}$, edge label={node[midway,left,xshift=-0pt,yshift=10pt]{\(\lambda = 0\)}}
            [$\widetilde{H}_{\alpha}^{\downarrow}$, edge label={node[midway,left,xshift=-0pt,yshift=10pt]{\(z = \alpha \)}}
                [ $H_{\min}^{\downarrow}$, edge label={node[midway,left,xshift=0pt,yshift=0pt]{\(\alpha = \infty\)}} ]
            ]
            [$\widebar{H}_{\alpha}^{\downarrow}$, edge label={node[midway,right,xshift=-0pt,yshift=10pt]{\(z = 1 \)}}]
        ]
       [ $ H $, edge label={node[midway,right,xshift=-7pt,yshift=-3pt]{%
    $\begin{array}{c} \textstyle \hskip-10pt\alpha=1 \\ \hskip2pt  (\lambda \neq -\infty) \end{array}$
}}, xshift=37.55mm]
        [ $H_{\alpha,z}^{\uparrow}$, edge label={node[midway,right,xshift=-0pt,yshift=10pt]{\(\lambda = 1\)}}
            [$\widetilde{H}_{\alpha}^{\uparrow}$, edge label={node[midway,left,xshift=-0pt,yshift=10pt]{\(z = \alpha \)}}
                [ $H_{\min}$, edge label={node[midway,left,xshift=-10pt,yshift=0pt]{\(\alpha = \infty\)}} ]
                [ $H_{\max}$, edge label={node[midway,right,xshift=10pt,yshift=0pt]{\(\alpha = \frac{1}{2}\)}} ]
            ]
            [$\widebar{H}_{\alpha}^{\uparrow}$, edge label={node[midway,right,xshift=-0pt,yshift=10pt]{\(z = 1 \)}}
                [ $H_{\max}^{\uparrow}$, edge label={node[midway,right,xshift=-0pt,yshift=0pt]{\(\alpha = 0\)}} ]
            ]
        ]
    ]  
\end{forest}
}
\caption{The diagram illustrates the connections between the newly introduced conditional entropy \(H^\lambda_{\alpha,z}\) and various established quantities explored in the literature.}
\label{tree of conditional entropies}
\end{figure}

We defined two min- and max-entropies, as various variants are sometimes introduced in the literature. We used the superscripts \(\downarrow\) and \(\uparrow\) in \(H^{\downarrow}_{\min}\) and \(H^{\uparrow}_{\max}\) to distinguish these quantities from the standard definitions of min- and max-conditional entropies. This notation reflects the fact that these entropies represent the minimal and maximal conditional entropies within our family of entropies (see Lemma~\ref{maximal and minimal} for further details).

These monotones just defined are known to coincide with the pointwise limits of the \(\alpha\)-\(z\) Rényi conditional entropies with up or down arrows. This means that for the up arrow case, the limit can be interchanged with the optimization over the states (see, e.g.,~\cite[Appendix A]{rubboli2022new}). Specifically, we have
 \begin{align}
    &H_{\min}(A|B)_{\rho} = \lim_{\alpha \rightarrow \infty}H^\uparrow_{\alpha,\alpha}(A|B)_{\rho}\,, \;\; H^{\uparrow}_{\max}(A|B)_{\rho} = \lim_{\alpha \rightarrow 0} H^\uparrow_{\alpha,1}(A|B)_{\rho}  \,, \;\;  H(A|B)_{\rho} = \lim_{\alpha \rightarrow 1} H^{\uparrow}_{\alpha,z}(A|B)_{\rho}  \,.
\end{align}
As stated at the beginning of the section, the \(\alpha\)-\(z\) conditional entropies with up and down arrows are specific instances of the conditional entropy \(H_{\alpha,z}^\lambda\), as defined in Definition~\ref{new entropy}. 
In particular, the down and up arrow \(\alpha\)-\(z\) conditional entropies in~\eqref{alpha-z conditional} can be recovered in the limits \(\lambda = 0\) and \(\lambda = 1\) of \(H_{\alpha,z}^\lambda\), respectively (see Proposition~\ref{limit lambda to zero and one} for a formal proof). Figure~\ref{tree of conditional entropies} summarizes all the dependencies in a diagram. Furthermore, the von Neumann entropy is obtained in the limit as \(\alpha \rightarrow 1\) for any finite values of \(\lambda\) (refer to Proposition~\ref{limit alpha to 1} for further details). We conclude this section by emphasizing that the limit \(\alpha \rightarrow 1\) not only encompasses the von Neumann conditional entropy but also extends to include more general quantities, such as the exponentiated mean of the von Neumann entropy of the conditional distribution (see Section~\ref{section limits}).

\section{Convexity properties of the trivariate trace functional}\label{Main technical results}

Let $\mathcal{H}_1$ and $\mathcal{H}_2$ be finite-dimensional Hilbert spaces. Let $A$ be a positive definite matrix on $\mathcal{H}_1 \otimes \mathcal{H}_2$, let $B$ be a positive definite matrix on $\mathcal{H}_2$, and $I$ the identity matrix on $\mathcal{H}_1$. Let $p,q,r,s$ be real parameters, and recall the functional $Q_{p,q,r,s}(A,B)$ defined by
\begin{equation}
    Q_{p,q,r,s}(A,B) = \opt_{C} \Phi_{p,q,r,s}(A,I \otimes B, I \otimes C),
\end{equation}
where $\Phi$ is defined as in \eqref{eq:trivariate} and the optimization runs over all density matrices $C$ on $\mathcal{H}_2$. For convenience, we always assume that $A, B, C$ are positive definite matrices and that $K$ is invertible, so that all negative powers are well-defined.

In this section, we establish our foundational results on the convexity properties of the function $(A,B) \mapsto Q_{p,q,r,s}(A,B)$. In particular, we show that, depending on the parameter regime, it is either convex, concave, or log-concave (Theorem \ref{log-concavity2}). Furthermore, we characterize the optimizers via a fixed-point equation that they must satisfy (Theorem \ref{fixed-point}).

As an application, these results will be used in subsequent sections to establish operational properties of certain entropic quantities, such as the data processing inequality and additivity.

The Schatten $p$-norm is defined as $\|A\|_p = \big(\Trm(AA^\dagger)^\frac{p}{2} \big)^\frac{1}{p}$ for $p \geq 1$.

\subsection{Joint convexity/concavity/log-concavity}

\begin{theorem}
\label{log-concavity2}
The function $(A,B) \mapsto Q_{p,q,r,s}(A,B)$ is 
\begin{enumerate}[itemsep=1pt,parsep=0pt]
    \item jointly concave if
    \begin{equation}
        0 < p,q,r \leq 1, \quad \text{and} \quad \frac{p}{p + q + r}, \frac{q+r}{p + q + r} \leq s \leq \frac{1}{p + q + r};
    \end{equation}
    \item  jointly log-concave if 
    \begin{align}
    &0< p,q \leq 1, \quad -1 \leq r<0, \quad \text{and} \quad  0<\frac{q + r}{p + q + r} \leq s \leq \frac{1}{p + q + r};
    \end{align}
    \item jointly convex if
    \begin{equation}
        1 \leq p \leq 2, \quad -1 \leq q,r < 0, \quad \text{and} \quad 0<\frac{1}{p+q+r} \leq  s \leq \frac{1}{1+q+r}.
    \end{equation}
\end{enumerate}
\end{theorem}

To prove this result, we need the following auxiliary lemma.

\begin{lemma}\label{primitive concavity2}
Let $K$ be an arbitrary matrix. Then,
\begin{equation}\label{primitive concavity2 function}
(A,B) \mapsto \textup{Tr}\left(B^\frac{q}{2} \textup{Tr}_1 \left(K^\dagger A^p K \right) B^\frac{q}{2}\right)^{s} 
\end{equation}

\begin{enumerate}[itemsep=2pt,parsep=0pt]
		\item is jointly concave if $0\leq q\leq p\leq1$ and $0<s\leq\frac{1}{p+q}$;
		\item is jointly convex if $-1\leq q\leq p\leq 0$ and $s>0$;
		\item is jointly convex if $-1\leq q\leq0,~1\leq p\leq 2,~(p,q)\ne(1,-1)$ and $s\geq\frac{1}{p+q}$.
\end{enumerate}

\end{lemma}
\begin{proof}
Let us define
\begin{equation}
X_{121'} = (K_{12} \otimes I_{1'}) \ket{\psi}_{1 1'},
\end{equation}
where $\ket{\psi}_{1 1' }$ denotes the (unnormalized) maximally entangled state, and $1'$ labels a Hilbert space $\mathcal{H}_{1'}$ that is isomorphic to $\mathcal{H}_1$. Using this auxiliary matrix, the function in \eqref{primitive concavity2 function} can be rewritten as
\begin{equation}
(A,B) \mapsto \Trm\left(B^{\frac{q}{2}} X^\dagger A^p X B^{\frac{q}{2}}\right)^s.
\end{equation}
The desired result now follows directly from~\cite[Theorem 1.1]{zhang20_alphaz}. Although that theorem is formulated for invertible square matrices, the argument extends to arbitrary matrices with compatible dimensions. Indeed, any rectangular matrix can be embedded into a square matrix by padding with zeros, and any matrix can be approximated arbitrarily well by invertible ones. Since convexity is preserved under such limits~\cite{zhang20_alphaz}, the conclusion remains valid in the general case.
\end{proof}

\begin{proof}[Proof of Theorem \ref{log-concavity2}]
Note first that
\begin{equation}
    Q(A,B) = \opt_{C}\Trm\left|A^{\frac{p}{2}} K \left(I \otimes B^{\frac{q}{2}} C^{\frac{r}{2}}\right)\right|^{2s},
\end{equation}
where we omit the indices of $Q_{p,q,r,s}$ for notational simplicity.
\smallskip

\noindent
(1) To show the first case, let $\mu := s(p+q+r) \in [\min\{p, q+r\}, 1]$. Note that in this case $\opt = \sup$, since, by Lemma \ref{primitive concavity/convexity}, $\Phi$ is concave in the parameter regime specified. We apply Lemma~\ref{variationalZhang} with
\begin{equation}
    (r_0, r_1, r_2) = \left(2s, \frac{2\mu}{p}, \frac{2\mu}{q+r}\right) \quad \textup{and} \quad (X, Y) = \left(A^{\frac{p}{2}} K, I \otimes B^{\frac{q}{2}} C^{\frac{r}{2}}\right),
\end{equation}
yielding
\begin{equation}\label{equ:proof of thm psi 1}
Q(A,B)= \sup_{C} \min_{Z \in \mathcal{L}(\mathcal{H})^{\times}} \left\{ \frac{p}{p+q+r} \Trm\left|A^{\frac{p}{2}} K Z\right|^{\frac{2\mu}{p}} + \frac{q+r}{p+q+r} \Trm\left|Z^{-1}(I \otimes B^{\frac{q}{2}} C^{\frac{r}{2}})\right|^{\frac{2\mu}{q+r}} \right\}.
\end{equation}

Setting $X = Z Z^\dagger$, we can replace the optimization over invertible matrices by optimization over positive definite matrices:
\begin{equation}
Q(A,B)= \sup_{C} \min_{X > 0} \left\{ \frac{p}{p+q+r} \Trm\left|A^{\frac{p}{2}} K X^{\frac{1}{2}}\right|^{\frac{2\mu}{p}} + \frac{q+r}{p+q+r} \Trm\left|X^{-\frac{1}{2}} (I \otimes B^{\frac{q}{2}} C^{\frac{r}{2}})\right|^{\frac{2\mu}{q+r}} \right\}.
\end{equation}

We may now invoke Sion's minimax theorem~\cite{sion1958general} to exchange the supremum and minimum. To apply this, we note that the set of quantum states is convex and compact, and the set of positive operators is convex, and that the functions
\begin{equation}
    X \mapsto \Trm(K^\dagger X K)^{\frac{\mu}{p}}, \quad X \mapsto \Trm(K^\dagger X^{-1} K)^{\frac{\mu}{q+r}}
\end{equation}
are convex, as $s \geq \frac{p}{p+q+r}$ and Lemma~\ref{primitive concavity/convexity} applies.

From~\cite[Proof of Theorem~3.7 given Lemma~3.1]{zhang20_alphaz}, for $0 \leq p \leq 1$ we have the variational identity
\begin{equation}
\label{eq:variational_trace_power}
\Trm(K^\dagger A K)^{\frac{1}{p}} = \max_{Z > 0} \left\{ \frac{1}{p} \Trm(K^\dagger A K Z^{1-p}) - \frac{1-p}{p} \Trm(Z) \right\}.
\end{equation}

Applying this to the second term in the variational form of $Q(A,B)$, for $s \geq \frac{q+r}{p+q+r}$ we obtain
\begin{equation}
\Trm\left|X^{-\frac{1}{2}} (I \otimes B^{\frac{q}{2}} C^{\frac{r}{2}})\right|^{\frac{2\mu}{q+r}} = \max_{Z > 0} \left\{ \frac{\mu}{q+r} \Trm\left( X^{-\frac{1}{2}} (I \otimes B^{\frac{q}{2}} C^r B^{\frac{q}{2}}) X^{-\frac{1}{2}} Z^{1 - \frac{q+r}{\mu}} \right) - \frac{\mu - q - r}{q+r} \Tr(Z) \right\}.
\end{equation}

We now solve the supremum over $C$ in the above expression using Lemma~\ref{optimization norm}, obtaining
\begin{equation}
\sup_{C} \Trm\left( X^{-\frac{1}{2}} \left(I \otimes B^{\frac{q}{2}} C^r B^{\frac{q}{2}}\right) X^{-\frac{1}{2}} Z^{1 - \frac{q+r}{\mu}} \right) 
= \left( \Trm  \left( B^{\frac{q}{2}} \textup{Tr}_1\left(X^{-\frac{1}{2}} Z^{1 - \frac{q+r}{\mu}} X^{-\frac{1}{2}} \right) B^{\frac{q}{2}} \right)^{\frac{1}{1-r}} \right)^{1 - r}.
\end{equation}

Putting all the pieces together, we arrive at
\begin{align}
&Q(A,B) = \min_{X > 0} \Bigg\{ \frac{p}{p+q+r} \Trm\left|A^{\frac{p}{2}} K X^{\frac{1}{2}}\right|^{\frac{2\mu}{p}}\\ 
&+\frac{1}{p+q+r} \max_{Z > 0} \Bigg\{ \mu \Bigg( \Trm \left| B^{\frac{q}{2}} \left( \textup{Tr}_1\left( X^{-\frac{1}{2}} Z^{1 - \frac{q+r}{\mu}} X^{-\frac{1}{2}} \right) \right)^{\frac{1}{2}} \right|^{\frac{2}{1 - r}} \Bigg)^{1 - r} - (\mu - q - r) \Tr(Z) \Bigg\} \Bigg\}.
\end{align}

To show that $Q(A,B)$ is jointly concave in $(A,B)$, we argue as follows: The first term inside the minimization is concave in $A$ by Lemma~\ref{primitive concavity/convexity}, since $0 < \frac{\mu}{p} \leq \frac{1}{p}$ implies that
\begin{equation}
    A \mapsto \Trm\left|A^{\frac{p}{2}} K X^{\frac{1}{2}}\right|^{\frac{2\mu}{p}}
\end{equation}
    is concave. The second term involves a supremum over $Z$ of a jointly concave function (by Lemma \ref{primitive concavity2}, which applies since $q + r \geq 0$), and the outer power $1 - r$ preserves concavity. Finally, taking the infimum over $X > 0$ of concave functions preserves concavity. Hence, $Q(A,B)$ is jointly concave in $(A,B)$.

\medskip
(2) In the second case, we consider the parameter regime where $r < 0$ but $q + r > 0$. To analyze this, we set $\mu := s(p + q + r) \in [ q + r, 1]$, and apply Lemma~\ref{variationalZhang} with 
\begin{equation}
(r_0, r_1, r_2) = \left(2s, \frac{2\mu}{p}, \frac{2\mu}{q + r}\right)
\quad \text{and} \quad 
(X, Y) = \left(A^{\frac{p}{2}} K, I \otimes B^{\frac{q}{2}} C^{\frac{r}{2}}\right),
\end{equation}
where we write $X = Z Z^\dagger$ to convert the optimization over invertible linear operators into an optimization over positive operators. This yields
\begin{equation}
Q(A,B) = \inf_C \min_{X > 0} \left\{
\frac{p}{p + q + r} \Trm\left|A^{\frac{p}{2}} K X^{\frac{1}{2}} \right|^{\frac{2\mu}{p}} 
+ \frac{q + r}{p + q + r} \Trm\left|X^{-\frac{1}{2}} (I \otimes B^{\frac{q}{2}} C^{\frac{r}{2}}) \right|^{\frac{2\mu}{q + r}} 
\right\}.
\end{equation}

To establish joint log-concavity of $Q(A,B)$ in $(A, B)$, we invoke Lemma~\ref{additive/multiplicative} with the same parameters $(r_0, r_1, r_2)$ as above, to obtain
\begin{align}
Q(A,B) &= \min_{X > 0} \bigg\{
\left(\Trm\left(X^{\frac{1}{2}} K^\dagger A^p K X^{\frac{1}{2}}\right)^{\mu/p}\right)^{\frac{p}{p + q + r}} \\
&\quad \times \inf_C \left(\Trm\left(X^{-\frac{1}{2}} (I \otimes B^{\frac{q}{2}} C^r B^{\frac{q}{2}}) X^{-\frac{1}{2}} \right)^{\mu / (q + r)}\right)^{\frac{q + r}{p + q + r}} \bigg\}.
\end{align}

We now apply the variational characterization of Schatten norms for positive operators. For $s \geq 1$, we have
\begin{equation}
\|A\|_s = \sup_{C} \Trm\left(A C^{1 - \frac{1}{s}}\right),
\end{equation}
which we use for the second term above. Specifically,
\begin{equation}
\Trm\left(X^{-\frac{1}{2}} (I \otimes B^{\frac{q}{2}} C^r B^{\frac{q}{2}}) X^{-\frac{1}{2}} \right)^{\frac{\mu}{q+r}}= 
\left( \sup_T \Trm\left( X^{-\frac{1}{2}} (I \otimes B^{\frac{q}{2}} C^r B^{\frac{q}{2}}) X^{-\frac{1}{2}} T^{1 - \frac{q + r}{\mu}} \right) \right)^{\frac{\mu}{q+r}},
\end{equation}
since $\mu / (q + r) \geq 1$ by assumption.

To proceed, we minimize over $C$ using Lemma~\ref{optimization norm}. By Sion's minimax theorem, the infimum and supremum can be exchanged; the function is convex in $C$ for $r \geq -1$ and concave in $T$ for $0 \leq (q + r)/\mu \leq 1$, and the set of quantum states is compact. This yields
\begin{equation}
\inf_C \Trm\left(X^{-\frac{1}{2}} (I \otimes B^{\frac{q}{2}} C^r B^{\frac{q}{2}}) X^{-\frac{1}{2}} T^{1 - \frac{q + r}{\mu}} \right)
= \left( \Trm\left( B^{\frac{q}{2}} \textup{Tr}_1\left(X^{-\frac{1}{2}} T^{1 - \frac{q + r}{\mu}} X^{-\frac{1}{2}} \right) B^{\frac{q}{2}} \right)^{\frac{1}{1 - r}} \right)^{1 - r}.
\end{equation}

Taking logarithms, we find
\begin{align}
\log Q(A,B) = &\inf_{X > 0} \bigg\{
\frac{p}{p + q + r} \log \Trm\left|A^{\frac{p}{2}} K X^{\frac{1}{2}} \right|^{\frac{2\mu}{p}} \\
&+ \frac{\mu(1 - r)}{p + q + r} \log \sup_T \Trm\left\{ \left( B^{\frac{q}{2}} \textup{Tr}_1\left(X^{-\frac{1}{2}} T^{1 - \frac{q + r}{\mu}} X^{-\frac{1}{2}} \right) B^{\frac{q}{2}} \right)^{\frac{1}{1 - r}} \right\} \bigg\}.
\end{align}
The concavity of $\log Q(A,B)$ now follows similarly to the previous case; the supremum over $T$ of concave functions is concave, the logarithm is concave, and the infimum over $X$ preserves concavity.

\medskip
(3) In the third case, we assume $q, r < 0$. To handle this case, we set $t^{-1} := s^{-1} - (q + r)$ and apply the second variational form of Lemma~\ref{variationalZhang} with
\begin{equation}
    (r_0, r_1, r_2) = \left(2t, 2s, \frac{-2}{q+r}\right) \quad \text{and} \quad (X, Y) = (A^{\frac{p}{2}} K, I \otimes B^{\frac{q}{2}} C^{\frac{r}{2}}),
\end{equation}
which yields
\begin{equation}
\label{eq:third_case_variational}
Q(A, B) = \inf_C \max_{X > 0} \left\{ \frac{s}{t} \Trm\left| A^{\frac{p}{2}} K X^{\frac{1}{2}} \right|^{2t} + s(q + r) \Trm\left| (I \otimes C^{-\frac{r}{2}} B^{-\frac{q}{2}}) X^{\frac{1}{2}} \right|^{\frac{-2}{q + r}} \right\}.
\end{equation}
We emphasize that the coefficient $s(q + r)$ is negative. Since $r < 0$, we may invoke Sion's minimax theorem~\cite{sion1958general} to exchange the infimum with the supremum. This is justified because, according to Lemma \ref{primitive concavity/convexity}, the functions
\begin{equation}
X \mapsto \Trm(K^\dagger X K)^t, \quad X \mapsto \Trm(K^\dagger X K)^{\frac{-1}{q+r}}
\end{equation}
are concave and convex, respectively, since $0 \leq t \leq 1$ (as $s \leq \frac{1}{1 + q + r}$) and $\frac{-1}{q + r} \geq 1$. The same lemma ensures that the function 
\begin{equation}
C \mapsto \Trm(K^\dagger C^{-r} K)^{\frac{-1}{q+r}}
\end{equation}
is concave as well, since $\frac{-1}{q + r} \leq \frac{-1}{r}$.

Applying \eqref{eq:variational_trace_power} to the second term in~\eqref{eq:third_case_variational}, we obtain
\begin{align}
\Trm\left|(I \otimes C^{-\frac{r}{2}} B^{-\frac{q}{2}}) X^{\frac{1}{2}} \right|^{\frac{-2}{q + r}}
&= \max_{Z > 0} \left\{ \frac{-1}{q + r} \Trm\left( X^{\frac{1}{2}} (I \otimes B^{-\frac{q}{2}} C^{-r} B^{-\frac{q}{2}}) X^{\frac{1}{2}} Z^{1 + q + r} \right) \right. \\
&\qquad \left. + \frac{1 + q + r}{q + r} \Trm(Z) \right\}.
\end{align}

To evaluate the supremum over $C$ in the first term of the above expression, we use Lemma~\ref{optimization norm}. Since $0 > r \geq -1$, we obtain
\begin{equation}
\sup_{C} \Trm\left( X^{\frac{1}{2}} (I \otimes B^{-\frac{q}{2}} C^{-r} B^{-\frac{q}{2}}) X^{\frac{1}{2}} Z^{1 + q + r} \right)
= \left( \Trm \left( B^{-\frac{q}{2}} \textup{Tr}_1\left( X^{\frac{1}{2}} Z^{1 + q + r} X^{\frac{1}{2}} \right) B^{-\frac{q}{2}} \right)^{\frac{1}{1 + r}} \right)^{1 + r}.
\end{equation}

Combining everything, we obtain the final variational expression
\begin{align}
Q(A, B) = &\max_{X > 0} \Bigg\{\frac{s}{t} \Trm\left| A^{\frac{p}{2}} K X^{\frac{1}{2}} \right|^{2t} \\
& - s \max_{Z > 0} \left\{ \left( \Trm \left( B^{-\frac{q}{2}} \textup{Tr}_1\left( X^{\frac{1}{2}} Z^{1 + q + r} X^{\frac{1}{2}} \right) B^{-\frac{q}{2}} \right)^{\frac{1}{1 + r}} \right)^{1 + r} - (1 + q + r)\Trm(Z) \right\} \Bigg\}.
\end{align}

To conclude convexity, note that by Lemma~\ref{primitive concavity/convexity}, the function
\begin{equation}
A \mapsto \Trm\left| A^{\frac{p}{2}} K X^\frac{1}{2}\right|^{2t}
\end{equation}
is convex, as $t \geq \frac{1}{p}$. The second term is concave in $B$, since it is the supremum over $Z$ of jointly concave functions (by Lemma \ref{primitive concavity2}), and the outer power $1 + r \in [0,1]$ preserves concavity. Finally, taking the supremum over $X$ of concave functions preserves convexity. Therefore, $Q(A, B)$ is jointly convex in $(A, B)$.
\end{proof}

\subsection{A characterization of optimizers}
\label{sec: characterization of optimizers}
In this section, we present the necessary and sufficient conditions that characterize the optimizer and show how these conditions naturally lead to a fixed-point equation.

Our first result is a set of necessary and sufficient conditions that the optimizer must satisfy, originally derived in~\cite{rubboli2025thesis}. These conditions generalize the approach adopted in~\cite[Theorem 4]{rubboli2022new}.
Before stating the following auxiliary lemma we need to introduce a piece of notation. We define the operator $\Xi_{s,r}(X,T)$ as
\begin{equation}
\label{equation problem}
\Xi_{s,r}(X,T) :=  
\begin{dcases} 
\chi_{s,r}(X,T) & \text{if } r = 1 \\
T^{-1} \chi_{s,r}(X,T) T^{-1} & \text{if } r = -1 \\
K_r \int_0^\infty (T + tI)^{-1} \chi_{s,r}(X,T) (T + tI)^{-1} t^r\, \mathrm{d}t & \text{if } r \neq \pm1
\end{dcases}
\end{equation}
with $K_r := \mathrm{sinc}(\pi r)$ and $\chi_{s,r}(X,T) := X^\dagger (X T^r X^\dagger)^{s-1} X$.
We have~\cite[Theorem 3.6]{rubboli2025thesis}
\begin{lemma}
\label{main Theorem}
Let $\mathcal{F}$ be a closed and convex subset of quantum states that contains a state with full support.
Let $X$ be an invertible matrix, and let $(s, r)$ satisfy one of the following conditions:
\begin{enumerate}
    \item $0 < r \leq 1$, $0 < s \leq r^{-1}$, $(r,s) \neq (1,1)$;
    \item $-1 \leq r \leq 0$, $s > 0$.
\end{enumerate}
Then, $S_* \in \mathcal{F}$ is an optimizer of the function $S \mapsto \Trm\big(X S^r X^\dagger\big)^s$ over $\mathcal{F}$ if and only if $S_* > 0$ and
\begin{equation}\label{optimality condition}
    \Tr\big(S\, \Xi_{s,r}(X,S_*)\big) \leq \Tr\big(X S_*^r X^\dagger\big)^s
\end{equation}
for all $S \in \mathcal{F}$.
\end{lemma}

\begin{remark}\label{remark on optimality}
    Note that \eqref{optimality condition} is indeed equivalent to the condition that the directional derivative of the function $S \mapsto \Trm\big(X S^r X^\dagger\big)^s$ at an optimizer $S_*$ is non-negative in every allowed direction. For further details, see \cite[Chapter 3]{rubboli2025thesis}.
\end{remark}

Throughout the paper, for a quantum state $\tau$ and a positive semidefinite matrix $A$, we write $\tau \propto A$ to denote $\tau = A / \Trm(A)$.

\begin{theorem}\label{fixed-point}
Let $A,B > 0$, and let $(s, r)$ satisfy the conditions of Lemma~\ref{main Theorem}. Then, $C_*$ is an optimizer of the function $C \mapsto \Phi(A,I \otimes B,I \otimes C)$ introduced in \eqref{eq:trivariate} if and only if $C_* > 0$ and
\begin{equation}\label{proportionality2}
	C_* \propto \Tr_1\left| A^{\frac{p}{2}} K \big(I \otimes B^{\frac{q}{2}} C_*^{\frac{r}{2}}\big) \right|^{2s}.
\end{equation}
\end{theorem}

\begin{proof}
Without loss of generality, we may assume that $p=q=2$. Moreover, we provide the proof for $K=I$ as the general case is similar.

\medskip
\textit{Sufficiency.} We use the optimizer characterization in the preceding lemma, noting that the optimization can be recast as one over the convex set of states of the form $\frac{I}{d} \otimes C$. According to Lemma~\ref{main Theorem}, it suffices to show that any $C_*$ satisfying~\eqref{proportionality2} also fulfills
\begin{equation}
\label{need to prove}
    \Tr\Big(\big(I \otimes S\big)\Xi_{s, r}\big(A(I \otimes B), I \otimes C_*\big)\Big) 
    \leq \Tr\left( A \big(I \otimes B C_*^r B \big) A \right)^s,
\end{equation}
for all density matrices $S$. We prove this for the case $r \neq \pm 1$, as the remaining cases are analogous.

To begin, let $X$, $Y$ and $S$ be arbitrary matrices with $S>0$. Then, for a function $f$ defined on the positive real numbers, we have $f(Y^\dagger Y)Y^\dagger = Y^\dagger f(Y Y^\dagger)$. Applying this identity with $Y:=XS^{\frac{r}{2}}$, we obtain 
\begin{equation}
\chi_{s, r}(X, S) = X^\dagger(X S^r X^\dagger)^{s-1} X 
= S^{-\frac{r}{2}} \left( S^{\frac{r}{2}} X^\dagger X S^{\frac{r}{2}} \right)^s S^{-\frac{r}{2}}.
\end{equation}
Invoking the preceding identity with $X:=A(I \otimes B)$ and $S:=I \otimes C_*$ in turn yields
\begin{equation}\label{rewrite Hermitian}
    \chi_{s, r}\big(A(I \otimes B), I \otimes C_*\big)=\left(I \otimes C_*^{-\frac{r}{2}}\right)\left| A(I \otimes B C_*^{\frac{r}{2}}) \right|^{2s} \left(I \otimes C_*^{-\frac{r}{2}}\right).
\end{equation}
Now we observe that
\begin{align}
   & \Tr\left((I \otimes S) \int_0^\infty (I \otimes C_* + t I)^{-1} 
   \chi_{s, r}\big(A(I \otimes B), I \otimes C_*\big) 
   (I \otimes C_* + t I)^{-1} t^r \, \mathrm{d}t \right) \\
   &\hspace{-.5cm}= \Tr\left( S \int_0^\infty (C_* + t I)^{-1} 
   \Tr_1\left( \chi_{s, r}(A(I \otimes B), I \otimes C_*) \right) 
   (C_* + t I)^{-1} t^r \, \mathrm{d}t \right) \\
   &\hspace{-.5cm}= \Tr\left( S \int_0^\infty (C_* + t I)^{-1} C_*^{-\frac{r}{2}} 
   \Tr_1\left| A(I \otimes B C_*^{\frac{r}{2}}) \right|^{2s} C_*^{-\frac{r}{2}} 
   (C_* + t I)^{-1} t^r \, \mathrm{d}t \right)
\end{align}
where in the last equality we used \eqref{rewrite Hermitian}.

By using~\eqref{proportionality2}, the integrand becomes diagonal in the same basis, allowing us to compute the integral explicitly by applying the identity
\begin{equation}\label{inversesquaredint}
 \int_0^\infty \frac{1}{(C_* + t)^2} t^r \, \mathrm{d}t 
 = \frac{C_*^{r - 1}}{K_r},
\end{equation}
where $K_r:=\mathrm{sinc}(\pi r)$. Putting everything together, we get
\begin{align}
    &\Tr\Big((I \otimes S)\, \Xi_{s, r}(A(I \otimes B), I \otimes C_*)\Big) \nonumber\\
    &\hspace{-.5cm}=\Tr\left(K_r S \int_0^\infty (C_* + t I)^{-1} C_*^{-\frac{r}{2}} 
    \Tr_1\left| A(I \otimes B C_*^{\frac{r}{2}}) \right|^{2s} C_*^{-\frac{r}{2}} 
    (C_* + t I)^{-1} t^r \, \mathrm{d}t \right) \nonumber\\
    &\hspace{-.5cm}=\Tr\left(K_r S C_*^{-\frac{r}{2}} 
    \Tr_1\left| A(I \otimes B C_*^{\frac{r}{2}}) \right|^{2s} C_*^{-\frac{r}{2}} \int_0^\infty (C_* + t I)^{-2} t^r \, \mathrm{d}t \right) \nonumber\\
    &\hspace{-.5cm}= \Tr(S) \cdot \Tr\left| A(I \otimes B C_*^{\frac{r}{2}}) \right|^{2s} \nonumber\\
    &\hspace{-.5cm}= \Tr\left( A (I \otimes B C_*^r B) A \right)^s,
\end{align}
where we used the fact that $S$ has unit trace. This shows~\eqref{need to prove} and the proof is complete.

\medskip
\textit{Necessity.} Suppose that $C_*$ is an optimizer. Invoking Lemma \ref{main Theorem}, we obtain
\begin{align}\label{vanishing derivative}
    \Tr &\left(K_r S \int_0^\infty (C_* + t I)^{-1} C_*^{-\frac{r}{2}} 
    \Tr_1\left| A(I \otimes B C_*^{\frac{r}{2}}) \right|^{2s} C_*^{-\frac{r}{2}} 
    (C_* + t I)^{-1} t^r \, \mathrm{d}t \right) \nonumber\\
    &=\Tr\Big((I \otimes S)\, \Xi_{s, r}(A(I \otimes B), I \otimes C_*)\Big) \nonumber\\
    &\leq\Tr\left( A (I \otimes B C_*^r B) A \right)^s.
\end{align}
Since $A$ is invertible, $C_*$ is also invertible by Lemma \ref{main Theorem}, and hence it lies in the interior of the set of density matrices. Hence, the derivative is well-defined in all directions and must vanish identically. This forces the preceding inequality to be an equality (see Remark \ref{remark on optimality}). Since $S$ is arbitrary, this in turn implies that the integral in \eqref{vanishing derivative} is proportional to the identity matrix. Since the identity does not have non-zero off-diagonal entries in any basis, by expanding the integral in terms of the eigenbasis of $C_*$, one can observe that $C_*$ and $C_*^{-\frac{r}{2}} \Tr_1\left| A(I \otimes B C_*^{\frac{r}{2}}) \right|^{2s} C_*^{-\frac{r}{2}}$ commute. This enables us to solve the integral to obtain
\begin{align}
    &K_r \int_0^\infty (C_* + t I)^{-1} C_*^{-\frac{r}{2}} 
    \Tr_1\left| A(I \otimes B C_*^{\frac{r}{2}}) \right|^{2s} C_*^{-\frac{r}{2}} 
    (C_* + t I)^{-1} t^r \, \mathrm{d}t \nonumber\\
    &=K_r C_*^{-\frac{r}{2}} 
    \Tr_1\left| A(I \otimes B C_*^{\frac{r}{2}}) \right|^{2s} C_*^{-\frac{r}{2}} \int_0^\infty (C_* + t I)^{-2} t^r \, \mathrm{d}t \nonumber\\
    &=\Tr_1\left| A(I \otimes B C_*^{\frac{r}{2}}) \right|^{2s} C_*^{-1},
\end{align}
where we used \eqref{inversesquaredint}. Consequently, we obtain
\begin{equation}
    \Tr\left( A (I \otimes B C_*^r B) A \right)^s C_* = \Tr_1\left| A(I \otimes B C_*^{\frac{r}{2}}) \right|^{2s},
\end{equation}
as desired.
\end{proof}

\begin{corollary}\label{Psi-Phi}
    Let $(r,s)$ be as above, and let $\mathcal{H}_1$ be trivial. Then,
    \begin{equation}
        \Tr \left(A^\frac{p}{2} K B^q K^\dagger A^\frac{p}{2}\right)^\frac{s}{1-rs} = \left(\opt_C \Tr \left(A^\frac{p}{2} K B^\frac{q}{2} C^r B^\frac{q}{2} K^\dagger A^\frac{p}{2}\right)^s\right)^\frac{1}{1-rs}.
    \end{equation}
\end{corollary}
\begin{proof}
Recall first that
\begin{equation}\label{eq:defQ}
    \Phi(A,B,C_*)=\Tr \left(A^\frac{p}{2} K B^\frac{q}{2} C_*^r B^\frac{q}{2} K^\dagger A^\frac{p}{2}\right)^s,
\end{equation}
and by Theorem \ref{fixed-point}
\begin{equation}
    \Phi(A,B,C_*) \cdot C_* = \left(C_*^\frac{r}{2} B^\frac{q}{2} K^\dagger A^p K B^\frac{q}{2} C_*^\frac{r}{2}\right)^s.
\end{equation}
By solving the preceding identity for $C_*$, we obtain
\begin{equation}
    C_* = \Phi(A,B,C_*)^{\frac{1}{rs-1}} \cdot \left(B^\frac{q}{2} K^\dagger A^p K B^\frac{q}{2}\right)^{\frac{s}{1-rs}}.
\end{equation}
Finally, substituting this last identity into the right-hand side of \eqref{eq:defQ}, we obtain
\begin{align}
    \Phi(A,B,C_*)^{\frac{1}{1-rs}} &= \Tr \left(A^\frac{p}{2} K B^\frac{q}{2} \left(B^\frac{q}{2} K^\dagger A^p K B^\frac{q}{2}\right)^{\frac{rs}{1-rs}} B^\frac{q}{2} K^\dagger A^\frac{p}{2}\right)^s\\
    &= \Tr \left(B^\frac{q}{2} K^\dagger A^p K B^\frac{q}{2} \left(B^\frac{q}{2} K^\dagger A^p K B^\frac{q}{2}\right)^{\frac{rs}{1-rs}}\right)^s\\
    &= \Tr \left(B^\frac{q}{2} K^\dagger A^p K B^\frac{q}{2}\right)^{\frac{s}{1-rs}},
\end{align}
as desired.
\end{proof}

\begin{remark}\label{Remark:PsiPhi}
    Using the reparameterization $t := \frac{s}{1-rs}$, we obtain the following variational expression for $\Psi$:
    \begin{equation}
        \Psi_{p,q,t}(A,B) = \opt_C \Phi_{p,q,\frac{1}{s}-\frac{1}{t},s}^{\frac{t}{s}}(A,B,C).
    \end{equation}
\end{remark}

\section{Additivity via tensorization of optimizers}
\label{sec:additivity}

In this section, we establish the additivity of $H_{\alpha,z}^\lambda$ with respect to tensor products, which will be utilized in the proof of the DPI in the next section.

To prove additivity, we leverage the characterization of the optimizers via the fixed-point equation derived in Section~\ref{sec: characterization of optimizers}. The parameter region specified in Lemma~\ref{main Theorem} contains the DPI region of the new conditional entropies in Definition~\ref{def:DPIregion} when setting \(s=z\) and \(r=\lambda(1-\alpha)/z\). Hence, the results derived in Section~\ref{sec: characterization of optimizers} for trace functionals, together with those obtained in this section, apply to $H_{\alpha,z}^\lambda$ in the DPI region.

\begin{theorem}[Additivity]\label{Additivity}
	Let $\rho_{AB}$ and $\sigma_{A'B'}$ be two quantum states, and let $s=z,r=\lambda(1-\alpha)/z$ satisfy the conditions of Lemma~\ref{main Theorem}. Then,
	\begin{equation}
		H_{\alpha,z}^\lambda(AA'|BB')_{\rho \otimes \sigma} = H^\lambda_{\alpha,z}(A|B)_{\rho} + H^\lambda_{\alpha,z}(A'|B')_{\sigma}.
	\end{equation}
\end{theorem}

To demonstrate this, we first derive a fixed-point equation that the optimizer must satisfy. From this condition, as shown below, we can directly conclude that the conditional entropy is additive.

\begin{proof}
It suffices to prove the result for full-rank states, since the general case follows from Proposition \ref{Continuity}. By an application of Theorem \ref{fixed-point}, $\tau_B$ is an optimizer of the function $\sigma_B\mapsto Q_{\alpha,z}^\lambda(\rho_{AB}|\sigma_B)$ if and only if
\begin{equation}\label{proportionality}
	\tau_B \propto \textup{Tr}_A\left|\rho_{AB}^{\frac{\alpha}{2z}} \rho_B^{\frac{(1-\lambda)(1-\alpha)}{2z}} \tau_B^{\frac{\lambda(1-\alpha)}{2z}}\right|^{2z}.
\end{equation}
From the definition of the conditional entropy, it is clear that it is enough to show that $\tau_{B} \otimes \gamma_{B'}$ is an optimizer of $\rho_{AB} \otimes \sigma_{A'B'}$, where $\tau_B$ and $\gamma_{B'}$ are the optimizers of $\rho_{AB}$ and $\sigma_{A'B'}$, respectively. Invoking \eqref{proportionality} for $\rho_{AB} \otimes \sigma_{A'B'}$ yields
\begin{equation}
	\eta_{BB'} \propto \textup{Tr}_{AA'}\left|(\rho_{AB} \otimes \sigma_{A'B'})^{\frac{\alpha}{2z}} (\rho_B \otimes \sigma_{B'})^{\frac{(1-\lambda)(1-\alpha)}{2z}} \eta_{BB'}^{\frac{\lambda(1-\alpha)}{2z}}\right|^{2z}.
\end{equation}
The latter condition is satisfied for $\eta_{BB'}=\tau_{B} \otimes \gamma_{B'}$ since, by substitution
\begin{equation}
	\tau_{B} \otimes \gamma_{B'} \propto \textup{Tr}_A \left|\rho_{AB}^{\frac{\alpha}{2z}} \rho_B^{\frac{(1-\lambda)(1-\alpha)}{2z}} \tau_B^{\frac{\lambda(1-\alpha)}{2z}} \right|^{2z} \otimes \textup{Tr}_{A'}\left|\sigma_{A'B'}^{\frac{\alpha}{2z}} \sigma_{B'}^{\frac{(1-\lambda)(1-\alpha)}{2z}} \gamma_{B'}^{\frac{\lambda(1-\alpha)}{2z}}\right|^{2z}.
\end{equation}
The latter proportionality relation holds since, by assumption, $\tau_B$ and $\gamma_{B'}$ are the optimizers of $\rho_{AB}$ and $\sigma_{A'B'}$ and satisfy the desired fixed-point equation.
\end{proof}

From the fixed-point condition~\eqref{proportionality}, we derive an explicit closed-form expression for $H^\lambda_{\alpha,1}$. Since all quantum divergences coincide in the classical case, this also yields a closed-form expression in that setting. The proof follows by observing that, in the classical case, the fixed-point equation~\eqref{proportionality} can be solved explicitly for $\tau_B$.

\begin{corollary}\label{closed-form Petz}
Let $\rho_{AB}$ be a quantum state, and let $s=1,r=\lambda(1-\alpha)$ satisfy the conditions of Lemma~\ref{main Theorem}. Then, the optimum of the optimization in the definition of $H^\lambda_{\alpha,1}(A|B)_{\rho}$ is attained at
\begin{equation}
	\tau_B \propto \left|\big(\Tr_A(\rho_{AB}^{\alpha})\big)^{\frac{1}{2}} \rho_B^{\frac{(1-\lambda)(1-\alpha)}{2}}\right|^{\frac{2}{1 - \lambda(1 - \alpha)}},
\end{equation}
and
\begin{equation}
	H^\lambda_{\alpha,1}(A|B)_{\rho} = \frac{1 - \lambda(1 - \alpha)}{1 - \alpha} \log{\Tr\left|\big(\Tr_A(\rho_{AB}^{\alpha})\big)^{\frac{1}{2}} \rho_B^{\frac{(1-\lambda)(1-\alpha)}{2}}\right|^{\frac{2}{1 - \lambda(1 - \alpha)}}}.
\end{equation}

\end{corollary}

\section{Data-processing inequality}
\label{sec:DPI}

\begin{definition} \label{def:DPIregion}
We define the region $\mathcal{D}:=\mathcal{D}_1\cup\mathcal{D}_2$, called the data-processing or DPI region, where $\mathcal{D}_1,\mathcal{D}_2$ are defined by
 \begin{align}
     \mathcal{D}_1 &:=\left\{(\alpha,z,\lambda) \in \mathbb{R}^3 : 0 < \alpha < 1, \ 1-z \leq \alpha \leq  z \ \text{and} \ 1-\frac{z}{1-\alpha} \leq \lambda \leq 1 \right\},\\
     \mathcal{D}_2 &:=\left\{(\alpha,z,\lambda) \in \mathbb{R}^3 : 1<\alpha<\infty, \ \alpha-1 \leq  z \leq \alpha \leq 2z \ \text{and} \ 1+\frac{z}{1-\alpha} \leq \lambda \leq 1 \right\}.
\end{align}
\end{definition}
This is the parameter region for which the data-processing inequality and several other properties of the conditional entropy can be established. The main goal of this section is to prove the DPI for $H^{\lambda}_{\alpha,z}$ whenever $(\alpha,z,\lambda)$ lies in $\D$, thereby justifying this terminology.

When studying conditional entropies, and hence measures of conditional uncertainty, it is natural to impose axioms motivated by operational principles. A central requirement is an appropriate data-processing inequality for quantum conditional entropies.
In the classical setting, i.e., for probability distributions, unconditional entropies are expected to be monotone under doubly stochastic maps~\cite[Section~2.B]{gour21_axiomatic}. By Birkhoff's theorem, every doubly stochastic map can be expressed as a convex combination of permutations. In the quantum setting, this notion is naturally extended by requiring monotonicity under convex combinations of unitary channels, namely maps of the form
\begin{equation}
\mathcal{M}(\cdot) = \sum_k p_k\, U_k (\cdot)\, U_k^\dagger,
\end{equation}
where $U_k$ are unitary operators and $\{p_k\}$ is a probability distribution. We refer to such maps as mixed unitary channels.

In the conditional scenario, however, additional requirements arise. Processing the side information system~$B$ alone should not reduce the uncertainty about~$A$. Indeed, by the non-signalling principle, local operations on~$B$ cannot increase correlations with~$A$, and a meaningful conditional entropy should reflect this property. Moreover, one must consider a larger class of operations. Suppose an observer performs a measurement on the side information system~$B$ and, conditioned on the measurement outcome, applies to system~$A$ a channel given by a convex combination of unitaries. Such conditioned operations allow information obtained from~$B$ to influence the processing of~$A$, while still prohibiting signalling from~$A$ to~$B$. Since no additional information about~$A$ is revealed in this procedure, one expects that the conditional entropy should not decrease under these channels. These operations correspond to bipartite channels of the form
\begin{align}
    \mathcal{E}_{AB \to AB'}(\cdot) 
    = \sum_i 
    \mathcal{M}_{A \to A}^{i} 
    \otimes 
    \mathcal{N}_{B \to B'}^{i}(\cdot),
\end{align}
where each $\mathcal{M}_{A \to A}^{i}$ is a convex combination of unitaries acting on $A$, and $\{\mathcal{N}_{B \to B'}^{i}\}_i$ is an instrument acting on~$B$.

These conditions may be viewed as minimal data-processing requirements. A natural question, however, is whether monotonicity continues to hold for a broader class of channels. In particular, the operations on~$A$ may be taken to be unital or sub-unital, as considered in~\cite{vempati2022conditional,gour2024inevitability}. This class is strictly larger, since the set of unital channels strictly contains the set of mixed unitary channels (see e.g.,~\cite[Problem 8.3]{nielsen2001quantum}).

Below, we establish this stronger form of the DPI. Finally, we note that monotonicity can also be studied under the even larger class of conditionally mixing channels introduced in~\cite{gour2024inevitability} (see Section~\ref{alternativeproof} for more details).

We first recall that a quantum instrument is a collection of completely positive trace-nonincreasing operations $\{\N^i\}_i$ such that the overall transformation $\sum_i \N^i$ is trace-preserving. We aim to prove

\begin{theorem}[Data-processing inequality] \label{DPItheorem}
Let $(\alpha, z, \lambda) \in \mathcal{D}$, and let $\rho_{AB}$ be a quantum state. Furthermore, let $\sigma_{A'B'}$ be defined as
\begin{equation}
    \sigma_{A'B'} = \sum_i \mathcal{M}^i_{A\rightarrow A'} \otimes \mathcal{N}^i_{B\rightarrow B'}(\rho_{AB})\,,
\end{equation}
where each $\mathcal{M}^i_{A\rightarrow A'}$ is a sub-unital channel and $\{\mathcal{N}^i_{B\rightarrow B'}\}_i$ is a quantum instrument. Then, the following inequality holds
\begin{equation}
    H^{\lambda}_{\alpha, z}(A|B)_{\rho} \leq H^{\lambda}_{\alpha, z}(A'|B')_{\sigma}\,.
\end{equation}
\end{theorem}

To prove this result, we first establish the data processing inequality separately for two types of operations: a sub-unital channel acting on Alice's subsystem and an arbitrary quantum channel acting on Bob's subsystem. The monotonicity under sub-unital channels on Alice's side is presented in Proposition~\ref{DPItheoremAlice}, while the corresponding result for general maps on Bob's side is stated in Proposition~\ref{DPItheoremBob}. Finally, in Section~\ref{conditionally subunital}, we combine these two ingredients to prove the full monotonicity under the most general class of maps involving conditioned operations.

We begin by proving monotonicity under channels acting on Bob's side. The key idea is to apply Theorem~\ref{log-concavity2} to observe that the quantity $Q_{\alpha,z}^\lambda(A|B)_{\rho} = \exp((1-\alpha) H^\lambda_{\alpha,z}(A|B)_{\rho})$ is convex, concave, or log-concave, depending on the parameter regime (see Proposition \ref{log-concavity}). We then show that these concavity/convexity properties imply the data processing inequality under the action of a channel on the side information (Subsection~\ref{Convexity implies DPI}). The monotonicity under the action of a sub-unital channel on Alice's system is easier to prove and follows directly from an expression of the conditional entropy in terms of the $\alpha$-$z$ R\'enyi relative entropy (Subsection \ref{SecDPItheoremAlice}).

\begin{remark}\label{counterexample}
It is natural to ask whether the DPI for $H^{\lambda}_{\alpha,z}$ can be derived directly from the DPI for an underlying ``trivariate divergence." However, as we demonstrate below, this is not the case.

In~\cite[Corollary 3.3]{carlen2016some}, the authors show that, for nonzero $r$, $s$ and $t$, the function
\begin{equation}\label{TraceFunctional}
(A,B,C) \mapsto \Tr\left(A^{\frac{r}{2}} B^{\frac{s}{2}} C^t B^{\frac{s}{2}} A^{\frac{r}{2}}\right)
\end{equation}
is never concave, and is convex only under the conditions $s = 2$, $t, r < 0$, and $-1 \leq t + r < 0$. Consequently, for example, the trace functional in \eqref{TraceFunctional} is not convex for $r = 2$, $s = t = -\frac{1}{2}$. It is well-known that the failure of convexity in such functionals implies a violation of the DPI. Yet, the conditional entropy $H^{1/2}_{2,1}$, which is associated with this functional, is the dual of $H^\downarrow_{1/3,2/3}$ (see Theorem \ref{Duality}). Since $H^\downarrow_{1/3,2/3}$ satisfies the DPI---by virtue of the DPI for the underlying $\alpha$-$z$ R\'enyi divergence---it follows from standard duality arguments (see, e.g.,~\cite[Section III]{tomamichel13_duality}) that $H^{1/2}_{2,1}$ must also satisfy the DPI.
\end{remark}

\begin{proposition}\label{log-concavity}
    Let $(\alpha,z,\lambda) \in \D$. Then,

    \medskip
    (i) $\rho \rightarrow Q_{\alpha,z}^\lambda(A|B)_{\rho}$ is concave for $0 < \alpha < 1, 0 < \lambda < 1$;
    
    \medskip
    (ii) $\rho \rightarrow Q_{\alpha,z}^\lambda(A|B)_{\rho}$ is convex for $\alpha > 1, 0 < \lambda < 1$;
    
    \medskip
    (iii) $\rho \rightarrow Q_{\alpha,z}^\lambda(A|B)_{\rho}$ is log-concave for $0 < \alpha < 1, 1-\frac{z}{1-\alpha}\leq \lambda < 0$.
\end{proposition}
\begin{proof}
    Invoking Theorem \ref{log-concavity2} with appropriate choices of the parameters, namely, $p=\frac{\alpha}{z}$, $q=\frac{(1-\lambda)(1-\alpha)}{z}$, $r=\frac{\lambda (1-\alpha)}{z}$ and $s=z$, yields the desired result.
\end{proof}

As shown below, the concavity and convexity results stated in Proposition~\ref{log-concavity}, together with the previously established additivity, are sufficient to prove the DPI.

\subsection{DPI under channels on Bob's system}\label{Convexity implies DPI}

Here, we demonstrate that the conditional entropy satisfies the data processing inequality under the action of a quantum channel on Bob's side, utilizing the convexity and concavity properties derived previously. Indeed, the argument follows a standard approach, showing that a concave function $f$, invariant under unitary operations and addition of a system, is monotone increasing under the action of a quantum channel (see e.g.~\cite[Proposition 4.5]{tomamichel16_book}).

\begin{proposition}\label{DPItheoremBob}
Let $(\alpha,z,\lambda) \in \D$ and $\rho_{AB}$ be a quantum state. Then, for any quantum channel $\N_{B\rightarrow B'}$ with $\sigma_{AB'}= \N (\rho_{AB})$,  we have
	\begin{equation}
		H^{\lambda}_{\alpha,z}(A|B)_{\rho}\leq H^{\lambda}_{\alpha,z}(A|B')_{\sigma},
	\end{equation}
\end{proposition}
To treat all parameter regimes with a unified argument, we introduce the following notation:
\begin{align}
    f(\rho_{AB}) :=
\begin{cases}
	Q_{\alpha,z}^\lambda(A|B)_{\rho}, & \text{if } \alpha < 1 \text{ and } 0 < \lambda< 1, \vspace{5pt} \\
	\log{Q_{\alpha,z}^\lambda(A|B)_{\rho}}, & \text{if } \alpha < 1 \text{ and } \lambda < 0, \vspace{5pt} \\
	-Q_{\alpha,z}^\lambda(A|B)_{\rho}, & \text{if } \alpha > 1 \text{ and } 0 < \lambda < 1.
\end{cases}
\end{align}
According to Proposition~\ref{log-concavity}, the function $\rho_{AB} \mapsto f(\rho_{AB})$ is concave in the specified range of parameters.
We begin by demonstrating that $f$ is invariant under the addition of an auxiliary system on Bob's side. Note that when we write $f(\rho_{AB} \otimes \gamma_{B'})$ below, it is implicitly understood that the conditioning systems are $B$ and $B'$.
\begin{lemma}
\label{lem: invariance additional conditioning}
    For any bipartite state \(\rho_{AB}\) and an additional conditioning system \(\gamma_{B'}\), we have
$
f(\rho_{AB} \otimes \gamma_{B'}) = f(\rho_{AB}) \,.
$
\end{lemma} 
\begin{proof}
    This property directly follows from the additivity of \(H_{\alpha,z}^\lambda\) (as stated in Theorem \ref{Additivity}) and leads to the multiplicativity of \(Q_{\alpha,z}^\lambda\) for states of the form \(\rho_{AB} \otimes 1_{A'} \otimes \gamma_{B'}\), where \(A'\) represents a trivial system. Additionally, we have
$
Q_{\alpha,z}^\lambda(A'|B')_{1_{A'} \otimes \gamma_{B'}} = 1.
$
\end{proof}
We now address the invariance under unitary transformations on Bob's system.
\begin{lemma}
    For \(\rho_{AB}\) and a unitary operator \(U_B\) acting on Bob's system, we have
$
f(\rho_{AB}) = f(\rho'_{AB})$,
where \(\rho'_{AB} = U_B \rho_{AB} U_B^\dagger\).
\end{lemma}
\begin{proof}
We only need to prove the invariance of $Q_{\alpha,z}^\lambda$.
 To see this, note that
\begin{align}
Q_{\alpha,z}^\lambda(A|B)_{\rho'} &= \opt_{\sigma_B} Q_{\alpha,z}^\lambda(\rho'_{AB} | \sigma_B ) \\
&= \opt_{\sigma_B} Q_{\alpha,z}^\lambda(U_B \rho_{AB} U_B^\dagger| \sigma_B ) \\
&= \opt_{\sigma_B} Q_{\alpha,z}^\lambda(U_B \rho_{AB} U_B^\dagger| U_B \sigma_B U_B^\dagger)\\
&= \opt_{\sigma_B} 
\Trm\Bigg|U_B\rho_{AB}^{\frac{\alpha}{2z}} U_B^\dagger U_B \rho_B^{\frac{(1-\lambda)(1-\alpha)}{2z}} U_B^\dagger U_B \sigma_B^{\frac{\lambda(1-\alpha)}{z}} U_B^\dagger \Bigg|^{2z} \\
&= \opt_{\sigma_B} Q_{\alpha,z}^\lambda(\rho_{AB} | \sigma_B )\\
&= Q_{\alpha,z}^\lambda(A|B)_{\rho},
\end{align}
where in the second line we used that since the optimization is performed over all states, the unitary acting on $\sigma_B$ does not change the value.
In the last line, we used that $U_B U_B^\dagger = I$ and that the norm is unitary invariant.    
\end{proof}

Now that we have all the necessary ingredients, we proceed to prove Proposition \ref{DPItheoremBob}.

\begin{proof}[Proof of Proposition \ref{DPItheoremBob}]
Since by Lemma~\ref{lem: invariance additional conditioning} the conditional entropy remains invariant under the addition of a pure state on Bob's side, we can assume that the dimensions of the input and output states coincide.
Let \(\N_{B\rightarrow B}\) be a quantum channel acting on Bob's system. We can express \(\N_{B\rightarrow B}\) as
\begin{align}
    \N_{B\rightarrow B}(\rho_{AB}) = \Tr_{B'}\big(U_{BB'}(\rho_{AB} \otimes \gamma_{B'}) U^\dagger_{BB'}\big) \,.
\end{align}
Here, \(U_{BB'}\) is unitary. Furthermore, we can employ a unitary design (e.g., the Heisenberg-Weyl operators) such that
\begin{equation}
\frac{1}{d_B^2}\sum_{i=1}^{d_B^2} \big(I_A \otimes V^i_B\big) \tau_{AB} \big(I_A \otimes V_B^{i,\dagger}\big) = \tau_A \otimes \pi_B,
\end{equation}
where $\pi_B=I_B/d_B$ is the maximally mixed state. Thus, we obtain
\begin{align}
f\big(\N_{B\rightarrow B}(\rho_{AB})\big) &= f\big(\N_{B\rightarrow B}\big(\rho_{AB}\big) \otimes \pi_{B'}\big) \\
&= f \left(\frac{1}{d_{B'}^2}\sum_{i=1}^{d_{B' }^2} \big(I_A \otimes V^i_{B'}\big) U_{BB'}\big(\rho_{AB}\otimes \gamma_{B'}\big) U^\dagger_{BB'} (I_A \otimes V_{B'}^{i,\dagger}\big)\right)\\
\label{concavity inequality DPI}
&\geq \sum_{i=1}^{d_{B' }^2} \frac{1}{d_{B'}^2} f \left( \big(I_A \otimes V^i_{B'}\big) U_{BB'}\big(\rho_{AB}\otimes \gamma_{B'}\big) U^\dagger_{BB'} \big(I_A \otimes V_{B'}^{i,\dagger}\big)\right) \\
&=  f \left( \rho_{AB}\otimes \gamma_{B'}\right) \\
&= f \left(\rho_{AB}\right),
\end{align}
where we have repeatedly utilized the additivity and invariance properties with respect to unitaries acting on Bob's side and~\eqref{concavity inequality DPI} follows from Proposition~\ref{log-concavity}. This completes the proof for the regions specified above for $f$. For the remaining region where \(\alpha > 1\) and \(1 - \frac{z}{\alpha - 1} \leq \lambda < 0\), the DPI can be proven by leveraging the duality relations in Theorem \ref{Duality}, and the fact that if a conditional entropy satisfies the DPI under channels acting only on Bob's system, then its dual also satisfies the DPI under those channels.  (For a proof of this latter fact, refer to~\cite[Section III]{tomamichel13_duality}, and see Tables~\ref{table I} and~\ref{table II} for a detailed discussion of the duality transformations concerning the DPI regions of the parameters). The endpoint cases for the parameters follow by a limiting argument; see Section~\ref{section limits}.
\end{proof}

\subsection{DPI under sub-unital channels on Alice's system}\label{SecDPItheoremAlice}

\noindent\textbf{Notation.} For $\rho,\sigma \in \mathcal{P}(A)$ and parameters $\lambda,q$, we define
\begin{equation}
\mathcal{G}_{\lambda,q}(\rho,\sigma)
:= \rho^{\frac{(1-\lambda)q}{2}}\sigma^{\lambda q}\rho^{\frac{(1-\lambda)q}{2}}.
\end{equation}

We aim to prove the following
\begin{proposition}\label{DPItheoremAlice}
Let $(\alpha,z,\lambda) \in \D$ and $\rho_{AB}$ be a quantum state. Then, for any sub-unital channel $\M_{A\rightarrow A'}$ with $\sigma_{A'B}= \M (\rho_{AB})$,  we have
	\begin{equation}
		H^{\lambda}_{\alpha,z}(A|B)_{\rho}\leq H^{\lambda}_{\alpha,z}(A'|B)_{\sigma},
	\end{equation}
\end{proposition}

In this case, the monotonicity proof is straightforward and directly follows from the fact that we can rewrite the conditional entropy in terms of the underlying $\alpha$-$z$ R\'enyi relative entropy.

\begin{proof}[Proof of Proposition \ref{DPItheoremAlice}]
Observe that the conditional entropy can be written as 
\begin{align}
    H^\lambda_{\alpha,z}(A|B)_{\rho} &= \opt_{\sigma_B} -D_{\alpha,z}\left(\rho_{AB} \middle\| I_A \otimes \mathcal{G}_{\lambda,\frac{1-\alpha}{z}}(\rho_B, \sigma_B)^{\frac{z}{1-\alpha}}\right) \,,
\end{align}
where the $\alpha$-$z$ R\'enyi relative entropy is defined in Definition~\ref{alpha-z relative entropy}.
Therefore, the DPI of the conditional entropy under sub-unital channels on Alice's side follows from the DPI of the $\alpha$-$z$ R\'enyi relative entropies established in~\cite{zhang20_alphaz}. Namely, we obtain
	\begin{align}
		-H^\lambda_{\alpha,z}(A|B)_{\rho} &= \textup{opt}_{\sigma_B} D_{\alpha,z}\left(\rho_{AB} \middle\| I_A \otimes \mathcal{G}_{\lambda,\frac{1-\alpha}{z}}(\rho_B, \sigma_B)^{\frac{z}{1-\alpha}}\right) \\
		&\geq \textup{opt}_{\sigma_B} D_{\alpha,z}\left(\M(\rho_{AB}) \middle\|\M(I_A) \otimes \mathcal{G}_{\lambda,\frac{1-\alpha}{z}}(\rho_B, \sigma_B)^{\frac{z}{1-\alpha}}\right) \\
		&\geq \textup{opt}_{\sigma_B} D_{\alpha,z}\left(\M(\rho_{AB}) \middle\| I_{A'} \otimes \mathcal{G}_{\lambda,\frac{1-\alpha}{z}}(\rho_B, \sigma_B)^{\frac{z}{1-\alpha}}\right) \\
		&= -H^\lambda_{\alpha,z}(A'|B)_{\sigma},
	\end{align}
Here, the second inequality follows from antimonotonicity with respect to the second argument (see e.g.~\cite[Lemma 17]{rubboli2022new}), together with \(\M(I_A) \leq I_{A'}\). Moreover, since the inequalities hold for any \(\sigma_B\), it is clear that taking inf/sup on both sides of each inequality does not change their direction.
\end{proof}

\subsection{Proof of Theorem~\ref{DPItheorem}}
\label{conditionally subunital}
We show that, due to additional structural properties of the function $(Q_{\alpha,z}^\lambda)^{\frac{1}{1 - \lambda(1 - \alpha)}}$, the monotonicity established earlier under tensor product channels extends to the more general setting of conditioned operations.

In particular, consider a function $f$ that maps positive semidefinite matrices to non-negative real numbers, which we will later specify as the optimized $(Q_{\alpha,z}^\lambda)^{\frac{1}{1 - \lambda(1 - \alpha)}}$. If $f$ is monotone under tensor product channels and satisfies the following two properties
\begin{enumerate}
    \item Positive homogeneity: $f(a X) = a f(X)$ for all $X \geq 0$ and all scalars $a \geq 0$;
    \item Additivity under direct sums: $f\big( \sum_i X_i \otimes \ketbra{i}{i}\big)=\sum_i f(X_i)$ for positive semidefinite $X_i$;
\end{enumerate}
then
$f$ is also monotone under conditioned operations.
\begin{lemma}
Let $f$ be a function from the set of positive semidefinite operators on a bipartite Hilbert space to the positive real numbers. We assume that $f$ is positively homogeneous and additive under direct sums.
If $f$ is monotone under tensor product channels, that is,
\begin{equation}
f\left( \mathcal{M}_{A\rightarrow A'} \otimes \mathcal{N}_{B\rightarrow B'}(\rho_{AB}) \right) \leq f(\rho_{AB}),
\end{equation}
for any sub-unital channel $\mathcal{M}_{A\rightarrow A'}$ and arbitrary quantum channel $\mathcal{N}_{B\rightarrow B'}$, then $f$ is also monotone under conditioned operations
\begin{equation}
f\left( \sum_i \mathcal{M}_{A \rightarrow A'}^i \otimes \mathcal{N}_{B\rightarrow B'}^i(\rho_{AB}) \right) \leq f(\rho_{AB}),
\end{equation}
where each $\mathcal{M}_{A\rightarrow A'}^i$ is sub-unital, and $\{\mathcal{N}_{B\rightarrow B'}^i\}_i$ is a quantum instrument.
\end{lemma}

\begin{proof}
We begin by applying the data processing inequality under partial trace over an auxiliary register $B'$:
\begin{align}
    f\left( \sum_i \mathcal{M}^i_{A\rightarrow A'} \otimes \mathcal{N}^i_{B\rightarrow B'}(\rho_{AB}) \right) \leq f\left( \sum_i \mathcal{M}^i_{A\rightarrow A'} \otimes \mathcal{N}^i_{B\rightarrow B'}(\rho_{AB}) \otimes \ketbra{i}{i}_{B''} \right).
\end{align}
Next, we use the fact that $f$ is additive under direct sums and positively homogeneous, which gives
\begin{align}
    &f\left( \sum_i \mathcal{M}^i_{A\rightarrow A'} \otimes \mathcal{N}^i_{B\rightarrow B'}(\rho_{AB}) \otimes \ketbra{i}{i}_{B''} \right) \\
    &\qquad \qquad \qquad = \sum_i f\left( \mathcal{M}^i_{A\rightarrow A'} \otimes \mathcal{N}^i_{B\rightarrow B'}(\rho_{AB}) \right) \\
    &\qquad \qquad \qquad= \sum_i \Tr\left[ (I_{A} \otimes \mathcal{N}^i_{B\rightarrow B'})(\rho_{AB}) \right] \cdot f\left( \frac{ \mathcal{M}^i_{A\rightarrow A'} \otimes \mathcal{N}^i_{B\rightarrow B'}(\rho_{AB}) }{ \Tr\left[ (I_A \otimes \mathcal{N}^i_{B\rightarrow B'})(\rho_{AB}) \right] } \right).
\end{align}
Define the normalized state:
\begin{equation}
    \widetilde{\rho}_{AB'}^{\,i} = \frac{ (\mathbb{I}_A \otimes \mathcal{N}^i_{B\rightarrow B'})(\rho_{AB}) }{ \Tr\left[ (\mathbb{I}_A \otimes \mathcal{N}^i_{B\rightarrow B'})(\rho_{AB}) \right] }.
\end{equation}
Applying the DPI to $\widetilde{\rho}_{AB'}^{\,i}$ under the sub-unital channel $\mathcal{M}^i_{A\rightarrow A'} \otimes \mathbb{I}_{B'}$ yields:
\begin{equation}
    f\left( (\mathcal{M}^i_{A\rightarrow A'} \otimes \mathbb{I}_{B'}) \, ( \widetilde{\rho}_{AB'}^{\,i} ) \right) \leq f\left( \widetilde{\rho}_{AB'}^{\,i} \right).
\end{equation}
Combining these observations, we obtain:
\begin{align}
   f\left( \sum_i \mathcal{M}^i_{A\rightarrow A'} \otimes \mathcal{N}^i_{B\rightarrow B'}(\rho_{AB}) \right)
    &\leq \sum_i \Tr\left[ (\mathbb{I}_A \otimes \mathcal{N}^i_{B\rightarrow B'})(\rho_{AB}) \right] f\left( \widetilde{\rho}_{AB'}^{\,i} \right) \\
    &= \sum_i f\left( (\mathbb{I}_A \otimes \mathcal{N}^i_{B\rightarrow B'})(\rho_{AB}) \right) \\
    &= f\left(  \sum_i \mathcal{N}^i_{B\rightarrow B'}(\rho_{AB}) \otimes \ketbra{i}{i}_{B''} \right).
\end{align}

Finally, applying the DPI once more, noting that $\sum_i \mathcal{N}^i_{B\rightarrow B'}(\cdot) \otimes \ketbra{i}{i}_{B''}$ defines a quantum channel that records the measurement outcome in a register, we have:
\begin{equation}
    f\left( \sum_i \mathcal{N}^i_{B\rightarrow B'}(\rho_{AB}) \otimes \ketbra{i}{i}_{B''} \right) \leq f(\rho_{AB}).
\end{equation}
Thus, the desired monotonicity under conditioned operations follows.
\end{proof}
Next, let us show that the optimized  $(Q_{\alpha,z}^\lambda)^{\frac{1}{1 - \lambda(1 - \alpha)}}$ satisfies the aforementioned properties.
Let us set $f$ equal to the optimized $(Q_{\alpha,z}^\lambda)^{\frac{1}{1 - \lambda(1 - \alpha)}}$. The positive homogeneity of the function follows directly from its definition. It remains to show that the function is additive under direct sums. A straightforward computation shows that for positive semidefinite $X_i$ (see Proposition~\ref{classical conditioning} for details):
\begin{align}
\opt_{\sigma_{BB'}} Q_{\alpha,z}^\lambda \Big( \sum_{i} X^i_{AB} \otimes \ketbra{i}{i}_{B'} \Big| \sigma_{BB'} \Big) = \left( \sum_i \Big( \opt_{\sigma_B} Q_{\alpha,z}^\lambda \left( X^i_{AB} \big| \sigma_B \right) \Big)^{\frac{1}{1 - \lambda(1 - \alpha)}} \right)^{1 - \lambda(1 - \alpha)}.
\end{align}
Hence, the optimized $(Q_{\alpha,z}^\lambda)^{\frac{1}{1 - \lambda(1 - \alpha)}}$ satisfies the assumptions of the previous lemma. Since $H_{\alpha,z}^\lambda$ differs from this expression only by multiplicative constants and a logarithm, Theorem~\ref{DPItheorem} follows for $H_{\alpha,z}^\lambda$ as well.

\subsection{DPI under conditionally mixing channels}\label{alternativeproof}
In this section, we establish that our new conditional entropies satisfy the data-processing inequality for a closely related class of channels, known as conditionally mixing channels, for a restricted set of parameters. These channels were introduced in~\cite[Section 3]{gour2024inevitability}, where the authors also define conditional entropies in~\cite[Definition 4]{gour2024inevitability} as quantities that are monotone under these channels, in addition to satisfying other desirable properties such as additivity, invariance under isometries, and normalization. We will show that $H_{\alpha,z}^\lambda$ fulfills all of these axioms, at least when $0 \leq \lambda \leq 1$.  As we have already argued that our quantities cannot be derived in canonical forms from an underlying divergence, this shows that not all conditional entropies admit such a representation, thereby answering a question posed in~\cite{gour2024inevitability}.

The remainder of this section is structured as follows: we first introduce conditionally mixing channels; next, we establish monotonicity under these channels for $0 \leq \lambda \leq 1$; and finally, we show that the properties of additivity, invariance under isometries, and normalization are also satisfied.

We begin by recalling the definition of conditionally mixing channels~\cite[Section 3]{gour2024inevitability}. These channels are characterized by being both semicausal and conditionally unital. A channel is semicausal if it does not allow information to flow from the system of interest to the side information. If this condition is violated, monotonicity of the conditional entropy cannot be expected, as revealing information about the system would reduce uncertainty. A channel is conditionally unital if it does not reduce maximal randomness. Formally, this means that the identity operator (corresponding to the maximally mixed state) on the system of interest must be preserved under the channel, while arbitrary transformations may be applied to the side information.
More rigorously, conditionally mixing channels are defined as follows~\cite[Definition 3]{gour2024inevitability}.

\begin{definition}
    Let $\mathcal{N}_{AB \rightarrow AB'}$ be a quantum channel. Then,

\begin{enumerate}
    \item $\mathcal{N}$ is said to be semicausal from $A$ to $B$, if there exists a quantum channel $\mathcal{M}_{B \rightarrow B'}$ such that for every quantum state $\rho_{AB} \in \mathcal{S}(AB)$, it holds that
    \begin{equation}
        \Tr_A\big(\mathcal{N}(\rho_{AB})\big)=\mathcal{M}(\rho_B);
    \end{equation}

    \item $\mathcal{N}$ is said to be conditionally unital, if for every quantum state $\rho_B \in \mathcal{S}(B)$, there exists a quantum state $\sigma_{B'} \in \mathcal{S}(B')$ such that
    \begin{equation}
        \mathcal{N}(I_A \otimes \rho_B)=I_A \otimes \sigma_{B'}.
    \end{equation}
\end{enumerate} 
A quantum channel is said to be conditionally mixing if it is semicausal and conditionally unital.
\end{definition}
We now prove that the new conditional entropies satisfy monotonicity under conditionally mixing channels for $\lambda \geq 0$.
\begin{theorem}[DPI under conditionally mixing channels] \label{DPItheorem CM}
Let $(\alpha,z,\lambda) \in \D$ with $\lambda \geq 0$ and $\rho_{AB}$ be a quantum state. Then, for any conditionally mixing channel $\N_{AB \rightarrow AB'}$ and $\sigma_{AB'} = \N(\rho_{AB})$, we have
	\begin{equation}
		H^{\lambda}_{\alpha,z}(A|B)_{\rho}\leq H^{\lambda}_{\alpha,z}(A|B')_{\sigma}.
	\end{equation}
\end{theorem}

Our proof leverages the asymptotic form for $H^{\lambda}_{\alpha,z}$ derived in Lemma \ref{H with universal any z}, which involves the weighted Kubo-Ando geometric mean and the universal state introduced in Section \ref{sec:characterization}. By combining the DPI for the $\alpha$-$z$ Rényi divergence with the operator monotonicity of the geometric mean under the action of a quantum channel for $0 \leq \lambda \leq 1$, we obtain the desired result. We remark that the argument breaks down for $\lambda < 0$, as the geometric mean is no longer monotone under channels in this range.

\begin{proof}
Let $\omega_{B^n}$ be the universal state in Section \ref{sec:characterization}. Since by Lemma \ref{universal state}, $\omega_{B^n}$ and $\rho_{B^n} :=\rho_B^{\otimes n}$ commute, Lemma~\ref{H with universal any z} implies that
\begin{equation}
\label{asymptotic alternative proof}
H^\lambda_{\alpha,z}(A|B)_\rho = \lim_{n \rightarrow \infty} \frac{1}{n} \frac{1}{1-\alpha} \log{ \textup{Tr}\left(\rho_{A^n B^n}^{\frac{\alpha}{2z}} \Big(\rho_{B^n}^{\frac{1-\lambda}{2}}\omega_{B^n}^{\lambda} \rho_{B^n}^{\frac{1-\lambda}{2}}\Big)^{\frac{1-\alpha}{z}}\rho_{A^n B^n}^{\frac{\alpha}{2z}}\right)^z}.
\end{equation}
Recall that the Kubo--Ando weighted geometric mean is defined as $\rho \#_\lambda \sigma := \rho^\frac{1}{2}(\rho^{-\frac{1}{2}}\sigma\rho^{-\frac{1}{2}})^\lambda \rho^\frac{1}{2}$. Since $\rho_{B^n}=\rho_B^{\otimes n}$ and $\omega_{B^n}$ commute, the sandwiched expression inside the logarithm in \eqref{asymptotic alternative proof} simplifies as
\begin{align}
\rho_{B^n}^{\frac{1-\lambda}{2}}\omega_{B^n}^{\lambda} \rho_{B^n}^{\frac{1-\lambda}{2}} = \rho_{B^n}\#_\lambda \omega_{B^n}.
\end{align}
With this substitution, the conditional entropy can be rewritten as 
\begin{align}
    H^\lambda_{\alpha,z}(A|B)_{\rho} &= \lim_{n \rightarrow \infty} -\frac{1}{n}D_{\alpha,z}\left(\rho_{A^n B^n} \middle\| I_{A^n} \otimes \rho_{B^n}\#_\lambda \omega_{B^n}\right),
\end{align}
where $D_{\alpha,z}$ denotes the $\alpha$-$z$ R\'enyi relative entropy as defined in Definition~\ref{alpha-z relative entropy}. 
Then, 
the DPI of the $\alpha$-$z$ R\'enyi relative entropy implies that
\begin{align}
D_{\alpha,z}&\left(\rho_{A^n B^n} \middle\| I_{A^n} \otimes \rho_{B^n}\#_\lambda \omega_{B^n}\right)\\
& \geq D_{\alpha,z}\left(  \N^{\otimes n}(\rho_{A^n B^n}) \middle\|      \N^{\otimes n}(I_{A^n} \otimes\rho_{B^n}\#_\lambda \omega_{B^n})\right)
 \\
\label{inequalitykubo}& \geq D_{\alpha,z}\left(    \N^{\otimes n}(\rho_{A^n B^n}) \middle\|     \N^{\otimes n}(I_{A^n} \otimes \rho_{B^n})\#_\lambda     \N^{\otimes n}(I_{A^n} \otimes \omega_{B^n})\right),
\end{align}
where in the last step we used that for a quantum channel $\mathcal{E}$, the geometric mean satisfies $\mathcal{E}(\rho \#_\lambda \sigma) \leq \mathcal{E}(\rho) \#_\lambda \mathcal{E}(\sigma)$ for $ \lambda \in [0,1]$ (see e.g.~\cite[Proposition 3.30]{hiai2017different} or~\cite[Lemma 6.3]{matsumoto2015new}). Moreover, we used that $\alpha$-$z$ R\'enyi relative entropy is antimonotone in the second argument (see, e.g.,~\cite[Lemma 17]{rubboli2022new}). Furthermore, we use the invariance under permutations of both $\mathcal{N}^{\otimes n}(\rho_{A^n B^n})$ and $\N^{\otimes n}(I_{A^n} \otimes \rho_{B^n})$ (as they are tensor products) to obtain that~\eqref{inequalitykubo} is equivalent to 
\begin{equation}
\label{No permutation}
    D_{\alpha,z}\left( \N^{\otimes n}(\rho_{A^n B^n}) \middle\| \N^{\otimes n}(I_{A^n} \otimes \rho_{B^n})\#_\lambda     \N^{\otimes n}(I_{A^n} \otimes \omega_{B^n})\right)
\end{equation}

Since by assumption $\N_{AB}$ is semicausal and conditionally unital, one can observe that (see~\cite[Appendix F]{gour2024inevitability})
\begin{equation}
    \N^{\otimes n}(I_{A^n} \otimes \rho_{B^n})=I_{A^n} \otimes \Tr_{A^n}\left(\N^{\otimes n}(\rho_{A^nB^n})\right) \,.
\end{equation}
Moreover, since the channel is conditionally unital, we have that
\begin{equation}
    \N^{\otimes n}(I_{A^n} \otimes \omega_{B^n})=I_{A^n} \otimes \Tr_{A^n}\left(\N^{\otimes n}(\pi_{A^n} \otimes \omega_{B^n})\right).
\end{equation}
In addition, since $\Tr_{A^n}\left(\N^{\otimes n}(\pi_{A^n} \otimes \omega_{B^n})\right)$ is permutation invariant, by Lemma~\ref{universal state} we can upper bound it with the universal state  $\omega_{B^{'n}}$ up to polynomial constant factors. Moreover, the weighted geometric mean satisfies $A \#_\lambda B \leq A \#_\lambda B'$ if $B \leq B'$ due to the operator monotonicity of the power in the interval $[0,1]$. Now, combining these facts, we once again invoke the antimonotonicity of the $\alpha$-$z$ R\'enyi relative entropy in its second argument.  With these ingredients, we obtain from~\eqref{No permutation} that
\begin{align}
D_{\alpha,z}&\left(\rho_{A^n B^n} \middle\| I_{A^n} \otimes \rho_{B^n}\#_\lambda \omega_{B^n}\right)\\
&\geq D_{\alpha,z}\left(\N^{\otimes n}(\rho_{A^n B^n}) \middle\| I_{A^n} \otimes \Tr_{A^n}\left(\N^{\otimes n}(\rho_{A^nB^n})\right)\#_\lambda \omega_{B^{'n}}\right)+ O(\log{n}).
\end{align}
Finally, we multiply by $-1/n$, and the limit $n \rightarrow \infty$ on both sides. Since the constant terms vanish as $O(\log{n}/n)$, both sides converge to the desired quantities.
\end{proof}
We have proved monotonicity under conditionally mixing operations. Next, according to~\cite[Definition 4]{gour2024inevitability}, to prove that our new quantity qualifies as a conditional entropy according to their definition, it is left to show invariance under isometries on $A$ and additivity.
Let us start with isometries. The proof is completely analogous to the one provided in~\cite[Appendix F]{gour2024inevitability}, Without loss of generality, we can write
\begin{equation}
    \mathcal{V}_{A \rightarrow A'}(\rho_{AB}) = \begin{pmatrix}
\rho_{AB} &  0 \\
 0 &  0_{CB}
\end{pmatrix} \,, \qquad 
I_{A'} =\begin{pmatrix}
I_A &  0 \\
 0 & I_{C}
\end{pmatrix} \,.
\end{equation}
Then, because the conditional entropy is invariant under local unitaries, we have that
\begin{align}
&D_{\alpha,z}\left(\mathcal{V}_{A \rightarrow A'}(\rho_{AB}) \middle\| I_{A'} \otimes \mathcal{G}_{\lambda,\frac{1-\alpha}{z}}(\rho_B, \sigma_B)^{\frac{z}{1-\alpha}}\right) \\
& \qquad \qquad =D_{\alpha,z}\left(\mathcal\rho_{AB} \oplus  0_{CB} \middle\| I_{A} \otimes \mathcal{G}_{\lambda,\frac{1-\alpha}{z}}(\rho_B, \sigma_B)^{\frac{z}{1-\alpha}} \oplus I_{C} \otimes \mathcal{G}_{\lambda,\frac{1-\alpha}{z}}(\rho_B, \sigma_B)^{\frac{z}{1-\alpha}}\right) \\
& \qquad \qquad = D_{\alpha,z}\left(\mathcal\rho_{AB} \middle\| I_{A} \otimes \mathcal{G}_{\lambda,\frac{1-\alpha}{z}}(\rho_B, \sigma_B)^{\frac{z}{1-\alpha}} \right) \,.
\end{align}
In the final step, the equality follows by explicitly evaluating the trace of the product of block matrices as defined in the $\alpha$-$z$ Rényi relative entropy.

Since we already proved additivity in Theorem~\ref{Additivity}, all the requirements of~\cite[Definition 4]{gour2024inevitability} are met. Finally, we note that the normalization~\cite[Section 3.2]{gour2024inevitability} is also consistent, as
\begin{equation}
H_{\alpha,z}^\lambda(A|B)_{\frac{I}{2} \otimes \rho} =1 \,.
\end{equation}
This condition is satisfied because by Lemma~\ref{maximal and minimal} the quantity lies between the minimal and maximal conditional entropies, both of which yield the same value for a state of the form $\frac{I}{2} \otimes \rho$ .

\section{Entropic duality relations}
\label{sec:duality}

In this section, we prove the duality relations for the family $H^\lambda_{\alpha,z}$. As an immediate consequence, one can observe that by making appropriate choices of the parameters $(\alpha,z,\lambda)$, relations \eqref{relation 1}, \eqref{relation 2} and \eqref{relation 3} are all special cases of this result.
\begin{theorem}[Entropic duality]\label{Duality}
	Let $(\alpha,z,\lambda)$ and $(\hat{\alpha},\hat{z},\hat{\lambda})$ be two points in $\D$. Then the duality relation
	\begin{equation}\label{duality1}
		H^{\lambda}_{\alpha,z}(A|B)_\rho+H^{\hat{\lambda}}_{\hat{\alpha},\hat{z}}(A|C)_\rho=0
	\end{equation}
	holds for every tripartite pure state $\rho_{ABC}$, whenever
	\begin{equation}\label{duality2}
		\frac{z}{1-\alpha}+\frac{\hat{z}}{1-\hat{\alpha}}=0, \qquad \frac{1-z}{1-\alpha}=\hat{\lambda}, \qquad \frac{1-\hat{z}}{1-\hat{\alpha}}=\lambda.
	\end{equation}
\end{theorem}
\begin{proof}
In the following, all the optimizations are carried out on the set of quantum states. Let us first consider the region where $\alpha<1$, $z \leq 1$, and $0 \leq \lambda \leq 1$. The relations in~\eqref{duality2} imply that $\hat{\alpha} > 1$, $\hat{z} \geq 1$ and $0 \leq \hat{\lambda} \leq 1$. Since $\frac{z}{1-\alpha}=-\frac{\hat{z}}{1-\hat{\alpha}}$, we observe that \eqref{duality1} is equivalent to 
\begin{equation}
\sup_{\sigma_B}\left\|\rho_{AB}^\frac{\alpha}{2z} \mathcal{G}_{\lambda,\frac{1-\alpha}{z}}(\rho_B,\sigma_B) \rho_{AB}^\frac{\alpha}{2z}\right\|_z = \inf_{\sigma_C} \left\|\rho_{AC}^\frac{\hat{\alpha}}{2\hat{z}} \mathcal{G}_{\hat{\lambda},\frac{1-\hat{\alpha}}{\hat{z}}}(\rho_C,\sigma_C) \rho_{AC}^\frac{\hat{\alpha}}{2\hat{z}}\right\|_{\hat{z}},
\end{equation}
where $\mathcal{G}_{\lambda,q}$ is defined in Subsection \ref{SecDPItheoremAlice}.
Next, we use the variational form of the norm from Lemma~\ref{optimization norm}. Specifically, since $z\leq 1$, it holds that $\|A\|_z = \inf_{\tau}\Trm(A \tau^\frac{z-1}{z})$, and therefore,
\begin{align}
\left\|\rho_{AB}^\frac{\alpha}{2z} \mathcal{G}_{\lambda,\frac{1-\alpha}{z}}(\rho_B,\sigma_B)^\frac{1-\alpha}{z} \rho_{AB}^\frac{\alpha}{2z}\right\|_z &=\inf_{\tau_{AB}} \Trm\left(\rho_{AB}^\frac{\alpha}{2z} \mathcal{G}_{\lambda,\frac{1-\alpha}{z}}(\rho_B,\sigma_B) \rho_{AB}^\frac{\alpha}{2z} \tau_{AB}^\frac{z-1}{z}\right) \\
&=\inf_{\tau_{AB}} \Trm\left(\rho_{AB}^\frac{1}{2} \mathcal{G}_{\lambda,\frac{1-\alpha}{z}}(\rho_B,\sigma_B) \rho_{AB}^\frac{1}{2} \mathcal{G}_{\hat{\lambda},\frac{1-\hat{\alpha}}{\hat{z}}}(\rho_{AB},\tau_{AB})\right),
\end{align}
since by \eqref{duality2} we have
\begin{equation}
\mathcal{G}_{\hat{\lambda},\frac{1-\hat{\alpha}}{\hat{z}}}(\rho_{AB},\tau_{AB}) =  \rho_{AB}^{\frac{(1-\hat{\lambda})(1-\hat{\alpha})}{2\hat{z}}} \tau_{AB}^{\frac{\hat{\lambda}(1-\hat{\alpha})}{\hat{z}}}  \rho_{AB}^{\frac{(1-\hat{\lambda})(1-\hat{\alpha})}{2\hat{z}}} = \rho_{AB}^\frac{\alpha-z}{2z} \tau_{AB}^\frac{z-1}{z} \rho_{AB}^\frac{\alpha-z}{2z}.
\end{equation}
Let $\ket{\rho}_{ABC}=\sum\limits_i \sqrt{\lambda}_i \ket{i}_{AB}\ket{i}_C$ be the Schmidt decomposition of $\rho_{ABC}$ with respect to the partition $AB:C$ and $\ket{\psi}_{AB:C}=\sum\limits_i \ket{i}_{AB}\ket{i}_C$ be the unnormalized maximally entangled state. Then, we can write
\begin{align}
\mathcal{G}_{\hat{\lambda},\frac{1-\hat{\alpha}}{\hat{z}}}(\rho_{AB},\tau_{AB})  \bigket{\psi}_{AB:C} &= \rho_{AB}^{\frac{(1-\hat{\lambda})(1-\hat{\alpha})}{2\hat{z}}} \tau_{AB}^{\frac{\hat{\lambda}(1-\hat{\alpha})}{\hat{z}}}  \rho_{AB}^{\frac{(1-\hat{\lambda})(1-\hat{\alpha})}{2\hat{z}}} \bigket{\psi}_{AB:C} \\
&=  \rho_{C}^{\frac{(1-\hat{\lambda})(1-\hat{\alpha})}{2\hat{z}}} \tau_{C}^{\top \frac{\hat{\lambda}(1-\hat{\alpha})}{\hat{z}}}  \rho_{C}^{\frac{(1-\hat{\lambda})(1-\hat{\alpha})}{2\hat{z}}} \bigket{\psi}_{AB:C}\\
& = \mathcal{G}_{\hat{\lambda},\frac{1-\hat{\alpha}}{\hat{z}}}\big(\rho_{C},\tau^{\top}_{C}\big)^\frac{1-\hat{\alpha}}{\hat{z}} \bigket{\psi}_{AB:C},
\end{align}
where $\tau_{C}$ is $\tau_{AB}$ embedded in the space $C$ through the isomorphism $AB \rightarrow C$ induced by the Schmidt basis of $\ket{\rho}_{ABC}$. The transpose is computed with respect to the basis $\ket{i}_C$. Therefore, we have
\begin{align}
&\Trm\Big(\rho_{AB}^\frac{1}{2} \mathcal{G}_{\lambda,\frac{1-\alpha}{z}}(\rho_B,\sigma_B) \rho_{AB}^\frac{1}{2} \mathcal{G}_{\hat{\lambda},\frac{1-\hat{\alpha}}{\hat{z}}}(\rho_{AB},\tau_{AB})\Big)\\
&\qquad \qquad \qquad = \bigbra{\psi} \rho_{AB}^\frac{1}{2} \mathcal{G}_{\lambda,\frac{1-\alpha}{z}}(\rho_B,\sigma_B) \rho_{AB}^\frac{1}{2} \mathcal{G}_{\hat{\lambda},\frac{1-\hat{\alpha}}{\hat{z}}}(\rho_{AB},\tau_{AB}) \bigket{\psi}_{AB:C} \\
&\qquad \qquad \qquad = \bigbra{\psi}\rho_{AB}^\frac{1}{2} \mathcal{G}_{\lambda,\frac{1-\alpha}{z}}(\rho_B,\sigma_B) \rho_{AB}^\frac{1}{2}  \otimes \mathcal{G}_{\hat{\lambda},\frac{1-\hat{\alpha}}{\hat{z}}}\big(\rho_{C},\tau^{\top}_{C}\big) \bigket{\psi}_{AB:C} \\
&\qquad \qquad \qquad = \bigbra{\rho} \mathcal{G}_{\lambda,\frac{1-\alpha}{z}}(\rho_B,\sigma_B)  \otimes \mathcal{G}_{\hat{\lambda},\frac{1-\hat{\alpha}}{\hat{z}}}\big(\rho_{C},\tau^{\top}_{C}\big) \bigket{\rho}_{ABC}.
\end{align}
Since we take the infimum of all states $\tau_C$, we can drop the transpose. This gives
\begin{align}
&\sup_{\sigma_B}\left\|\rho_{AB}^\frac{\alpha}{2z} \mathcal{G}_{\lambda,\frac{1-\alpha}{z}}(\rho_B,\sigma_B) \rho_{AB}^\frac{\alpha}{2z}\right\|_z \\
&\qquad \qquad = \sup_{\sigma_B}\inf_{\tau_C} \bigbra{\rho} \mathcal{G}_{\lambda,\frac{1-\alpha}{z}}(\rho_B,\sigma_B) \otimes \mathcal{G}_{\hat{\lambda},\frac{1-\hat{\alpha}}{\hat{z}}}(\rho_{C},\tau_{C}) \bigket{\rho}_{ABC} \\
\label{purification 1}
&\qquad \qquad= \sup_{\sigma_B}\inf_{\tau_C} \bigbra{\rho} \rho_B^{\frac{(1-\lambda)(1-\alpha)}{2z}} \sigma_B^{\frac{\lambda(1-\alpha)}{z}} \rho_B^{\frac{(1-\lambda)(1-\alpha)}{2z}} \otimes \rho_{C}^{\frac{(1-\hat{\lambda})(1-\hat{\alpha})}{2\hat{z}}} \tau_{C}^{\frac{\hat{\lambda}(1-\hat{\alpha})}{\hat{z}}}  \rho_{C}^{\frac{(1-\hat{\lambda})(1-\hat{\alpha})}{2\hat{z}}}\bigket{\rho}_{ABC}.
\end{align}
We rewrote the last line since it is clear that the latter function is concave in $\sigma_B$ as $\lambda(1-\alpha)/z \in (0,1)$ and convex in $\tau_C$ as $\hat{\lambda}(1-\hat{\alpha})/\hat{z} \in (-1,0)$ (this follows from operator convexity/concavity in this range). Hence, we can use Sion's minimax theorem to exchange the $\inf$ with the $\sup$, and we obtain what we would get if we repeat the argument for
\begin{equation}
    \inf_{\tau_C} \Big\|\rho_{AC}^\frac{\hat{\alpha}}{2\hat{z}} \mathcal{G}_{\hat{\lambda},\frac{1-\hat{\alpha}}{\hat{z}}}(\rho_C,\tau_C)\rho_{AC}^\frac{\hat{\alpha}}{2\hat{z}}\Big\|_{\hat{z}}.
\end{equation}
This shows the duality relation for the first region.

The proof of the duality relation between the other regions is completely analogous. In some cases, the minimax theorem is not even needed, as we either have a double supremum or a double infimum. For example, let us consider $0 \leq \alpha \leq 1$, $z \geq 1$  and $0 \leq \lambda \leq 1$. Using the above relations, this implies that $\hat{\alpha} \geq 1$, $\hat{z} \geq 1$ and $\hat{\lambda} \leq 0$. We now repeat the same steps. The only difference is that now the variational formulation of the norm contains a supremum for $z\geq 1$ and the dual conditional entropy is now defined with a supremum since $\hat{\lambda} \leq 0$. We obtain
\begin{align}
&\sup_{\sigma_B}\left\|\rho_{AB}^\frac{\alpha}{2z} \mathcal{G}_{\lambda,\frac{1-\alpha}{z}}(\rho_B,\sigma_B) \rho_{AB}^\frac{\alpha}{2z}\right\|_z \\
&\qquad \qquad = \sup_{\sigma_B}\sup_{\tau_C} \bigbra{\rho} \mathcal{G}_{\lambda,\frac{1-\alpha}{z}}(\rho_B,\sigma_B)  \otimes \mathcal{G}_{\hat{\lambda},\frac{1-\hat{\alpha}}{\hat{z}}}(\rho_{C},\tau_{C}) \bigket{\rho}_{ABC} \\
\label{purification 2}
& \qquad \qquad = \sup_{\sigma_B}\sup_{\tau_C} \bigbra{\rho} \rho_B^{\frac{(1-\lambda)(1-\alpha)}{2z}} \sigma_B^{\frac{\lambda(1-\alpha)}{z}} \rho_B^{\frac{(1-\lambda)(1-\alpha)}{2z}} \otimes \rho_{C}^{\frac{(1-\hat{\lambda})(1-\hat{\alpha})}{2\hat{z}}} \tau_{C}^{\frac{\hat{\lambda}(1-\hat{\alpha})}{\hat{z}}}  \rho_{C}^{\frac{(1-\hat{\lambda})(1-\hat{\alpha})}{2\hat{z}}}\bigket{\rho}_{ABC}.
\end{align}
This is equivalent to the result we would obtain by repeating the argument for 
\begin{equation}
    \sup_{\sigma_C} \left\|\rho_{AC}^\frac{\hat{\alpha}}{2\hat{z}} \mathcal{G}_{\hat{\lambda},\frac{1-\hat{\alpha}}{\hat{z}}}(\rho_C,\sigma_C) \rho_{AC}^\frac{\hat{\alpha}}{2\hat{z}}\right\|_{\hat{z}} \,.
\end{equation}
The proof is complete.
\end{proof}

\begin{remark}
\label{variational purification}
   From equations~\eqref{purification 1} and~\eqref{purification 2}, as well as similar equations for the other ranges, we obtain the following variational form involving the purification of $\rho_{AB}$. As shown in Section~\ref{sec:chainrules}, this form is particularly useful for deriving chain rules for the new conditional entropy. Specifically, for any state $\rho_{AB}$, we have
	\begin{equation}
		H^\lambda_{\alpha,z}(A|B)= \frac{z}{1-\alpha}\log{\opt_{\sigma_B}\opt_{\tau_C} \bigbra{\rho} \mathcal{G}_{\lambda,\frac{1-\alpha}{z}}(\rho_B,\sigma_B)  \otimes \mathcal{G}_{\hat{\lambda},\frac{1-\hat{\alpha}}{\hat{z}}}(\rho_{C},\tau_{C})\bigket{\rho}_{ABC}},
	\end{equation}
 where $\ket{\rho}_{ABC}$ is a purification of $\rho_{AB}$, and $\mathcal{G}_{\lambda,q}$ is as in Subsection \ref{SecDPItheoremAlice}.
\end{remark}

\begin{remark}
\label{explicit computation}
	The dual parameters $(\hat{\alpha},\hat{z},\hat{\lambda})$ can be computed explicitly in terms of the primal parameters $(\alpha,z,\lambda)$ via~\eqref{duality2}, as follows:
	\begin{equation}
    \label{explcit computation duality}
		\hat{\alpha}=\frac{z+(\lambda-1)(\alpha-1)}{z+\lambda(\alpha-1)}, \qquad  \hat{z}=\frac{z}{z+\lambda(\alpha-1)},\qquad \hat{\lambda}=\frac{z-1}{\alpha-1}.
	\end{equation}
    Note that the regions $\mathcal{D}_1$ and $\mathcal{D}_2$ are mapped to each other bijectively under this transformation.
\end{remark}

\begin{remark}
The duality relations provide a closed-form solution for pure states. Specifically, for a pure state $\ket{\rho}_{AB}$, we have
\begin{equation}
H^\lambda_{\alpha,z}(A|B)_\rho = - H_{\hat{\alpha}}(\rho_B), \qquad \hat{\alpha}=\frac{z+(\lambda-1)(\alpha-1)}{z+\lambda(\alpha-1)},
\end{equation}
where $H_{\hat{\alpha}}(\rho_B) = \frac{1}{1-\hat{\alpha}}\log \Trm\left(\rho_B^{\hat{\alpha}}\right)$ is the R\'enyi entropy of the marginal state.
\end{remark}

\section{A reparameterization of the DPI region}
\label{sec:reparametrization}

In this section, we derive a new coordinate system that allows us to express the constraints on the DPI region 
$\D$ in a linear form. This transformation not only simplifies the representation but also allows us to visualize the duality relations as reflections across the third axis, thereby providing deeper insights into their geometric properties.
Furthermore, the dual constraint of each DPI condition is readily observable (see Tables~\ref{table I} and~\ref{table II}).

\renewcommand{\arraystretch}{1.3}
\begin{table}[h]
  \centering
   \begin{minipage}[t]{0.38\textwidth}
  \begin{tabular}{c|c}
    \hline
    $(\alpha,z,\lambda)$ & $(x_1,x_2,x_3)$  \\
    \hline
    $\alpha \geq 0$ & $x_1+x_2-x_3-2 \geq 0$  \\
    $z \geq 1-\alpha$ & $x_1 \geq 1$  \\
    $z \geq \alpha$ & $x_1-x_2+x_3+2 \geq 0$  \\
    $\lambda \leq 1$ & $x_1-x_2-x_3-2 \leq 0$\\
    $\lambda \geq 1-\frac{z}{1-\alpha}$ & $3x_1-x_2-x_3-2 \geq 0$\\
    \hline
  \end{tabular}
  \caption{Constraints for $\alpha \leq 1$}
  \label{table I}
  \end{minipage}
  \hspace{-5pt}
\begin{tikzpicture}[thick, <->]
  \draw (-3.5,-0.7) -- (0,-0.7);
  \node[above] at (-1.7,0) {$(x_1,x_2,x_3)$};
  \node[above] at (-1.7,-0.5) {$\leftrightarrow (-\hat{x}_1,-\hat{x}_2,\hat{x}_3)$};
\end{tikzpicture}
  \begin{minipage}[t]{0.38\textwidth}
  \begin{tabular}{c|c}
    \hline
    $(\hat{\alpha},\hat{z},\hat{\lambda})$ & $(\hat{x}_1,\hat{x}_2,\hat{x}_3)$  \\
    \hline
    $\hat{\lambda} \geq 1+\frac{\hat{z}}{1-\hat{\alpha}}$ & $\hat{x}_1+\hat{x}_2+\hat{x}_3+2 \leq 0$  \\
    $\hat{z} \geq \hat{\alpha}-1$ & $\hat{x}_1 \leq 1$  \\
    $\hat{\lambda} \leq 1$ & $\hat{x}_1-\hat{x}_2-\hat{x}_3-2 \leq 0$  \\
    $\hat{z} \leq \hat{\alpha}$ & $\hat{x}_1-\hat{x}_2+\hat{x}_3+2 \geq 0$\\
    $2\hat{z} \geq \hat{\alpha}$ & $3\hat{x}_1-\hat{x}_2+\hat{x}_3+2 \leq 0$\\
    \hline
  \end{tabular}
  \caption{Constraints for $\hat{\alpha} \geq 1$}
  \label{table II}
  \end{minipage}
  \caption*{The two tables list the DPI constraints for $\alpha \leq 1$ and $\alpha \geq 1$. In each row, we write the same constraint in both the coordinates $(\alpha,z,\lambda)$ and $(x_1,x_2,x_3)$. We also show the action of the duality transformations on the DPI constraints. Under duality, each row of one table is mapped into the same row of the other table. In the new coordinates, this mapping corresponds to the reflection across the third axis $(x_1,x_2,x_3) \leftrightarrow (-\hat{x}_1,-\hat{x}_2,\hat{x}_3)$.}
\end{table}

\begin{figure}[h]
	\centering
	\includegraphics[width=0.75\textwidth]{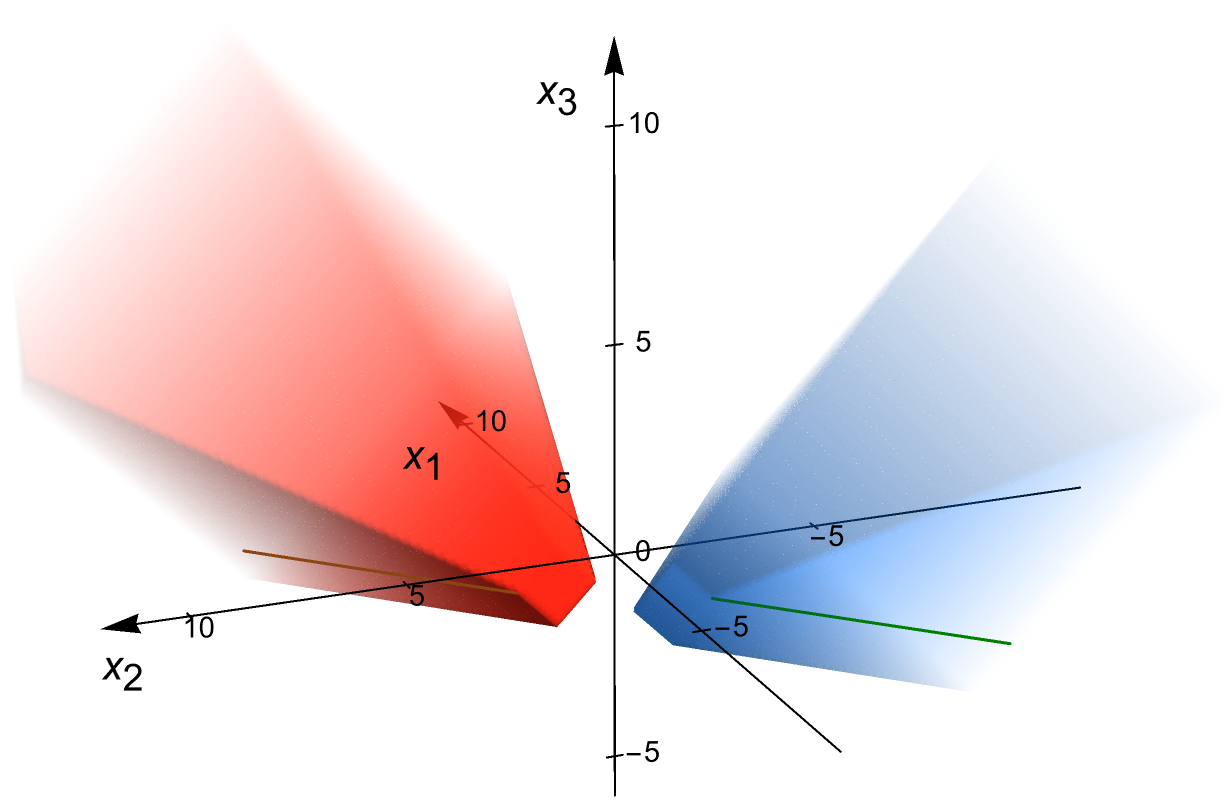}
	\caption{The figure shows the DPI region $\D$ in the coordinates $x_1 = \frac{z}{1-\alpha}$,  $x_2 = \frac{1}{1-\alpha} - \lambda$ and $x_3=\frac{z-1}{1-\alpha}-\lambda$. In these coordinates, the DPI region, depicted in yellow, consists of a truncated polyhedral cone and its mirror image. Moreover, the duality relations in Theorem~\ref{Duality} become $
		x_1=-\hat{x}_1$, $x_2=-\hat{x}_2$ and $x_3=\hat{x}_3$. Hence, under duality, the points get reflected across the $x_3$ axis. We show the duality relation~\eqref{relation 3} that connects the Petz arrow up with the sandwiched arrow down. The green line corresponds to the Petz arrow up for $\alpha >1$ while the brown line to the sandwiched arrow down for $\alpha<1$.}
	\label{DPI figure}
\end{figure}

Note first that the duality relations in equation~\eqref{duality2} can be rewritten in the equivalent form
\begin{equation}
	\frac{z}{1-\alpha}+\frac{\hat{z}}{1-\hat{\alpha}}=0,\qquad \frac{1}{1-\alpha}+\frac{1}{1-\hat{\alpha}} = \lambda+\hat{\lambda},\qquad \frac{z-1}{1-\alpha}-\frac{\hat{z}-1}{1-\hat{\alpha}} = \lambda-\hat{\lambda}.
\end{equation}
These relations suggest the following change of coordinates
\begin{equation}
x_1 = \frac{z}{1-\alpha}, \qquad  x_2 = \frac{1}{1-\alpha} - \lambda, \qquad x_3=\frac{z-1}{1-\alpha}-\lambda.
\end{equation}
In these coordinates, the duality relations turn into
\begin{equation}
x_1=-\hat{x}_1,\quad x_2=-\hat{x}_2,\quad x_3=\hat{x}_3,
\end{equation}
where the $\hat{x}_1,\hat{x}_2,\hat{x}_3$ are defined similarly using $\hat{\alpha},\hat{z},\hat{\lambda}$. Therefore, under duality, the parameters get reflected with respect to the $x_3$ axis.
For completeness, we also write the inverse relations. We have
\begin{equation}
\alpha = \frac{x_1+x_2-x_3-2}{x_1+x_2-x_3}, \qquad z=\frac{2x_1}{x_1+x_2-x_3}, \qquad  \lambda = \frac{x_1-x_3-x_2}{2}.
\end{equation}
Using these, we rewrite the DPI constraints in terms of the new coordinates in Table~\ref{table I} and~\ref{table II}. Each row of Table~\ref{table I} corresponds to a specific DPI constraint. We write the same constraint in both coordinates $(\alpha,z,\lambda)$ and $(x_1,x_2,x_3)$. Under duality, each constraint is mapped to its dual in the corresponding row of Table~\ref{table II}, and vice versa.

In Fig.~\ref{DPI figure}, the DPI region $\D$ is illustrated in the coordinates $(x_1,x_2,x_3)$. The entropy $\widebar{H}^\uparrow_{\alpha}$ for $\alpha>1$ corresponds to the green line and the entropy $\widetilde{H}^\downarrow_{\alpha}$ for $\alpha<1$ corresponds to the brown line. These two quantities are dual to each other according to the relation~\eqref{relation 3}. Thus, these two lines are symmetric with respect to the third axis  $x_3$.

\section{Asymptotic and variational characterizations}
\label{sec:characterization}

In this section, we derive alternative formulations for $H^\lambda_{\alpha,z}$, which play a key role in establishing the chain rule in the subsequent section. We begin by obtaining expressions that incorporate the universal state in Proposition \ref{H with universal any z}, followed by the derivation of several useful forms involving the Schatten norm in Lemma~\ref{Dupuis variational}, inspired by the approach in~\cite{dupuis2015chain} for sandwiched conditional entropies. In Lemmas~\ref{regularized z geq 1} and~\ref{regularized z leq 1}, these elements are combined to obtain an additional form. This formulation will be applied in the following section, using complex interpolation theory as developed in~\cite{dupuis2015chain} to establish the chain rule.

It is noteworthy that the results contained in this section are of independent interest. Notably, we utilize the asymptotic expressions to establish the monotonicity properties with respect to the parameters $\alpha$ and $z$ in Section~\ref{sec:mono}. Moreover, the same form provides the foundation for an alternative proof of the DPI for the range $0 \leq \lambda \leq 1$ in Subsection~\ref{alternativeproof}.

We frequently make use of the following lemma from the representation theory of the symmetric group $S_n$, consisting of permutations of $n$ elements. Let $U_{A^n}$ denote the natural unitary representation of $S_n$ on $A^n$ that permutes the subsystems $A_1, A_2, \dots, A_n$. An operator $L_{A^n}$ is said to be permutation-invariant if it satisfies $U_{A^n}(\pi) L_{A^n} U_{A^n}(\pi)^\dagger = L_{A^n}$ for all $\pi \in S_n$. Similarly, $L_{A^n}$ is said to be invariant under $n$-fold product unitaries if it satisfies $V_A^{\otimes n} L_{A^n} (V_A^{\dagger})^{\otimes n} = L_{A^n}$ for all unitaries $V_A$ on $A$. This lemma is the key to canceling out unitaries in certain optimizations. (See~\cite[Lemma 1]{hayashitomamichel15c} for a proof.)
\begin{lemma}\label{universal state}
Let $A$ be a quantum system with $|A|=d$. For any $n \in \mathbb{N}$, there exists a state $\omega_{A^n} \in \mathcal{S}(A^n)$, referred to as the universal state, such that for any permutation-invariant state $\tau_{A^n} \in \mathcal{S}(A^n)$,
\begin{equation}
    \tau_{A^n} \leq g_{n,d}\, \omega_{A^n} \qquad \text{where} \qquad g_{n,d} \leq (n+1)^{d^2-1}.
\end{equation}
Moreover, the state $\omega_{A^n}$ is permutation-invariant, commutes with all permutation-invariant states, and is invariant under $n$-fold product unitaries.
\end{lemma}

\begin{proposition}[Asymptotic expression]
	\label{H with universal any z}
	Let $\rho_{AB}$ be a quantum state and $(\alpha,z,\lambda) \in \D$. Then, we have
	\begin{align}
		(i) \qquad H^\lambda_{\alpha,z}(A|B)_\rho &= \lim_{n \rightarrow \infty} \frac{1}{n} \frac{1}{1-\alpha} 
		\log{ \textup{Tr}\left(\rho_{A^nB^n}^{\frac{\alpha}{2z}} \rho_{B^n}^{\frac{(1-\lambda)(1-\alpha)}{2z}}\omega_{B^n}^{\frac{\lambda(1-\alpha)}{z}} \rho_{B^n}^{\frac{(1-\lambda)(1-\alpha)}{2z}}\rho_{A^n B^n}^{\frac{\alpha}{2z}}\right)^z}. \\
		(ii) \qquad H^\lambda_{\alpha,z}(A|B)_\rho &= \lim_{n \rightarrow \infty} \frac{1}{n} \frac{z}{1-\alpha} 
		\log{\opt_{\sigma_B} \textup{Tr}\left(\rho_{A^n B^n}^{\frac{\alpha}{2z}} \rho_{B^n}^{\frac{(1-\lambda)(1-\alpha)}{2z}} (\sigma^{\otimes n}_B)^\frac{\lambda(1-\alpha)}{z}\rho_{B^n}^{\frac{(1-\lambda)(1-\alpha)}{2z}}\rho_{A^n B^n}^{\frac{\alpha}{2z}}\omega_{A^nB^n}^{1-\frac{1}{z}}\right)}.
	\end{align}
	 $(iii) \; \;  $ For $z \leq 1$ and $\lambda(1-\alpha) \leq 0$, we have
\begin{equation}
	H^\lambda_{\alpha,z}(A|B)_\rho = \lim_{n \rightarrow \infty} \frac{1}{n} \frac{z}{1-\alpha} \log{ \textup{Tr}\left(\rho_{A^n B^n}^{\frac{\alpha}{2z}} \rho_{B^n}^{\frac{(1-\lambda)(1-\alpha)}{2z}}\omega_{B^n}^{\frac{\lambda(1-\alpha)}{z}} \rho_{B^n}^{\frac{(1-\lambda)(1-\alpha)}{2z}}\rho_{A^n B^n}^{\frac{\alpha}{2z}}\omega_{A^nB^n}^{1-\frac{1}{z}}\right)}.
\end{equation}
Here, we denoted $\rho_{A^n B^n} = \rho_{AB}^{\otimes n}$ and $\rho_{B^n}=\rho_B^{\otimes n}$.
\end{proposition}
\begin{proof}
	(i) Let us start with the case  $\lambda(1-\alpha)/z \in [-1,0)$. In this range, the power is operator antimonotone. We apply the properties of the universal state in Lemma~\ref{universal state}. The following chain of inequalities holds:
\begin{align}
&\log{\inf_{\sigma_B}\Trm\left(\rho_{AB}^\frac{\alpha}{2z} \rho_B^{\frac{(1-\lambda)(1-\alpha)}{2z}}\sigma_B^{\frac{\lambda(1-\alpha)}{z}} \rho_B^{\frac{(1-\lambda)(1-\alpha)}{2z}}\rho_{AB}^\frac{\alpha}{2z}\right)^z} \\
&\quad \qquad= \frac{1}{n} \log{\inf_{\sigma_B} \textup{Tr}\left(\rho_{A^n B^n}^{\frac{\alpha}{2z}} \rho_{B^n}^{\frac{(1-\lambda)(1-\alpha)}{2z}}\big(\sigma_{B}^{\otimes n}\big)^{\frac{\lambda(1-\alpha)}{z}} \rho_{B^n}^{\frac{(1-\lambda)(1-\alpha)}{2z}}\rho_{A^n B^n}^{\frac{\alpha}{2z}}\right)^z} \\
\label{symmetric inequality}
&\quad \qquad \geq \frac{1}{n}  \log{ \textup{Tr}\left(\rho_{A^n B^n}^{\frac{\alpha}{2z}} \rho_{B^n}^{\frac{(1-\lambda)(1-\alpha)}{2z}}\omega_{B^n}^{\frac{\lambda(1-\alpha)}{z}} \rho_{B^n}^{\frac{(1-\lambda)(1-\alpha)}{2z}}\rho_{A^n B^n}^{\frac{\alpha}{2z}}\right)^z} + 
\lambda(1-\alpha)\frac{1}{n}\log{g_{n,d}} \\
&\quad \qquad \geq \frac{1}{n}  \log{\inf_{\sigma_{B^n}} \textup{Tr}\left(\rho_{A^n B^n}^{\frac{\alpha}{2z}} \rho_{B^n}^{\frac{(1-\lambda)(1-\alpha)}{2z}}\sigma_{B^n}^{\frac{\lambda(1-\alpha)}{z}} \rho_{B^n}^{\frac{(1-\lambda)(1-\alpha)}{2z}}\rho_{A^n B^n}^{\frac{\alpha}{2z}}\right)^z} + \lambda(1-\alpha) \frac{1}{n}\log{g_{n,d}} \\
\label{additivity conditional}
&\quad \qquad = \log{\inf_{\sigma_B}\Trm\left(\rho_{AB}^\frac{\alpha}{2z}  \rho_B^{\frac{(1-\lambda)(1-\alpha)}{2z}}\sigma_B^{\frac{\lambda(1-\alpha)}{z}} \rho_B^{\frac{(1-\lambda)(1-\alpha)}{2z}}\rho_{AB}^\frac{\alpha}{2z}\right)^z} + \lambda(1-\alpha)\frac{1}{n}\log{g_{n,d}}.
\end{align}
	In~\eqref{symmetric inequality}, we used that the exponent of $\sigma_B$ is operator antimonotone, the inequality $\sigma_B^{\otimes n} \leq g_{n,d}\, \omega_{B^n}$ holds, and the trace functional $M \mapsto \Tr(f(M))$ inherits monotonicity from the function $f$ (see, e.g.,~\cite{carlen2010trace}). Here, $d=d_B$. Furthermore, in~\eqref{additivity conditional}, we applied the additivity of the conditional entropy. We then multiply by $1/(1-\alpha)$ and take the limit $n \rightarrow \infty$ on both sides. Note that they both converge to the conditional entropy. Indeed, the term $\frac{1}{n} \log{g_{n,d}}$ vanishes. Therefore, the quantity involving the universal state in~\eqref{symmetric inequality} must equal the conditional entropy. The case  $\lambda(1-\alpha)/z \in (0,1]$  is similar, with all the inequalities above reversed.
 
	(ii) The proof follows a similar argument to (i), but instead of employing $\sigma_B^{\otimes n} \leq g_{n,d} \omega_{B^n}$, it relies on the inequality
	\(\tau_{AB}^{\otimes n} \leq g_{n,d^2}\omega_{A^n B^n}\) and the variational form of the norm. Here, $d=\max\{d_A,d_B\}$. Let us first consider the case $z \leq 1$.  The following chain of inequalities holds:
\begin{align}
&\log{\opt_{\sigma_B}\Trm\left(\rho_{AB}^\frac{\alpha}{2z} \rho_B^{\frac{(1-\lambda)(1-\alpha)}{2z}}\sigma_B^{\frac{\lambda(1-\alpha)}{z}} \rho_B^{\frac{(1-\lambda)(1-\alpha)}{2z}}\rho_{AB}^\frac{\alpha}{2z}\right)^z} \\
	& \qquad= \frac{1}{n} \log{\opt_{\sigma_B} \textup{Tr}\left(\rho_{A^n B^n}^{\frac{\alpha}{2z}} \rho_{B^n}^{\frac{(1-\lambda)(1-\alpha)}{2z}}\big(\sigma_{B}^{\otimes n}\big)^{\frac{\lambda(1-\alpha)}{z}} \rho_{B^n}^{\frac{(1-\lambda)(1-\alpha)}{2z}}\rho_{A^n B^n}^{\frac{\alpha}{2z}}\right)^z} \\
  & \qquad= \frac{z}{n} \log{\opt_{\sigma_B} \inf_{\tau_{AB}} \textup{Tr}\left(\rho_{A^n B^n}^{\frac{\alpha}{2z}} \rho_{B^n}^{\frac{(1-\lambda)(1-\alpha)}{2z}}\big(\sigma_{B}^{\otimes n}\big)^{\frac{\lambda(1-\alpha)}{z}} \rho_{B^n}^{\frac{(1-\lambda)(1-\alpha)}{2z}}\rho_{A^n B^n}^{\frac{\alpha}{2z}}\big(\tau_{AB}^{\otimes n}\big)^{1-\frac{1}{z}}\right)} \\
  \label{symmetric inequality (ii)}
  & \qquad \geq \frac{z}{n}  \log{\opt_{\sigma_B} \textup{Tr}\left(\rho_{A^n B^n}^{\frac{\alpha}{2z}} \rho_{B^n}^{\frac{(1-\lambda)(1-\alpha)}{2z}}\big(\sigma_{B}^{\otimes n}\big)^{\frac{\lambda(1-\alpha)}{z}} \rho_{B^n}^{\frac{(1-\lambda)(1-\alpha)}{2z}}\rho_{A^n B^n}^{\frac{\alpha}{2z}}\omega_{A^n B^n}^{1-\frac{1}{z}}\right)} + \frac{z-1}{n}  \log{g_{n,d^2}}\\
  & \qquad \geq \frac{z}{n}  \log{\opt_{\sigma_B} \inf_{\tau_{A^n B^n}}\textup{Tr}\left(\rho_{A^n B^n}^{\frac{\alpha}{2z}} \rho_{B^n}^{\frac{(1-\lambda)(1-\alpha)}{2z}}\big(\sigma_{B}^{\otimes n}\big)^{\frac{\lambda(1-\alpha)}{z}} \rho_{B^n}^{\frac{(1-\lambda)(1-\alpha)}{2z}}\rho_{A^n B^n}^{\frac{\alpha}{2z}}\tau_{A^n B^n}^{1-\frac{1}{z}}\right)} + \frac{z-1}{n}  \log{g_{n,d^2}}\\ 
  \label{norm lower bound}
  & \qquad = \frac{1}{n}  \log{\opt_{\sigma_B} \textup{Tr}\left(\rho_{A^n B^n}^{\frac{\alpha}{2z}} \rho_{B^n}^{\frac{(1-\lambda)(1-\alpha)}{2z}}\big(\sigma_{B}^{\otimes n}\big)^{\frac{\lambda(1-\alpha)}{z}} \rho_{B^n}^{\frac{(1-\lambda)(1-\alpha)}{2z}}\rho_{A^n B^n}^{\frac{\alpha}{2z}}\right)^z} + \frac{z-1}{n}  \log{g_{n,d^2}}\\
  \label{additivity conditional (ii)}
  & \qquad = \log{\opt_{\sigma_B}\Trm\left(\rho_{AB}^\frac{\alpha}{2z}  \rho_B^{\frac{(1-\lambda)(1-\alpha)}{2z}}\sigma_B^{\frac{\lambda(1-\alpha)}{z}} \rho_B^{\frac{(1-\lambda)(1-\alpha)}{2z}}\rho_{AB}^\frac{\alpha}{2z}\right)^z} + \frac{z-1}{n}  \log{g_{n,d^2}}.
\end{align}
	In~\eqref{symmetric inequality (ii)}, we used that the exponent of $\tau_{AB}^{\otimes n}$ is operator antimonotone and that the inequality $\tau_{AB}^{\otimes n} \leq g_{n,d^2} \omega_{A^n B^n}$ holds. In~\eqref{norm lower bound} we lower bounded the expression with the infimum over all states $\tau_{A^n B^n}$. Furthermore, in~\eqref{additivity conditional (ii)}, we applied the variational form of the norm and its additivity. We then multiply by $1/(1-\alpha)$ and take the limit $n \rightarrow \infty$ on both sides. Both expressions converge to the conditional entropy. Indeed, the term $\frac{1}{n} \log{g_{n,d^2}}$ vanishes. Therefore, the quantity involving the universal state in~\eqref{symmetric inequality (ii)} must equal the conditional entropy. The case $z \geq 1$ follows similarly, but with the inequalities reversed.
 
	(iii) In this case, we utilize both the inequality \(\sigma_B^{\otimes n} \leq g_{n,d}\omega_{B^n}\) and \(\tau_{AB}^{\otimes n} \leq g_{n,d^2}\omega_{A^n B^n}\). Note that this form does not apply when both a supremum and an infimum are involved in the optimizations, as the latter inequalities go in the same direction. Therefore, we have limited our statement to a restricted range of parameters. The following chain of inequalities holds:
	\begin{align}	&\log{\inf_{\sigma_B}\Trm\left(\rho_{AB}^\frac{\alpha}{2z} \rho_B^{\frac{(1-\lambda)(1-\alpha)}{2z}}\sigma_B^{\frac{\lambda(1-\alpha)}{z}} \rho_B^{\frac{(1-\lambda)(1-\alpha)}{2z}}\rho_{AB}^\frac{\alpha}{2z}\right)^z} \\
		& \; = \frac{1}{n} \log{\inf_{\sigma_B} \textup{Tr}\left(\rho_{A^n B^n}^{\frac{\alpha}{2z}} \rho_{B^n}^{\frac{(1-\lambda)(1-\alpha)}{2z}}\big(\sigma_{B}^{\otimes n}\big)^{\frac{\lambda(1-\alpha)}{z}} \rho_{B^n}^{\frac{(1-\lambda)(1-\alpha)}{2z}}\rho_{A^n B^n}^{\frac{\alpha}{2z}}\right)^z} \\
& \; = \frac{z}{n} \log{\inf_{\sigma_B} \inf_{\tau_{AB}}\textup{Tr}\left(\rho_{A^n B^n}^{\frac{\alpha}{2z}} \rho_{B^n}^{\frac{(1-\lambda)(1-\alpha)}{2z}}\big(\sigma_{B}^{\otimes n}\big)^{\frac{\lambda(1-\alpha)}{z}} \rho_{B^n}^{\frac{(1-\lambda)(1-\alpha)}{2z}}\rho_{A^n B^n}^{\frac{\alpha}{2z}}\big(\tau_{AB}^{\otimes n}\big)^{1-\frac{1}{z}}\right)} \\
\label{antimonotone 1 (iii)}
  & \; \geq  \frac{z}{n} \log{ \inf_{\tau_{AB}}\textup{Tr}\left(\rho_{A^n B^n}^{\frac{\alpha}{2z}} \rho_{B^n}^{\frac{(1-\lambda)(1-\alpha)}{2z}}\omega_{B^n}^{\frac{\lambda(1-\alpha)}{z}} \rho_{B^n}^{\frac{(1-\lambda)(1-\alpha)}{2z}}\rho_{A^n B^n}^{\frac{\alpha}{2z}}\big(\tau_{AB}^{\otimes n}\big)^{1-\frac{1}{z}}\right)} + \lambda(1-\alpha) \frac{1}{n} \log{g_{n,d}} \\
  \label{antimonotone  2 (iii)}
  & \; \geq \frac{z}{n} \log{\textup{Tr}\left(\rho_{A^n B^n}^{\frac{\alpha}{2z}} \rho_{B^n}^{\frac{(1-\lambda)(1-\alpha)}{2z}}\omega_{B^n}^{\frac{\lambda(1-\alpha)}{z}} \rho_{B^n}^{\frac{(1-\lambda)(1-\alpha)}{2z}}\rho_{A^n B^n}^{\frac{\alpha}{2z}}\omega_{A^n B^n}^{1-\frac{1}{z}}\right)} \\
  &\qquad \qquad + \lambda(1-\alpha) \frac{1}{n} \log{g_{n,d}} +  \frac{z-1}{n} \log{g_{n,d^2}} \\
		\label{norm lower bound 1 (iii)}
		& \; \geq \frac{z}{n} \log{\inf_{\tau_{A^n B^n}} \textup{Tr}\left(\rho_{A^n B^n}^{\frac{\alpha}{2z}} \rho_{B^n}^{\frac{(1-\lambda)(1-\alpha)}{2z}}\omega_{B^n}^{\frac{\lambda(1-\alpha)}{z}} \rho_{B^n}^{\frac{(1-\lambda)(1-\alpha)}{2z}}\rho_{A^n B^n}^{\frac{\alpha}{2z}}\tau_{A^n B^n}^{1-\frac{1}{z}}\right)} + O\left(\frac{\log{n}}{n}\right)\\
  \label{variational norm (iii)}
  & \; = \frac{1}{n} \log{\textup{Tr}\left(\rho_{A^n B^n}^{\frac{\alpha}{2z}} \rho_{B^n}^{\frac{(1-\lambda)(1-\alpha)}{2z}}\omega_{B^n}^{\frac{\lambda(1-\alpha)}{z}} \rho_{B^n}^{\frac{(1-\lambda)(1-\alpha)}{2z}}\rho_{A^n B^n}^{\frac{\alpha}{2z}}\right)^z} + O\left(\frac{\log{n}}{n}\right)\\
  \label{norm lower bound 2 (iii)}
  &\; \geq \frac{1}{n} \log{\inf_{\sigma_{B^n}}\textup{Tr}\left(\rho_{A^n B^n}^{\frac{\alpha}{2z}} \rho_{B^n}^{\frac{(1-\lambda)(1-\alpha)}{2z}}\sigma_{B^n}^{\frac{\lambda(1-\alpha)}{z}} \rho_{B^n}^{\frac{(1-\lambda)(1-\alpha)}{2z}}\rho_{A^n B^n}^{\frac{\alpha}{2z}}\right)^z} + O\left(\frac{\log{n}}{n}\right)\\
		\label{additivity conditional (iii)}
  & \; = \log{\opt_{\sigma_B}\Trm\left(\rho_{AB}^\frac{\alpha}{2z}  \rho_B^{\frac{(1-\lambda)(1-\alpha)}{2z}}\sigma_B^{\frac{\lambda(1-\alpha)}{z}} \rho_B^{\frac{(1-\lambda)(1-\alpha)}{2z}}\rho_{AB}^\frac{\alpha}{2z}\right)^z} + O\left(\frac{\log{n}}{n}\right).
	\end{align}
	In~\eqref{antimonotone  1 (iii)} and~\eqref{antimonotone  2 (iii)}, we used that the exponents of $\sigma_B^{\otimes n}$ and $\tau_{AB}^{\otimes n}$ are operator antimonotone and that the inequalities $\sigma_{B}^{\otimes n} \leq g_{n,d} \omega_{B^n}$ and $\tau_{AB}^{\otimes n} \leq g_{n,d^2} \omega_{A^n B^n}$ hold. In~\eqref{norm lower bound 1 (iii)} and~\eqref{norm lower bound 2 (iii)} we simply lower bounded the expressions with the infimum over all states. Furthermore, in~\eqref{variational norm (iii)}, we applied the variational form of the norm, and in the last step~\eqref{additivity conditional (iii)} we used the additivity of the conditional entropy. We then multiply by $1/(1-\alpha)$ and take the limit $n \rightarrow \infty$ on both sides. Both expressions converge to the conditional entropy as the term $O(\log{n}/n)$ vanishes. Therefore, the quantity involving the universal states in~\eqref{norm lower bound 1 (iii)} must equal the conditional entropy. 
\end{proof}

\noindent{\bf Notation.} We employ the operator-vector identification, with each Hilbert space assigned its computational basis, denoted by $|i\rangle_A$ for space $A$, and similarly for others. The linear map $\Op_{A \rightarrow B}$ is defined through its action on the computational basis as
\begin{align}\label{operator-vector}
    \Op \!\!\,_{A \rightarrow B} (|i\rangle_A \otimes |j\rangle_B) = |j\rangle_B \langle i|_A.
\end{align}

Next, we derive the alternative variational forms that involve Schatten norms.
\begin{lemma}
	\label{Dupuis variational}
	Let $\ket{\rho}_{ABCD}$ be a normalized pure state, $X_{AD \rightarrow BC}=\Op_{AD \rightarrow BC} (\ket{\rho}_{ABCD})$, and $(\alpha,z,\lambda)\in \D$. Then, 
	\begin{align}
		 (i) \qquad &H^\lambda_{\alpha,z}(AB|C)_\rho = -\log{\opt_{\sigma_C}\opt_{\tau_D} \left\|  \sigma_C^{\lambda\frac{1-\alpha}{2z}}  \rho_C^{\frac{1-\lambda}{2}\frac{1-\alpha}{z}} X_{AD \rightarrow BC} \rho_D^{\top \frac{1-\hat{\lambda}}{2}\frac{1-\hat{\alpha}}{\hat{z}}} \tau_D^{\top \frac{\hat{\lambda}}{2}\frac{1-\hat{\alpha}}{\hat{z}}} \right\|_2^{\frac{2z}{\alpha-1}}}; \\
		(ii) \qquad &H^\lambda_{\alpha,z}(B|C)_\rho = -\log{\opt_{\sigma_C} \left\| \sigma_C^{\lambda \frac{1-\alpha}{2z}} \rho_C^{\frac{1-\lambda}{2}\frac{1-\alpha}{z}}  X_{AD \rightarrow BC} \rho_{AD}^{\frac{1-\hat{\lambda}}{2}\frac{1-\hat{\alpha}}{\hat{z}}} \right\|_{2z}^{\frac{2z}{\alpha-1}}}; \\
		(iii) \qquad &H^\lambda_{\alpha,z}(A|BC)_\rho 
		=- \log{\opt_{\tau_D} \left\| \rho_{BC}^{\frac{1-\lambda}{2}\frac{1-\alpha}{z}} X_{AD \rightarrow BC} \rho_D^{\top\frac{1-\hat{\lambda}}{2}\frac{1-\hat{\alpha}}{\hat{z}}} \tau_D^{\top\frac{\hat{\lambda}}{2}\frac{1-\hat{\alpha}}{\hat{z}}}  \right\|_{2\hat{z}}^{{\frac{2z}{\alpha-1}}}}.
	\end{align}
 Here, we use the hats to denote the dual parameters defined in Theorem~\ref{Duality} (see also Remark~\ref{explicit computation} for an explicit computation).
\end{lemma}
\begin{proof}
	Let us consider the variational form involving the purification $\ket{\rho}_{ABC}$ of $\rho_{AB}$ stated in Remark~\ref{variational purification}
 \begin{equation}
		H^\lambda_{\alpha,z}(A|B)_{\rho}= \frac{z}{1-\alpha}\log{\opt_{\sigma_B}\opt_{\tau_C} \bigbra{\rho} \mathcal{G}_{\lambda,\frac{1-\alpha}{z}}(\rho_B,\sigma_B)  \otimes \mathcal{G}_{\hat{\lambda},\frac{1-\hat{\alpha}}{\hat{z}}}(\rho_{C},\tau_{C})\bigket{\rho}_{ABC}}.
	\end{equation}
Using a similar approach to that in~\cite[Lemma 6]{dupuis2015chain}, we can express the formula in terms of a Schatten norm. Explicitly, 
	\begin{align}
		&2^{-H^\lambda_{\alpha,z}(AB|C)_{\rho}} \\ 
		&\quad =\opt_{\sigma_C}\opt_{\tau_D}\bigbra{\rho}  \mathcal{G}_{\lambda,\frac{1-\alpha}{z}}(\rho_C,\sigma_C)  \otimes \mathcal{G}_{\hat{\lambda},\frac{1-\hat{\alpha}}{\hat{z}}}(\rho_{D},\tau_{D}) \bigket{\rho}_{ABCD}^\frac{z}{\alpha-1}\\
		&\quad =\opt_{\sigma_C} \opt_{\tau_D}\left\| \mathcal{G}_{\lambda,\frac{1-\alpha}{z}}(\rho_C,\sigma_C)^\frac{1}{2}  \otimes \mathcal{G}_{\hat{\lambda},\frac{1-\hat{\alpha}}{\hat{z}}}(\rho_{D},\tau_{D})^\frac{1}{2} \bigket{\rho}_{ABCD} \right\|^\frac{2z}{\alpha-1} \\
  \label{First Lemma 11}
		&\quad =\opt_{\sigma_C} \opt_{\tau_D} \left\|\Op \!\,_{AD \rightarrow BC}\Big(\mathcal{G}_{\lambda,\frac{1-\alpha}{z}}(\rho_C,\sigma_C)^\frac{1}{2}  \otimes \mathcal{G}_{\hat{\lambda},\frac{1-\hat{\alpha}}{\hat{z}}}(\rho_{D},\tau_{D})^\frac{1}{2}\bigket{\rho}_{ABCD}\Big)\right\|_2^{\frac{2z}{\alpha-1}} \\
    \label{Then Lemma 12}
		&\quad = \opt_{\sigma_C} \opt_{\tau_D} \left\| \mathcal{G}_{\lambda,\frac{1-\alpha}{z}}(\rho_C,\sigma_C)^\frac{1}{2} \Op \!\,_{AD \rightarrow BC} \big(\ket{\rho}_{ABCD}\big)  \mathcal{G}_{\hat{\lambda},\frac{1-\hat{\alpha}}{\hat{z}}}(\rho_D,\tau_D)^{\top\frac{1}{2}} \right\|_2^{\frac{2z}{\alpha-1}}\\
		&\quad = \opt_{\sigma_C}\opt_{\tau_D} \Big(\Trm\Big[\rho_C^{\frac{1-\lambda}{2}\frac{1-\alpha}{z}} \sigma_C^{\lambda\frac{1-\alpha}{z}}  \rho_C^{\frac{1-\lambda}{2}\frac{1-\alpha}{z}} X_{AD \rightarrow BC} \rho_D^{\top\frac{1-\hat{\lambda}}{2}\frac{1-\hat{\alpha}}{\hat{z}}} \tau_D^{\top\hat{\lambda}\frac{1-\hat{\alpha}}{\hat{z}}}  \rho_D^{\top\frac{1-\hat{\lambda}}{2}\frac{1-\hat{\alpha}}{\hat{z}}} X^\dagger_{AD \rightarrow BC} \Big]\Big)^\frac{z}{\alpha-1}
	\end{align}
	where $\Op_{AD \rightarrow BC} (\ket{\rho}_{ABCD}) = X_{AD \rightarrow BC}$. In \eqref{First Lemma 11}, we applied Lemma~\ref{Dupuis Lemma 12}, and in \eqref{Then Lemma 12}, we used Lemma~\ref{Dupuis Lemma 11}. Hence, we have
	\begin{equation}
		H^\lambda_{\alpha,z}(AB|C)_{\rho} = -\log{\opt_{\sigma_C}\opt_{\tau_D} \left\|  \sigma_C^{\lambda\frac{1-\alpha}{2z}}  \rho_C^{\frac{1-\lambda}{2}\frac{1-\alpha}{z}} X_{AD \rightarrow BC} \rho_D^{\top \frac{1-\hat{\lambda}}{2}\frac{1-\hat{\alpha}}{\hat{z}}} \tau_D^{\top \frac{\hat{\lambda}}{2}\frac{1-\hat{\alpha}}{\hat{z}}} \right\|_2^{2\frac{z}{\alpha-1}}}.
	\end{equation}
	Next, we derive a variational expression for $H^\lambda_{\alpha,z}(B|C)$. Similarly, we obtain 
	\begin{align}
		2^{-H^\lambda_{\alpha,z}(B|C)_{\rho}} &= \opt_{\sigma_C}\opt_{\tau_{AD}} \left\| \sigma_C^{\lambda\frac{1-\alpha}{2z}}  \rho_C^{\frac{1-\lambda}{2}\frac{1-\alpha}{z}} X_{AD \rightarrow BC} \rho_{AD}^{\frac{1-\hat{\lambda}}{2}\frac{1-\hat{\alpha}}{\hat{z}}} \tau_{AD}^{\frac{\hat{\lambda}}{2}\frac{1-\hat{\alpha}}{\hat{z}}} \right\|_2^{2\frac{z}{\alpha-1}} \\
		&= \opt_{\sigma_C}\opt_{\tau_{AD}} \Tr(X^\dagger_{AD \rightarrow BC} \rho_C^{\frac{1-\lambda}{2}\frac{1-\alpha}{z}}  \sigma_C^{\lambda\frac{1-\alpha}{z}}  \rho_C^{\frac{1-\lambda}{2}\frac{1-\alpha}{z}}  X_{AD \rightarrow BC} \rho_{AD}^{\frac{1-\hat{\lambda}}{2}\frac{1-\hat{\alpha}}{\hat{z}}} \tau_{AD}^{\hat{\lambda}\frac{1-\hat{\alpha}}{\hat{z}}}  \rho_{AD}^{\frac{1-\hat{\lambda}}{2}\frac{1-\hat{\alpha}}{\hat{z}}})^{\frac{z}{\alpha-1}} \\
		&= \opt_{\sigma_C} \left\| \rho_{AD}^{\frac{1-\hat{\lambda}}{2}\frac{1-\hat{\alpha}}{\hat{z}}} X^\dagger_{AD \rightarrow BC} \rho_C^{\frac{1-\lambda}{2}\frac{1-\alpha}{z}}  \sigma_C^{\lambda\frac{1-\alpha}{z}}  \rho_C^{\frac{1-\lambda}{2}\frac{1-\alpha}{z}}  X_{AD \rightarrow BC} \rho_{AD}^{\frac{1-\hat{\lambda}}{2}\frac{1-\hat{\alpha}}{\hat{z}}}\right\|_z^{\frac{z}{\alpha-1}}\\
		&= \opt_{\sigma_C} \left\| \sigma_C^{\lambda \frac{1-\alpha}{2z}} \rho_C^{\frac{1-\lambda}{2}\frac{1-\alpha}{z}}  X_{AD \rightarrow BC} \rho_{AD}^{\frac{1-\hat{\lambda}}{2}\frac{1-\hat{\alpha}}{\hat{z}}} \right\|_{2z}^{2\frac{z}{\alpha-1}},
	\end{align}
	where we used that $\hat{\lambda}\frac{1-\hat{\alpha}}{\hat{z}} = 1-\frac{1}{z}$. Finally, using the latter result, and using the norm invariance under transpose, we have 
	\begin{align}
		H^\lambda_{\alpha,z}(A|BC)_{\rho} 
		&=-H^{\hat{\lambda}}_{\hat{\alpha},\hat{z}}(A|D)_{\rho} \\
		&=- \log{\opt_{\tau_D} \left\| \rho_{BC}^{\frac{1-\lambda}{2}\frac{1-\alpha}{z}} X_{AD \rightarrow BC} \rho_D^{\top\frac{1-\hat{\lambda}}{2}\frac{1-\hat{\alpha}}{\hat{z}}} \tau_D^{\top\frac{\hat{\lambda}}{2}\frac{1-\hat{\alpha}}{\hat{z}}}  \right\|_{2\hat{z}}^{{2\frac{z}{\alpha-1}}}}.
	\end{align}
The proof is complete.
\end{proof}

The combination of Proposition \ref{H with universal any z} and Lemma~\ref{Dupuis variational} leads to the following two lemmas. 
\begin{lemma}
	\label{regularized z geq 1}
	Let $(\alpha,z,\lambda)\in \D$, $\ket{\rho}_{ABCD}$ be a pure quantum state and $\Op_{AD \rightarrow BC} (\ket{\rho}^{\otimes n}_{ABCD}) = X_{A^n D^n \rightarrow B^n C^n}$. Then, 
	\begin{equation}
 \label{H(B|C) (i)}
		H^\lambda_{\alpha,z}(B|C)_{\rho} = \lim_{n \rightarrow \infty}- \frac{1}{n}\log{ \left\| \omega_{C^n}^{\frac{\lambda}{2} \frac{1-\alpha}{z}} \rho_{C^n}^{ \; \frac{1-\lambda}{2}\frac{1-\alpha}{z}}  X_{A^n D^n \rightarrow B^n C^n} \rho_{A^n D^n}^{ \;\frac{1-\hat{\lambda}}{2}\frac{1-\hat{\alpha}}{\hat{z}}} \right\|_{2z}^{2\frac{z}{\alpha-1}}}.
	\end{equation}
	Moreover, for $z \geq 1$,
	\begin{align}
 \label{H(AB|C) (i)}
		&H^\lambda_{\alpha,z}(AB|C)_{\rho} = \lim_{n \rightarrow \infty} - \frac{1}{n} {\frac{2z}{\alpha-1}} \log{\sup_{\tau_{D^n}} \left\|  \omega_{C^n}^{\frac{\lambda}{2}\frac{1-\alpha}{z}}  \rho_{C^n}^{ \;\frac{1-\lambda}{2}\frac{1-\alpha}{z}} X_{A^n D^n \rightarrow B^n C^n} \rho_{D^n}^{ \;\top \frac{1-\hat{\lambda}}{2}\frac{1-\hat{\alpha}}{\hat{z}}} \tau_{D^n}^{\top \frac{\hat{\lambda}}{2}\frac{1-\hat{\alpha}}{\hat{z}}} \right\|_2},\\
  \label{H(A|BC) (i)}
		&H^\lambda_{\alpha,z}(A|BC)_{\rho} 
		=\lim_{n \rightarrow \infty} - \frac{1}{n}\frac{2z}{\alpha-1}\log{\sup_{\tau_{D^n}} \left\| \rho_{B^n C^n}^{ \;\frac{1-\lambda}{2}\frac{1-\alpha}{z}} X_{A^n D^n \rightarrow B^n C^n} \rho_{D^n}^{ \;\top\frac{1-\hat{\lambda}}{2}\frac{1-\hat{\alpha}}{\hat{z}}} \tau_{D^n }^{\top\frac{\hat{\lambda}}{2}\frac{1-\hat{\alpha}}{\hat{z}}}  \right\|_{2\hat{z}}}.
	\end{align}
 Here, we denoted $\rho_{C^n} = \rho_C^{\otimes n}$, $\rho_{D^n} = \rho_D^{\otimes n}$, $\rho_{B^n C^n } = \rho_{BC}^{\otimes n}$ and $\rho_{A^n D^n } = \rho_{AD}^{\otimes n}$.
  Moreover, we use the hats to denote the dual parameters defined in Theorem~\ref{Duality} (see also Remark~\ref{explicit computation} for an explicit computation).
\end{lemma}
\begin{proof}
The result for the expressions~\eqref{H(B|C) (i)} and ~\eqref{H(AB|C) (i)} follows from a combination of Lemma~\ref{Dupuis variational} and Proposition~\ref{H with universal any z}. 
 Specifically, we repeat the same steps as in the proof of Lemma~\ref{Dupuis variational} for $\ket{\rho}_{ABCD}^{\otimes n}$, but instead of optimizing over the marginal $\sigma_C^{\otimes n}$, we replace it with the permutation-invariant state $\omega_{C^n}$ by virtue of the form (i) in Proposition~\ref{H with universal any z}.

For the last expression~\eqref{H(A|BC) (i)}, we apply the additivity of the conditional entropy with its form as stated in~\ref{Dupuis variational} (iii).
\end{proof}

\begin{lemma}
	\label{regularized z leq 1}
	Let $(\alpha,z,\lambda)\in \D$, $\ket{\rho}_{ABCD}$ be a normalized pure state and $\Op_{AD \rightarrow BC} (\ket{\rho}^{\otimes n}_{ABCD}) = X_{A^n D^n \rightarrow B^n C^n}$. Then, we have
	\begin{equation}
 \label{H(A|BC) (ii)}
		H^\lambda_{\alpha,z}(A|BC)_{\rho} 
		=\lim_{n \rightarrow \infty} - \frac{1}{n}\log{ \left\| \rho_{B^n C^n}^{ \;\frac{1-\lambda}{2}\frac{1-\alpha}{z}} X_{A^n D^n \rightarrow B^n C^n} \rho_{D^n}^{ \;\top\frac{1-\hat{\lambda}}{2}\frac{1-\hat{\alpha}}{\hat{z}}} \omega_{D^n }^{\top\frac{\hat{\lambda}}{2}\frac{1-\hat{\alpha}}{\hat{z}}}  \right\|_{2\hat{z}}^{{2\frac{z}{\alpha-1}}}} \,.
	\end{equation}
	For $z \leq 1$, $\lambda (1-\alpha) \geq 0$, we have
	\begin{align}
 \label{H(AB|C) (ii)}
		&H^\lambda_{\alpha,z}(AB|C)_{\rho} = \lim_{n \rightarrow \infty} - \frac{1}{n} \frac{2z}{\alpha-1} \log{\sup_{\sigma_{C}} \left\|  (\sigma_C^{\otimes n})^{\;\lambda\frac{1-\alpha}{2z}}  \rho_{C^n}^{ \;\frac{1-\lambda}{2}\frac{1-\alpha}{z}} X_{A^n D^n \rightarrow B^n C^n} \rho_{D^n}^{ \;\top \frac{1-\hat{\lambda}}{2}\frac{1-\hat{\alpha}}{\hat{z}}} \omega_{D^n}^{\top \frac{\hat{\lambda}}{2}\frac{1-\hat{\alpha}}{\hat{z}}} \right\|_2}, \\
  \label{H(B|C) (ii)}
		&H^\lambda_{\alpha,z}(B|C)_{\rho} = \lim_{n \rightarrow \infty}- \frac{1}{n}\frac{2z}{\alpha-1}\log{ \sup_{\sigma_C}\left\| (\sigma_C^{\otimes n})^{ \;\lambda \frac{1-\alpha}{2z}} \rho_{C^n}^{ \; \frac{1-\lambda}{2}\frac{1-\alpha}{z}}  X_{A^n D^n \rightarrow B^n C^n} \rho_{A^n D^n}^{ \;\frac{1-\hat{\lambda}}{2}\frac{1-\hat{\alpha}}{\hat{z}}} \right\|_{2z}}.
	\end{align}
	For $z \leq 1$, $\lambda (1-\alpha) \leq 0$, we have
	\begin{align}
 \label{H(AB|C) (iii)}
		&H^\lambda_{\alpha,z}(AB|C)_{\rho} = \lim_{n \rightarrow \infty} - \frac{1}{n} \log{ \left\|  \omega_{C^n}^{\lambda\frac{1-\alpha}{2z}}  \rho_{C^n}^{ \;\frac{1-\lambda}{2}\frac{1-\alpha}{z}} X_{A^n D^n \rightarrow B^n C^n} \rho_{D^n}^{ \;\top \frac{1-\hat{\lambda}}{2}\frac{1-\hat{\alpha}}{\hat{z}}} \omega_{D^n}^{\top \frac{\hat{\lambda}}{2}\frac{1-\hat{\alpha}}{\hat{z}}} \right\|_2^{2\frac{z}{\alpha-1}}}.
	\end{align}
 Here, we denoted $\rho_{C^n} = \rho_C^{\otimes n}$, $\rho_{D^n} = \rho_D^{\otimes n}$, $\rho_{B^n C^n } = \rho_{BC}^{\otimes n}$ and $\rho_{A^n D^n } = \rho_{AD}^{\otimes n}$.
  Moreover, we use the hats to denote the dual parameters defined in Theorem~\ref{Duality} (see also Remark~\ref{explicit computation} for an explicit computation).
\end{lemma}
\begin{proof}
    The proof follows a similar structure to that of Lemma~\ref{regularized z geq 1}. However, when repeating the same steps as in the proof of Lemma~\ref{Dupuis variational} for $\ket{\rho}_{ABCD}^{\otimes n}$, a different selection of the permutation-invariant states must be due to the different parameter ranges. 

    In particular, for the expressions~\eqref{H(A|BC) (ii)} we use the form (i) in Proposition~\ref{H with universal any z}. For~\eqref{H(AB|C) (ii)} we use the form (ii) in Proposition~\ref{H with universal any z}. For~\eqref{H(AB|C) (iii)} we use the form (iii) in Proposition~\ref{H with universal any z} with two universal states. Finally,~\eqref{H(B|C) (ii)} simply follows by additivity of the conditional entropy.
\end{proof}

\section{Chain rules}
\label{sec:chainrules}

In this section, we derive chain rules for the conditional entropy $H^\lambda_{\alpha,z}(A|B)$. We integrate the formulations derived in the previous section with a result from complex interpolation theory for Schatten norms, with the latter technique being similar to that in \cite{dupuis2015chain}. However, unlike the simpler case addressed in \cite{dupuis2015chain}, this context presents further challenges. Notably, results from complex interpolation theory for multivariate trace inequalities often require optimizations over unitaries, particularly in the non-commuting case involving more than two matrices~\cite{sutter16}. However, the expressions derived earlier utilize universal permutation-invariant states, effectively eliminating the usual requirement for optimization over unitaries. Indeed, the universal state commutes with any permutation-invariant state, enabling us to asymptotically reduce the problem to one involving only two non-commuting matrices.

\begin{theorem}[Chain rule]
	\label{Final chain rule}
	Let $(\alpha,z,\lambda), (\beta,w,\mu), (\gamma,v,\lambda) \in \D$ and $\rho_{ABC}$ be a tripartite quantum state. Suppose the parameters satisfy the relations
	\begin{equation}\label{conditions}
		\frac{z}{1-\alpha}=\frac{w}{1-\beta}+\frac{v}{1-\gamma}, \qquad \frac{1-z}{1-\alpha} = \frac{1-w}{1-\beta},\qquad \mu=\frac{1-v}{1-\gamma}.
	\end{equation}
	If $(\alpha-1)(\beta-1)(\gamma-1) > 0$, the following chain rule holds:
	\begin{equation}
		H^\lambda_{\alpha,z}(AB|C)_\rho\geq	H_{\beta,w}^\mu(A|BC)_\rho+	H^\lambda_{\gamma,v}(B|C)_\rho.
	\end{equation}
	The inequality is reversed when $(\alpha-1)(\beta-1)(\gamma-1)<0$.
\end{theorem}

We divide the full proof into several cases. The first case considers $\alpha>1, (\beta-1)(\gamma-1)>0$, while the second case addresses scenarios where \(\alpha > 1, (\beta-1)(\gamma-1)<0\). For the first case, we further divide the proof into subcases: \(z \geq 1\) and \(z \leq 1\). The arguments for the second case are analogous, so we outline only the key steps. Finally, the case \(\alpha < 1\) can be handled by duality.

\subsection{Proof of the case $\alpha > 1$ with $(\beta-1)(\gamma-1) > 0$}
Note that from the first relation among the parameters for the chain rule in equation~\eqref{conditions}, we need to consider only the case \(\alpha, \beta, \gamma > 1\) as the case \(\alpha>1 \) and \(\beta, \gamma < 1\) cannot happen.

Let us first consider the case $z\geq 1$. Note that since $\alpha,\beta > 1$, the second condition in~\eqref{conditions} implies that $w \geq 1$. We aim to prove that
\begin{equation}
	H^\lambda_{\alpha,z}(AB|C)_{\rho}\geq	H_{\beta,w}^\mu(A|BC)_{\rho}+	H^\lambda_{\gamma,v}(B|C)_{\rho}.
\end{equation}
Using the asymptotic forms in Lemma~\ref{regularized z geq 1}, it is enough to prove that 
\begin{align}\label{to prove geq 1}
& \sup_{\tau_{D^n}} \left\|  \omega_{C^n}^{\frac{\lambda}{2}\frac{1-\alpha}{z}}  \rho_{C^n}^{\;\frac{1-\lambda}{2}\frac{1-\alpha}{z}} X_{A^nD^n \rightarrow B^nC^n} \rho_{D^n}^{ \;\top \frac{1-\hat{\lambda}}{2}\frac{1-\hat{\alpha}}{\hat{z}}} \tau_{D^n}^{\top \frac{\hat{\lambda}}{2}\frac{1-\hat{\alpha}}{\hat{z}}} \right\|_2^{2\frac{z}{\alpha-1}} \\
&\qquad\qquad  \leq \sup_{\tau_{D^n}} \left\|  \rho_{B^n C^n}^{ \;\frac{1-\mu}{2}\frac{1-\beta}{w}} X_{A^nD^n \rightarrow B^nC^n} \rho_{D^n}^{ \;\top \frac{1-\hat{\mu}}{2}\frac{1-\hat{\beta}}{\hat{w}}} \tau_{D^n}^{\top\frac{\hat{\mu}}{2}\frac{1-\hat{\beta}}{\hat{w}}}  \right\|_{2\hat{w}}^{{2\frac{w}{\beta-1}}} \\
&\qquad\qquad\qquad \cross\left\|    \omega_{C^n}^{ \frac{\lambda}{2}\frac{1-\gamma}{v}} \rho_{C^n}^{\; \frac{1-\lambda}{2}\frac{1-\gamma}{v}}  X_{A^n D^n \rightarrow B^nC^n} \rho_{A^n D^n}^{\;\frac{1-\kappa}{2}\frac{1-\hat{\gamma}}{\hat{v}}} \right\|_{2v}^{2\frac{v}{\gamma-1}},
\end{align}
where $\kappa = (1-v)/(1-\gamma)$.
The result then follows by taking the logarithm on both sides, dividing by $1/n$, and taking the limit $n \rightarrow \infty$.

We define the ratios $r_1=(\alpha-1)/z$, $r_2=(\beta-1)/w$, $r_3=(\gamma-1)/v$.
We then use Lemma~\ref{Hadamard} with $\theta:=\frac{r_1}{r_3}$, and $p_{\theta} = 2$, $p_0 = 2 \hat{w}$ and $p_1=2v$. Note that the condition $p_0,p_1 \geq 1$ is always satisfied when the parameters lie within the DPI region $\D$. To apply the theorem, we need to verify the relation
\begin{equation}
\frac{1-\theta}{p_0}+\frac{\theta}{p_1} = \frac{1}{p_{\theta}}.
\end{equation}
We use the conditions in~\eqref{conditions} and the duality relation $1/\hat{w} = 1+ \mu r_2$ to establish this. The first conditions of~\eqref{conditions} implies that $1-r_1/r_3 = r_1/r_2 $. Moreover $\mu = (1-v)/(1-\gamma)$ implies that $1/v = 1-\mu r_3$. 
Using these relations, we obtain 
\begin{align}
\frac{1-\theta}{p_0}+\frac{\theta}{p_1} &= \frac{r_1}{2} \left(\frac{1}{r_2 \hat{w}}+\frac{1}{r_3 v}\right) \\
&= \frac{r_1}{2} \left(\frac{1}{r_2}+ \mu + \frac{1}{r_3} - \mu \right) \\
&=  \frac{1}{2}.
\end{align}
Furthermore, the third conditions with the relation $\mu = (1-v)/(1-\gamma)$ yields $(z-1)/(1-\alpha) = (w-1)/(\beta-1)$ which is equivalent to $\hat{\lambda}=\hat{\mu}$.
Let us define
\begin{equation}
F(s):= \omega_{C^n}^{ -\lambda  \frac{r_3}{2}s} \rho_{C^n}^{ \; -\frac{1-\lambda}{2}r_3 s} \rho_{B^n C^n}^{ \;-\frac{1-\mu}{2}r_2(1-s)+\frac{1-\mu}{2}r_3 s} X_{A^n D^n \rightarrow B^nC^n} \rho_{D^n}^{ \top \frac{1-\hat{\lambda}}{2} r_2 (1-s)} \tau_{D^n}^{ \top\frac{\hat{\lambda}}{2}r_2 (1-s)},
\end{equation}
where $\hat{\lambda} = (1-z)/(1-\alpha)$ and $\hat{\lambda} = (1-v)/(1-\gamma)$. Recall also that
$(1-\hat{\alpha})/\hat{z} = -(1-\alpha)/z$.
We have 
\begin{equation}
F(\theta) =  \omega_{C^n}^{-\lambda \frac{r_1}{2}} \rho_{C^n}^{\; -\frac{1-\lambda}{2}\frac{r_1}{2}} X_{A^n D^n  \rightarrow B^n C^n} \rho_{D^n}^{\top\frac{1-\hat{\lambda}}{2} r_1} \tau_{D^n}^{\top\frac{\hat{\lambda}}{2}r_1},
\end{equation}
since $-\frac{1-\mu}{2}r_2(1-\frac{r_1}{r_3})+\frac{r_1}{2} (1-\mu) = 0$. Moreover, we have, up to external unitaries that do not change the norm (and $\hat{\lambda}=\hat{\mu}$),
\begin{equation}
F(it)=  \rho_{B^n C^n}^{ -\frac{1-\mu}{2}r_2} X_{A^n D^n \rightarrow B^n C^n} \rho_{D^n}^{\top\frac{1-\hat{\mu}}{2} r_2} \rho_{D^n}^{\top\frac{1-\hat{\mu}}{2} r_2 it} \tau_{D^n}^{\top\frac{\hat{\mu}}{2}r_2}.
\end{equation}
For non-full-rank states, we can consider a slightly depolarized version, take the limit in the chain rule inequality, and then apply the continuity result from Proposition~\ref{Continuity}.
Additionally, up to external unitaries
\begin{align}
F(1+it)&= \omega_{C^n}^{ -\lambda \frac{r_3}{2}}  \rho_{C^n}^{\; -\frac{1-\lambda}{2}\frac{r_3}{2} it} \rho_{C^n}^{ \;-\frac{1-\lambda}{2}\frac{r_3}{2}} \rho_{B^n C^n}^{ \;\frac{1-\hat{\theta}}{2}r_3} X_{A^n D^n \rightarrow B^n C^n}\\ 
&= \omega_{C^n}^{-\lambda \frac{r_3}{2}} \rho_{C^n}^{\; -\frac{1-\lambda}{2}\frac{r_3}{2} it} \rho_{C^n}^{ \;-\frac{1-\lambda}{2}\frac{r_3}{2}} X_{A^n D^n \rightarrow B^n C^n} \rho_{A^n D^n}^{\;\frac{1-\kappa}{2}r_3},
\end{align}
where we used that $\kappa=\mu$ and $\rho_{B^n C^n}^q X_{A^n D^n  \rightarrow B^n C^n} = X_{A^n D^n  \rightarrow B^n C^n} \rho_{A^n D^n }^q$ for  some power $q$. We note that the terms $F(it)$ and $F(1+it)$ contain undesired phases.
However, we now show that the optimization over all states allows us to remove them.  We then apply Hadamard's three-line theorem for Schatten norms, as stated in Lemma~\ref{Hadamard}, which asserts that
\begin{equation}
\|F(\theta)\|_{p_{\theta}} \leq \sup_{t_1 \in \mathbb{R}}\|F(it_1)\|_{p_0}^{1-\theta} \sup_{t_2 \in \mathbb{R}}\|F(1+it_2)\|_{p_1}^{\theta}.
\end{equation}
In our case, this means that
\begin{align}
& \left\|  \omega_{C^n}^{\frac{\lambda}{2}\frac{1-\alpha}{z}}  \rho_{C^n}^{\;\frac{1-\lambda}{2}\frac{1-\alpha}{z}} X_{A^nD^n \rightarrow B^nC^n} \rho_{D^n}^{ \;\top \frac{1-\hat{\lambda}}{2}\frac{1-\hat{\alpha}}{\hat{z}}} \tau_{D^n}^{\top \frac{\hat{\lambda}}{2}\frac{1-\hat{\alpha}}{\hat{z}}} \right\|_2^{2\frac{z}{\alpha-1}} \\
&\quad \qquad  \leq  \sup_{t_1 \in \mathbb{R}} \left\|  \rho_{B^n C^n}^{ \;\frac{1-\mu}{2}\frac{1-\beta}{w}} X_{A^nD^n \rightarrow B^nC^n} \rho_{D^n}^{ \;\top \frac{1-\hat{\mu}}{2}\frac{1-\hat{\beta}}{\hat{w}}} \rho_{D^n}^{ \top\frac{1-\hat{\mu}}{2} r_2 it_1} \tau_{D^n}^{\top\frac{\hat{\mu}}{2}\frac{1-\hat{\beta}}{\hat{w}}}  \right\|_{2\hat{w}}^{{2\frac{w}{\beta-1}}} \\
& \quad \qquad \qquad \cross \sup_{t_2 \in \mathbb{R}}  \left\|    \omega_{C^n}^{ \frac{\lambda}{2}\frac{1-\gamma}{v }} \rho_{C^n}^{ -\frac{1-\lambda}{2}\frac{r_3}{2} it_2} \rho_{C^n}^{\; \frac{1-\lambda}{2}\frac{1-\gamma}{2v}}  X_{A^n D^n \rightarrow B^nC^n} \rho_{A^n D^n}^{\;\frac{1-\kappa}{2}\frac{1-\hat{\gamma}}{\hat{v}}} \right\|_{2v}^{2\frac{v}{\gamma-1}}.
\end{align}
We then take the supremum over $\tau_{D^n}$  on both sides. Note that we can exchange the latter supremums with the supremum over $t_1$. In addition,
since the optimizations are performed over all the states, we can drop the phases. For example, 
\begin{align}
&\sup_{\tau_{D^n}} \sup_{t_1 \in \mathbb{R}} \left\|  \rho_{B^n C^n}^{ \;\frac{1-\mu}{2}\frac{1-\beta}{w}} X_{A^nD^n \rightarrow B^nC^n} \rho_{D^n}^{ \;\top \frac{1-\hat{\mu}}{2}\frac{1-\hat{\beta}}{\hat{w}}} \rho_{D^n}^{ \top\frac{1-\hat{\mu}}{2} r_2 it_1} \tau_{D^n}^{\top\frac{\hat{\mu}}{2}\frac{1-\hat{\beta}}{\hat{w}}}  \right\|_{2\hat{w}}^{{2\frac{w}{\beta-1}}} \\
&\qquad \qquad \qquad = \sup_{t_1 \in \mathbb{R}} \sup_{\tau_{D^n}}  \left\|  \tau_{D^n}^{\frac{\hat{\mu}}{2}\frac{1-\hat{\beta}}{\hat{w}}} \rho_{D^n}^{ \frac{1-\hat{\mu}}{2} r_2 it_1} \rho_{D^n}^{ \; \frac{1-\hat{\mu}}{2}\frac{1-\hat{\beta}}{\hat{w}}}  X_{B^nC^n \rightarrow A^nD^n} \rho_{B^n C^n}^{ \;\frac{1-\mu}{2}\frac{1-\beta}{w}} \right\|_{2\hat{w}}^{{2\frac{w}{\beta-1}}} \\
&\qquad \qquad \qquad = \sup_{t_1 \in \mathbb{R}} \sup_{\tau_{D^n}}  \left\|  \tau_{D^n}^{\frac{\hat{\mu}}{2}\frac{1-\hat{\beta}}{\hat{w}}} \rho_{D^n}^{ \; \frac{1-\hat{\mu}}{2}\frac{1-\hat{\beta}}{\hat{w}}}  X_{B^nC^n \rightarrow A^nD^n} \rho_{B^n C^n}^{ \;\frac{1-\mu}{2}\frac{1-\beta}{w}} \right\|_{2\hat{w}}^{{2\frac{w}{\beta-1}}} \\
&\qquad \qquad \qquad =  \sup_{\tau_{D^n}}  \left\|  \tau_{D^n}^{\frac{\hat{\mu}}{2}\frac{1-\hat{\beta}}{\hat{w}}} \rho_{D^n}^{ \; \frac{1-\hat{\mu}}{2}\frac{1-\hat{\beta}}{\hat{w}}}  X_{B^nC^n \rightarrow A^nD^n} \rho_{B^n C^n}^{ \;\frac{1-\mu}{2}\frac{1-\beta}{w}} \right\|_{2\hat{w}}^{{2\frac{w}{\beta-1}}} \,,
\end{align}
where we noted that the transpose does not change the norm. Finally, the phase in $F(1+it)$ can be simplified since the universal state $\omega_{C^n}$ commutes with $\rho_{C^n} = \rho_C^{\otimes n}$.  Therefore, we obtain the inequality~\eqref{to prove geq 1} which is what we wanted to prove.

In the case where \(z \leq 1\), the proof follows a similar approach to the \(z \geq 1\) case discussed earlier. 
However, since the variational form of the norm for \(z \leq 1\) includes an additional infimum, 
we rely on the asymptotic formulations from Lemma \ref{regularized z leq 1}. 
The proof proceeds analogously, with the replacements \(\tau_{D^n} \rightarrow \omega_{D^n}\) 
and \(\omega_{C^n} \rightarrow \sigma_C^{\otimes n}\) for \(\lambda(1-\alpha) \geq 0\). 
For \(\lambda(1-\alpha) \leq 0\), we replace \(\tau_{D^n} \rightarrow \omega_{D^n}\).

\subsection{Proof of the case $\alpha > 1$ with $(\beta-1)(\gamma-1) < 0$}
In this case, we need to prove
\begin{equation}
	H^\lambda_{\alpha,z}(AB|C)_{\rho}\leq	H_{\beta,w}^\mu(A|BC)_{\rho}+	H^\lambda_{\gamma,v}(B|C)_{\rho}.
\end{equation}
The proof follows a similar structure to the case $\alpha > 1 , (\beta-1)(\gamma-1) > 0$, but the holomorphic function used in the interpolation inequality requires modification. Let us first consider the case $z\geq 1$. If $\beta > 1$, $\gamma < 1$, then the preceding inequality is tantamount to
\begin{align}
\label{to prove leq 1}
& \sup_{\tau_{D^n}} \left\|  \rho_{B^n C^n}^{ \;\frac{1-\mu}{2}\frac{1-\beta}{w}} X_{A^nD^n \rightarrow B^nC^n} \rho_{D^n}^{ \;\top \frac{1-\hat{\mu}}{2}\frac{1-\hat{\beta}}{\hat{w}}} \tau_{D^n}^{\top\frac{\hat{\mu}}{2}\frac{1-\hat{\beta}}{\hat{w}}}  \right\|_{2\hat{w}}^{{2\frac{w}{\beta-1}}}\\
&\qquad \qquad \leq \sup_{\tau_{D^n}} \left\|  \omega_{C^n}^{\frac{\lambda}{2}\frac{1-\alpha}{z}}  \rho_{C^n}^{\;\frac{1-\lambda}{2}\frac{1-\alpha}{z}} X_{A^nD^n \rightarrow B^nC^n} \rho_{D^n}^{ \;\top \frac{1-\hat{\lambda}}{2}\frac{1-\hat{\alpha}}{\hat{z}}} \tau_{D^n}^{\top \frac{\hat{\lambda}}{2}\frac{1-\hat{\alpha}}{\hat{z}}} \right\|_2^{2\frac{z}{\alpha-1}} \\
& \quad \qquad \qquad \qquad \cross  \left\|    \omega_{C^n}^{ \frac{\lambda}{2}\frac{1-\gamma}{v}} \rho_{C^n}^{\; \frac{1-\lambda}{2}\frac{1-\gamma}{v}}  X_{A^n D^n \rightarrow B^nC^n} \rho_{A^n D^n}^{\;\frac{1-\kappa}{2}\frac{1-\hat{\gamma}}{\hat{v}}} \right\|_{2v}^{-2\frac{v}{\gamma-1}},
\end{align}
where $\kappa = (1-v)/(1-\gamma)$. Then, we use 
\begin{equation}
F(s):= \omega_{C^n}^{ -\frac{\lambda}{2}  r_1\left(1+\frac{r_3}{r_2}s\right)} \rho_{C^n}^{ \; -\frac{1-\lambda}{2}  r_1\left(1+\frac{r_3}{r_2}s\right)} \rho_{B^n C^n}^{ \;\frac{1-\mu}{2}r_3 s} X_{A^n D^n \rightarrow B^nC^n} \rho_{D^n}^{ \top \frac{1-\hat{\lambda}}{2} r_1 (1-s)} \tau_{D^n}^{ \top\frac{\hat{\lambda}}{2}r_1 (1-s)} 
\end{equation}
with $p_0=2$, $p_1=2v$ and $p_\theta=2\hat{w}$. Moreover, $\theta=-r_2/r_3$.

If $\beta < 1$, $\gamma > 1$, then the preceding inequality is tantamount to
\begin{align}
\label{to prove leq 1 other} 
& \left\|    \omega_{C^n}^{ \frac{\lambda}{2}\frac{1-\gamma}{v}} \rho_{C^n}^{\; \frac{1-\lambda}{2}\frac{1-\gamma}{v}}  X_{A^n D^n \rightarrow B^nC^n} \rho_{A^n D^n}^{\;\frac{1-\kappa}{2}\frac{1-\hat{\gamma}}{\hat{v}}} \right\|_{2v}^{2\frac{v}{\gamma-1}}\\
&\qquad \qquad \leq \sup_{\tau_{D^n}} \left\|  \omega_{C^n}^{\frac{\lambda}{2}\frac{1-\alpha}{z}}  \rho_{C^n}^{\;\frac{1-\lambda}{2}\frac{1-\alpha}{z}} X_{A^nD^n \rightarrow B^nC^n} \rho_{D^n}^{ \;\top \frac{1-\hat{\lambda}}{2}\frac{1-\hat{\alpha}}{\hat{z}}} \tau_{D^n}^{\top \frac{\hat{\lambda}}{2}\frac{1-\hat{\alpha}}{\hat{z}}} \right\|_2^{2\frac{z}{\alpha-1}} \\
& \quad \qquad \qquad \qquad \cross \sup_{\tau_{D^n}} \left\|  \rho_{B^n C^n}^{ \;\frac{1-\mu}{2}\frac{1-\beta}{w}} X_{A^nD^n \rightarrow B^nC^n} \rho_{D^n}^{ \;\top \frac{1-\hat{\mu}}{2}\frac{1-\hat{\beta}}{\hat{w}}} \tau_{D^n}^{\top\frac{\hat{\mu}}{2}\frac{1-\hat{\beta}}{\hat{w}}}  \right\|_{2\hat{w}}^{{-2\frac{w}{\beta-1}}},
\end{align}
where $\kappa = (1-v)/(1-\gamma)$. Then, we use 
\begin{equation}
F(s):= \omega_{C^n}^{ -\frac{\lambda}{2}  r_1(1-s)} \rho_{C^n}^{ \; -\frac{1-\lambda}{2}  r_1(1-s)} \rho_{B^n C^n}^{ \;-\frac{1-\mu}{2}r_2 s} X_{A^n D^n \rightarrow B^nC^n} \rho_{D^n}^{ \top \frac{1-\hat{\lambda}}{2} r_1 \left(1+s\frac{r_2}{r_3}\right)} \tau_{D^n}^{ \top\frac{\hat{\lambda}}{2}r_1 \left(1+s\frac{r_2}{r_3}\right)} 
\end{equation}
with $p_0=2$, $p_1=2\hat{w}$ and $p_\theta=2v$. Moreover, $\theta=-r_3/r_2$.

The case $z \leq 1$ is similar with $\omega_{C^n} \leftrightarrow \sigma_C^{\otimes n}$ and $\tau_{D^n} \leftrightarrow \omega_{D^n}$ interchanged when needed.

\subsection{Proof of the case $\alpha < 1$}

The case for $\alpha < 1$ can be entirely derived by applying duality to the relations obtained previously. The parameter relations remain unchanged even when $\alpha < 1$. Now, let us apply the duality relation to the chain rule derived for the case where $\alpha > 1$. Consider a purification $\ket{\rho}_{ABCD}$. We then have the inequality
\begin{align}
    -H^{\hat{\lambda}}_{\hat{\alpha},\hat{z}}(AB|D)_{\rho} \geq -H_{\hat{\beta},\hat{w}}^{\hat{\mu}}(A|D)_{\rho} - H^{\hat{\theta}}_{\hat{\gamma},\hat{v}}(B|AD)_{\rho}.
\end{align}
By relabeling $A \leftrightarrow B$ and $D \leftrightarrow C$, this becomes
\begin{align}
    H^{\hat{\lambda}}_{\hat{\alpha},\hat{z}}(AB|C)_{\rho} \leq H_{\hat{\beta},\hat{w}}^{\hat{\mu}}(B|C)_{\rho} + H^{\hat{\theta}}_{\hat{\gamma},\hat{v}}(A|BC)_{\rho}.
\end{align}
The second relation in \eqref{conditions} shows that $\hat{\lambda} = \hat{\mu}$ for the original parameters. Therefore, as expected, $H^{\hat{\lambda}}_{\hat{\alpha},\hat{z}}(AB|C)$ and $H_{\hat{\beta},\hat{w}}^{\hat{\mu}}(B|C)$ must share the same third parameter.

Let us now verify that the relations among the parameters in \eqref{conditions} hold for the hatted parameters. The first relation in \eqref{conditions} leads to an analogous relation for the hatted parameters. The second relation follows from the fact that $\lambda = \frac{1-\hat{z}}{1-\hat{\alpha}} = \frac{1-\hat{v}}{1-\hat{\gamma}}$. To prove the third relation, we need to demonstrate that $\hat{\theta} = \frac{\hat{w}-1}{\hat{\beta}-1}$. This holds due to the third relation for the original parameters in \eqref{conditions}, as $\hat{\theta} = \frac{v-1}{\gamma-1} = \mu = \frac{\hat{w}-1}{\hat{\beta}-1}$.

\subsection{Implications for Petz and sandwiched R\'enyi conditional entropy}

Here we discuss special cases of the above chain rule that are in terms of the Petz and sandwiched conditional R\'enyi  entropy.

\begin{corollary} \label{col:chain1}
    Let $\alpha,\beta,\gamma$ lie in their respective DPI regions. If $(\alpha-1)(\beta-1)(\gamma-1) > 0$, then the following chain rules hold:
	\begin{align}
		&\widebar{H}^\downarrow_{\alpha}(AB|C)_{\rho}\geq	\widebar{H}^\downarrow_{\beta}(A|BC)_{\rho}+	\widebar{H}^\downarrow_{\gamma}(B|C)_{\rho}, \ \ \ \text{when} \ \ \ \frac{1}{1-\alpha}=\frac{1}{1-\beta}+\frac{1}{1-\gamma}; \\
        &\widetilde{H}^\uparrow_{\alpha}(AB|C)_{\rho}\geq	\widetilde{H}^\uparrow_{\beta}(A|BC)_{\rho}+	\widetilde{H}^\uparrow_{\gamma}(B|C)_{\rho}, \ \ \ \text{when} \ \ \ \frac{\alpha}{1-\alpha}=\frac{\beta}{1-\beta}+\frac{\gamma}{1-\gamma}.
	\end{align}
    In each inequality, the arrows on the entropies indexed by $\alpha$ and $\gamma$ may simultaneously be reversed under the same parameter relation. The inequalities themselves are reversed when $(\alpha - 1)(\beta - 1)(\gamma - 1) < 0$.
\end{corollary}

In the next corollary, we obtain hybrid chain rules involving both $\widebar{H}_\alpha$ and $\widetilde{H}_\alpha$.
\begin{corollary} \label{col:chain2}
    Let $\alpha,\beta,\gamma$ lie in their respective DPI regions. If $(\alpha-1)(\beta-1)(\gamma-1) > 0$, then the following chain rules hold:
	\begin{align}\label{petzchainrule}
		&\widebar{H}^\uparrow_{\alpha}(AB|C)_{\rho}\geq	\widebar{H}^\uparrow_{\beta}(A|BC)_{\rho}+	\widetilde{H}^\uparrow_{\gamma}(B|C)_{\rho}, \ \ \ \text{when} \ \ \ \frac{1}{1-\alpha}=\frac{1}{1-\beta}+\frac{\gamma}{1-\gamma}\\
        &\widetilde{H}^\downarrow_{\alpha}(AB|C)_{\rho}\geq	\widetilde{H}^\downarrow_{\beta}(A|BC)_{\rho}+	\widebar{H}^\downarrow_{\gamma}(B|C)_{\rho}, \ \ \ \text{when}  \ \ \ \frac{\alpha}{1-\alpha}=\frac{\beta}{1-\beta}+\frac{1}{1-\gamma}.
	\end{align}
	In each inequality, the arrows on the entropies indexed by $\alpha$ and $\gamma$ may simultaneously be reversed under the same parameter relation. The inequalities themselves are reversed when $(\alpha - 1)(\beta - 1)(\gamma - 1) < 0$.
\end{corollary}

\section{Monotonicity in the parameters}
\label{sec:mono}

In this section, we establish the monotonicity of the conditional entropy in the parameters $\alpha$, $z$, and $\lambda$. Additionally, we address further monotonicities at the end of the section. The key component of the proof is the asymptotic form derived in Lemma~\ref{H with universal any z}(i). Specifically, since $\omega_{B^n}$ and $\rho_{B^n} =\rho_B^{\otimes n}$ commute, we obtain the following expression:
\begin{equation}
	\label{H universal commuting}
	H^\lambda_{\alpha,z}(A|B)_\rho = \lim_{n \rightarrow \infty} \frac{1}{n} \frac{1}{1-\alpha} \log \textup{Tr}\left(\rho_{A^n B^n}^{\frac{\alpha}{2z}} \left(\rho_{B^n}^{\frac{1-\lambda}{2}} \omega_{B^n}^{\lambda} \rho_{B^n}^{\frac{1-\lambda}{2}}\right)^{\frac{1-\alpha}{z}} \rho_{A^n B^n}^{\frac{\alpha}{2z}}\right)^z.
\end{equation}
This form clearly shows that the monotonicity of the conditional entropy in the parameters $\alpha$ and $z$ is closely related to the monotonicity in the same parameters of the $\alpha$-$z$ Rényi divergence (see Definition~\ref{alpha-z relative entropy}). We prove the monotonicity with respect to the parameters $\alpha$ and $z$ by employing complex interpolation techniques (see Lemma~\ref{Hadamard}). We note that the monotonicity of the $\alpha$-$z$ Rényi divergence with respect to $\alpha$ was first established in~\cite{hiai2024log}. Our approach, based on complex interpolation theory, also offers a straightforward alternative proof of this result in the DPI region. Finally, to establish monotonicity in $\lambda$, we follow a completely different approach by explicitly computing the derivative and proving its positivity.

\subsection{Monotonicity in $\alpha$}

The proof of the following result primarily relies on the asymptotic form in \eqref{H universal commuting} and Lemma~\ref{Hadamard}.
\begin{proposition}
Let $(\alpha,z,\lambda) \in \D$ and $\rho_{AB}$ be a quantum state. The function $\alpha \rightarrow H_{\alpha,z}^\lambda(A|B)_{\rho}$ is monotonically decreasing.
\end{proposition}
\begin{proof}
We can write $Q_{\alpha,z}^\lambda(\rho_{A^n B^n}|\omega_{B^n})$ as
\begin{align}
Q_{\alpha,z}^\lambda(\rho_{A^n B^n}|\omega_{B^n}) =  \left\|\rho_{A^n B^n}^{\frac{\alpha}{2z}} \big(\rho_{B^n}^{\frac{1-\lambda}{2}}\omega_{B^n}^{\frac{\lambda}{2}} \big)^{\frac{1-\alpha}{z}}\right\|_{2z}^{2z}, 
\end{align}
We apply Lemma~\ref{Hadamard} with 
\begin{equation}
	F(s):=\rho_{A^n B^n}^{ \; \frac{1}{2z}\left(s\alpha+(1-s) \beta\right)}\big(\rho_{B^n}^{\frac{1-\lambda}{2}} \omega_{B^n}^{\lambda}\big)^{ \frac{1}{z}\left(1-(s\alpha+(1-s)\beta \right)}
\end{equation}
and $p_0=p_1=p_{\theta} = 2z$. Note that we need to verify the condition $p_0,p_1 \geq 1/2$, which is satisfied because $z \geq 1/2$ holds in the DPI region $\D$.  By leveraging the unitary invariance of the norm, which allows us to simplify the overall phases, we conclude that $Q_{\alpha,z}^\lambda(\rho_{A^n B^n}|\omega_{B^n})$ is log-convex, i.e., for any $0 \leq \theta \leq 1$, it holds
\begin{equation}
Q_{\theta \alpha +(1-\theta)\beta,z}^\lambda(\rho_{A^n B^n}|\omega_{B^n}) \leq Q_{\alpha,z}^\lambda(\rho_{A^n B^n}|\omega_{B^n})^\theta  Q_{\beta,z}^\lambda(\rho_{A^n B^n}|\omega_{B^n})^{1-\theta}.
\end{equation}
Let us assume that $\alpha > \beta > 1$ and set $\theta = \frac{\beta-1}{\alpha-1} \in (0,1)$. Then, log-convexity gives
\begin{equation}
\log{Q_{\beta,z}^\lambda(\rho_{A^n B^n}|\omega_{B^n})} = \log{Q_{\theta \alpha + 1-\theta,z}^\lambda(\rho_{A^n B^n}|\omega_{B^n})} \leq \theta \log{Q_{\alpha,z}^\lambda(\rho_{A^n B^n}|\omega_{B^n})},
\end{equation}
since $ \log{Q_{1,z}^\lambda(\rho_{A^n B^n}|\omega_{B^n})} = 0$. We obtain the desired result by multiplying both sides by $-1/n$ and letting $n \to \infty$.

For the case $ 1 > \alpha > \beta$, an analogous argument holds with $\theta = \frac{1-\alpha}{1-\beta}$.

Finally, the case \( \alpha < 1 < \beta \) follows from the previously derived results, in conjunction with Proposition~\ref{limit alpha to 1}, as \( H^\lambda_{\alpha,z} < H < H^\lambda_{\beta,z} \). 
The latter statement is true because it is always possible to find a monotone path in \( \mathcal{D} \) that leads to the conditional entropy. To show that \( H^\lambda_{\alpha, z} < H \) for \( \alpha < 1 \) and \( z \geq 1 \), we can use the monotonicity in \( \lambda \) and the monotonicity in \( \alpha \) shown earlier. For \( \alpha < 1 \) and \( z \leq 1 \) instead we can use the monotonicity in $\lambda$ and the monotonicity in $\alpha$ of the sandwiched conditional entropy that follows from the monotonicity of the underlying divergence (see e.g.~\cite{lennert13_renyi}). Similarly, to prove \( H < H^\lambda_{\beta, z} \) for \( \beta > 1 \), we can apply duality and follow the same steps just described for $\alpha<1$.
\end{proof}

\subsection{Monotonicity in $z$}

The monotonicity with respect to $z$ follows from the Araki-Lieb-Thirring inequality, similar to the method used to establish the same monotonicity for the $\alpha$-$z$ R\'enyi relative entropies in~\cite[Proposition 6]{lin2015investigating}. It is worth noting that this inequality can also be derived using complex interpolation theory (see, e.g., \cite{sutter16}).

\begin{proposition}
	Let $(\alpha,z,\lambda) \in \D$ and $\rho_{AB}$ be a quantum state. The function $z \mapsto Q_{\alpha,z}^\lambda(A|B)_{\rho}$ is monotonically decreasing. Hence, the function $z \mapsto H_{\alpha,z}^\lambda(A|B)_{\rho}$ is monotonically decreasing for $\alpha<1$ and monotonically increasing for $\alpha>1$.
\end{proposition}
\begin{proof}
	We first demonstrate that for $z' \leq z$, the following inequality holds:
	\begin{equation}
		\textup{Tr}\left(\rho_{A^n B^n}^{\frac{\alpha}{2z}} \Big(\rho_{B^n}^{\frac{1-\lambda}{2}}\omega_{B^n}^{\lambda} \rho_{B^n}^{\frac{1-\lambda}{2}}\Big)^{\frac{1-\alpha}{z}}\rho_{A^n B^n}^{\frac{\alpha}{2z}}\right)^z \leq  \textup{Tr}\left(\rho_{A^n B^n}^{\frac{\alpha}{2z'}} \Big(\rho_{B^n}^{\frac{1-\lambda}{2}}\omega_{B^n}^{\lambda} \rho_{B^n}^{\frac{1-\lambda}{2}}\Big)^{\frac{1-\alpha}{z'}}\rho_{A^n B^n}^{\frac{\alpha}{2z'}}\right)^{z'}.
	\end{equation}
	We apply the Araki-Lieb-Thirring (ALT) inequality, which states that for $0 \leq r \leq 1$, $q \geq 0$, and positive semidefinite matrices $A$ and $B$, the following holds:
	\begin{equation}
		\textup{Tr}\left(A^{\frac{r}{2}}B^{r}A^{\frac{r}{2}}\right)^{\frac{q}{r}} \leq \textup{Tr}\left(A^{\frac{1}{2}}B A^{\frac{1}{2}}\right)^{q}.
	\end{equation}
	Setting $q = z'$, $r = \frac{z'}{z}$, $A = \rho_{A^n B^n}^{\frac{\alpha}{z'}}$, and $B = \Big(\rho_{B^n}^{\frac{1-\lambda}{2}}\omega_{B^n}^{\lambda} \rho_{B^n}^{\frac{1-\lambda}{2}}\Big)^{\frac{1-\alpha}{z'}}$, we multiply both sides by $\frac{1}{n(1-\alpha)}$. Finally, taking the limit as $n \to \infty$ on both sides, together with \eqref{H universal commuting}, yields the desired result.
\end{proof}

\subsection{Monotonicity in $\lambda$}

In this subsection, we aim to establish the monotonicity in $\lambda$.
\begin{proposition}\label{mono.lambda}
	Let $(\alpha,z,\lambda) \in \D$ and $\rho_{AB}$ be a quantum state. The function $\lambda \rightarrow H_{\alpha,z}^\lambda(A|B)_{\rho}$ is monotonically increasing.
\end{proposition}
To prove this, we avoid using complex interpolation theory. Instead, we directly prove the positivity of the derivative. We first derive an alternative variational form for the conditional entropy and show that the term within the optimization is monotone in $\lambda$. Since the optimization on both sides of the inequality preserves the order, this leads directly to the desired result.
\begin{lemma}
	\label{variational form with norms}
	Let $(\alpha, z, \lambda) \in \D$ and $\rho_{AB}$ be a quantum state. Then, we have
	\begin{equation}
		H_{\alpha,z}^\lambda(A|B)_{\rho} = \opt_{\tau_{AB}} \frac{1}{\hat{z}(1-\alpha)} \log{\textup{Tr}\left( \rho_B^{\frac{(1-\lambda)(1-\alpha)}{2z}} \Tr_A\Big(\rho_{AB}^{\frac{\alpha}{2z}} \tau_{AB}^{1-\frac{1}{z}} \rho_{AB}^{\frac{\alpha}{2z}}\Big) \rho_B^{\frac{(1-\lambda)(1-\alpha)}{2z}}\right)^{\hat{z}}}.
	\end{equation}
 Here, $\hat{z} = z/(z+\lambda(\alpha-1))$ is the dual parameter defined in Remark~\ref{explicit computation} and the optimization is performed over all bipartite quantum states $\tau_{AB}$.
\end{lemma}

\begin{proof}
	Using the variational form of the norm in Lemma~\ref{optimization norm}, we obtain
	\begin{equation}
		H_{\alpha,z}^\lambda(A|B)_{\rho} = - \opt_{\sigma_B} \opt_{\tau_{AB}} \frac{z}{\alpha-1} \log{\Trm\left(\rho_{AB}^{\frac{\alpha}{2z}} I_A \otimes \Big(\rho_B^{\frac{(1-\lambda)(1-\alpha)}{2z}} \sigma_B^{\frac{\lambda(1-\alpha)}{z}} \rho_B^{\frac{(1-\lambda)(1-\alpha)}{2z}}\Big)\rho_{AB}^{\frac{\alpha}{2z}} \tau_{AB}^{1-\frac{1}{z}}\right)}.
	\end{equation}
	Within the region $\D$, we can exchange the supremum and infimum whenever necessary by applying Sion's minimax theorem. The optimization problem over $\sigma_B$ can then be solved using Lemma~\ref{optimization norm}, which leads to the desired result.
\end{proof}

The following lemma demonstrates that the term within the optimization in Lemma~\ref{variational form with norms} is monotone in $\lambda$. However, our result is more general and may be of independent interest.
\begin{lemma}
	\label{monotonicity}
	Let $\rho$ be a quantum state, $X \geq 0$, $x \neq 0$, and let $f(x), g(x)$ be functions such that $x g'(x) = g(x)$ and $x f'(x) g(x) = -1$. Then, the function
	\begin{equation}
		x \mapsto \frac{1}{x} \log{\textup{Tr}\left(X \rho^{f(x)} X\right)^{g(x)}}
	\end{equation}
	is monotone.
\end{lemma}
\begin{proof}
	We follow a similar approach to the one used to prove~\cite[Lemma 25]{rubboli2023mixed}. We denote $A := X \rho^{f(x)} X$.  The derivative with respect to $x$ is 
	\begin{align}
		& -\frac{1}{x^2} \log{\Trm\left(A^{g(x)}\right)} + \frac{1}{x \Trm\left(A^{g(x)}\right)} \Trm\left(A^{g(x)} (\log{A}) g'(x) + \big(g(x) f'(x)\big) X \rho^{f(x)} (\log{\rho}) X A^{g(x)-1} \right) \\
		& \quad = \frac{1}{x^2 \Trm\left(A^{g(x)}\right)} \Big(\Trm\left(A^{g(x)} \log{A^{x g'(x)}}\right) - \Trm\left(A^{g(x)}\right) \log{\Trm\left(A^{g(x)}\right)} \\
		& \qquad\qquad \qquad\qquad\qquad + \big(g(x) f'(x) x\big) \Trm\left(X \rho^{f(x)} (\log{\rho}) X A^{g(x)-1} \right) \Big).
	\end{align}
	
	We use the property $f(Y^\dagger Y)Y^\dagger = Y^\dagger f(YY^\dagger)$, and set $Y = X \rho^\frac{f(x)}{2}$. We can then rewrite
	\begin{equation}
		\Trm\left(X \rho^{f(x)} (\log{\rho}) X A^{g(x)-1}\right) = \Trm\left(\log{\rho}) \tilde{A}^{g(x)}\right),
	\end{equation}
	where we define $\tilde{A} := \rho^{\frac{f(x)}{2}} X \rho^{\frac{f(x)}{2}}$. Note that $A$ and $\tilde{A}$ share the same eigenvalues. Therefore, to prove that the function is monotonically increasing, using $g(x) f'(x) x = -1$, we need to show that
	\begin{equation}
		\Trm\left(\tilde{A}^{g(x)} \log{\tilde{A}^{g(x)}}\right) - \Trm\left(\tilde{A}^{g(x)}\right) \log{\Trm\left(\tilde{A}^{g(x)}\right)} - \Trm\left(\log{\rho}) \tilde{A}^{g(x)}\right) \geq 0 \,.
	\end{equation}
	This condition is a consequence of the positivity of the relative entropy. Indeed,
	\begin{align}
		\Trm\left(\tilde{A}^{g(x)} \log{\rho}\right) &= \Trm\left(\tilde{A}^{g(x)}\right) \Trm\left(\frac{\tilde{A}^{g(x)}}{\Trm\left(\tilde{A}^{g(x)}\right)} \log{\rho}\right) \\
		&\leq \Trm\left(\tilde{A}^{g(x)}\right) \Trm\left(\frac{\tilde{A}^{g(x)}}{\Trm\left(\tilde{A}^{g(x)}\right)} \log{\frac{\tilde{A}^{g(x)}}{\Trm\left(\tilde{A}^{g(x)}\right)}}\right) \\
		&= \Trm\left(\tilde{A}^{g(x)} \log{\tilde{A}^{g(x)}}\right) - \Trm\left(\tilde{A}^{g(x)}\right) \log{\Trm\left(\tilde{A}^{g(x)}\right)}.
	\end{align}
The proof is complete.
\end{proof}

Finally, we combine the above result with the variational form in Lemma~\ref{variational form with norms} to establish the desired monotonicity.
\begin{proof}[Proof of Proposition \ref{mono.lambda}]
	Note that since $\Trm(AA^\dagger)^s = \Trm(A^\dagger A)^s$, it is sufficient to derive the monotonicity for any $\tau_{AB}$ of the function 
	\begin{equation}
		\lambda \rightarrow F(\lambda) := \frac{1}{\hat{z}(1-\alpha)} \log{\textup{Tr}\left( X_B \rho_B^{ \frac{(1-\lambda)(1-\alpha)}{z}} X_B \right)^{\hat{z}}},
	\end{equation}
	where \( X_B := \left(\Tr_A\left(\rho_{AB}^{\frac{\alpha}{2z}} \tau_{AB}^{1-\frac{1}{z}} \rho_{AB}^{\frac{\alpha}{2z}}\right)\right)^{\frac{1}{2}} \). Indeed, since the optimization over \( \tau_{AB} \) preserves monotonicity, the variational form in Lemma~\ref{variational form with norms} implies the monotonicity of \( H_{\alpha,z}^\lambda(A|B) \).
	
	Let us set \( x = \hat{z} \). We can rewrite the function in the new coordinates as 
	\begin{equation}
		F(x) = \frac{1}{1-\alpha} \frac{1}{x} \log{\textup{Tr}\left(\left( X_B \rho_B^{f(x)} X_B \right)^{g(x)}\right)},
	\end{equation}
	where \( f(x) = \frac{1-\alpha}{z} - 1 + x^{-1} \) and \( g(x) = x \).
	
	To prove monotonicity, we show that the derivative with respect to \( \lambda \) is positive. By the chain rule for derivatives, we have
	\begin{align}
		\frac{\d}{\d \lambda} F(\lambda) &= \frac{\d}{\d x} F(x) \frac{\d x}{\d \lambda} \\
		&= \frac{\hat{z}^2(1-\alpha)}{z} \frac{\d}{\d x} F(x) \\
		&= \frac{\hat{z}^2}{z} \frac{\d}{\d x} \left(\frac{1}{x} \log{\textup{Tr}\left( X_B \rho_B^{f(x)} X_B \right)^{g(x)}}\right).
	\end{align}
	Since \( z,\hat{z}^2 \geq 0 \), it suffices to show that
	\begin{equation}
		\frac{\d}{\d x} \left(\frac{1}{x} \log{\textup{Tr}\left( X_B \rho_B^{f(x)} X_B \right)^{g(x)}}\right) \geq 0.
	\end{equation}
	This is a consequence of Lemma~\ref{monotonicity}.
\end{proof}

We conclude this section by highlighting additional monotonicity results. These monotonicities also hold along the lines \( z = 1 - \alpha \), \( z = \alpha - 1 \), and \( z = \alpha \). For the case \( z = 1 - \alpha \), the proof follows a similar approach to that of \cite[Lemma 25]{rubboli2023mixed} and is derived from Lemma~\ref{monotonicity} with \( x = 1 - \alpha \), \( f(x) = -1 + \frac{1}{x} \), and \( g(x) = x \), in conjunction with the asymptotic form given in \eqref{H universal commuting}. The case \( z = \alpha - 1 \) follows an analogous reasoning. Finally, the proof for \( z = \alpha \) employs the methods outlined in \cite{beigi13_sandwiched}.

\section{Classical information}
\label{sec:classical}

In this section, we derive an expression for partly classical registers and a closed-form expression for fully classical states. 

\begin{proposition}[Classical registers]
\label{classical conditioning}
Let $(\alpha,z,\lambda) \in \D$ and $\rho_{ABY}=\sum_y \rho(y)
\hat{\rho}_{AB}(y)\otimes \ketbra{y}{y}$. Then, the conditional entropy satisfies
\begin{equation}
H^\lambda_{\alpha,z}(A|BY)_{\rho}=\left(\frac{1}{1-\alpha}-\lambda\right)\log{\bigg(\sum \nolimits_{y}\rho(y)\exp \bigg(\frac{1-\alpha}{1-\lambda(1-\alpha)}H^\lambda_{\alpha,z}(A|B, Y=y)_{\rho}\bigg)\bigg)}.
\end{equation}
Here, we denoted $H^\lambda_{\alpha,z}(A|B, Y=y)_{\rho} = H^\lambda_{\alpha,z}(A|B)_{\hat{\rho}_{AB}(y)}$.
\end{proposition}
\begin{proof}
This follows by guessing as ansatz the state
\begin{equation}
\sigma_{BY} = \sum_y \frac{\rho(y) Q_{\alpha,z}^\lambda(A|B)^{\frac{1}{1-\lambda(1-\alpha)}}_{\hat{\rho}_{AB}(y)}}{\sum_x \rho(x) Q_{\alpha,z}^\lambda(A|B)^{\frac{1}{1-\lambda(1-\alpha)}}_{\hat{\rho}_{AB}(x)}} \; \hat{\sigma}_{B}(y) \otimes \ketbra{y}{y},
\end{equation}
and verify that it satisfies the sufficient conditions for the minimum derived in Lemma~\ref{fixed-point}. Here, $ \hat{\sigma}_{B}(y)$ is the optimizer of $\hat{\rho}_{AB}(y)$. 
\end{proof}

Let us consider a probability distribution $\rho(x,y)$. We define
\begin{align}
\label{Hayashi conditional entropy}
&H^{\mathrm{H}}_{\alpha}(X|Y)_{\rho}=\frac{1}{1-\alpha}\log{\Big(\sum\nolimits_{x,y} \rho(y) \rho(x|y)^\alpha\Big)} \,, \\
\label{Arimoto conditional entropy}
&H^{\mathrm{A}}_{\alpha}(X|Y)_{\rho}=\frac{\alpha}{1-\alpha}\log{\bigg(\sum\nolimits_y \rho(y)\Big(\sum\nolimits_x \rho(x|y)^\alpha\Big)^\frac{1}{\alpha}\bigg)} \,.
\end{align}
The first conditional entropy, $H^{\mathrm{H}}_{\alpha}$, was introduced by Hayashi and Skoric et al.~\cite{hayashi2011exponential,vskoric2011sharp}, while the second one $H^{\mathrm{A}}_{\alpha}$ was introduced by Arimoto~\cite{arimoto1977information}.

In the following, we derive a classical expression for the new conditional entropies. In the classical case, our quantity reduces to the one introduced by Hayashi and Tan in~\cite{hayashi2016equivocations}. For $\lambda=0$, we recover the quantity $H^{\mathrm{H}}_\alpha$ while for $\lambda=1$ we obtain $H^{\mathrm{A}}_\alpha$. Moreover, the new quantity interpolates (and extrapolates) between these two definitions. 

The fully classical expression for the new conditional entropy $H_{\alpha,z}^\lambda$ can be derived by considering the state $\rho_{x,y}=\sum_{xy}\rho(x,y)\ketbra{x}{x}\otimes \ketbra{y}{y}$ and, for example, using the solution for the Petz case $z=1$ in Corollary~\ref{closed-form Petz}. Indeed, in the classical case, the conditional entropy does not depend on $z$. Note that the closed form can also be extended outside the range for which the case $z=1$ satisfies the DPI. To this purpose, let us define $\beta=\alpha/(1-\lambda(1-\alpha))$. We obtain 
\begin{equation}
\label{classical closed-form}
H_{\alpha,\beta}(X|Y)_{\rho}=\frac{\alpha}{1-\alpha}\log{\bigg(\sum\nolimits_y \rho(y)\Big(\sum\nolimits_x \rho(x|y)^\alpha\Big)^\frac{\beta}{\alpha}\bigg)^\frac{1}{\beta}}.
\end{equation}
Hence, in the classical case, we obtain a weighted $\beta$ quasi-norm of the $\alpha$ quasi-norm of the conditional distribution.
From the result obtained in Theorem~\ref{DPItheorem} for the fully quantum case, we have that the classical conditional entropy satisfies the DPI in the region $-\infty \leq \lambda \leq 1$ for $0 <\alpha < 1$ and  $1-\alpha/(\alpha-1) \leq \lambda \leq 1$ for $\alpha>1$. In the coordinates $\alpha$ and $\beta$, we obtain
\begin{align}
\D_{\text{cl}} :&=\Big\{(\alpha,\beta)\in\mathbb{R}^2: 0 < \alpha < 1,  0 < \beta \leq 1 \Big\} 
\; \cup \; \Big\{(\alpha,\beta)\in\mathbb{R}^2:\alpha>1, \beta \geq 1\Big\}.
\end{align} 
Finally, by considering the state $\rho_{xyz}=\sum_{xyz}\rho(x,y,z)\ketbra{x}{x}\otimes \ketbra{y}{y} \otimes \ketbra{z}{z}$, Thereom~\ref{Final chain rule} yield a chain rule for classical states.  For a probability distribution $\rho(x,y,z)$, we obtain if $(\alpha-1)(\beta-1)(\gamma-1) > 0$
\begin{equation}
    H_{\alpha}^\lambda(XY|Z)_{\rho} \geq H_{\beta}^\mu(X|YZ)_{\rho} + H_{\gamma}^\lambda(Y|Z)_{\rho} \quad \text{when} \quad \frac{1}{1-\alpha}=\frac{1}{1-\beta}-\mu + \frac{1}{1-\gamma}
\end{equation}
that holds whenever there exist parameters \(z, w\), and \(v\) such that the corresponding triples lie within the DPI region. Moreover, the inequality is reversed when $(\alpha-1)(\beta-1)(\gamma-1) < 0$.

\subsection{Dimension bounds}
\begin{lemma}
\label{maximal and minimal}
    Let $(\alpha,z,\lambda) \in \D$ and $\rho_{AB}$ be a quantum state. Then, it holds that $H^{\downarrow}_{\min}(A|B)_{\rho} \leq H^\lambda_{\alpha,z}(A|B) _{\rho}\leq H^{\uparrow}_{\max}(A|B)_{\rho}$.
\end{lemma}
\begin{proof}
    The upper bound is established as  
    \begin{align}
        H^\lambda_{\alpha,z}(A|B)_{\rho} \leq H^{\lambda=1}_{\alpha,z}(A|B)_{\rho} \leq H^{\uparrow}_{\max}(A|B)_{\rho},
    \end{align} 
where the first inequality follows from the monotonicity of \(H^\lambda_{\alpha,z}(A|B)_{\rho}\) with respect to \(\lambda\) (Proposition~\ref{mono.lambda}), and the second inequality arises from the property \(D_{\alpha,z} \geq \widebar{D}_{0}\)~\cite[Theorem 3]{gour2020optimal}.

The lower bound is obtained by applying the duality relations from Theorem~\ref{Duality} in conjunction with the previously derived upper bound.
\end{proof}
 The latter bounds imply the following result~\cite[Lemma 5.11]{tomamichel16_book}
\begin{corollary}\label{dimension bound}
Let $(\alpha,z,\lambda) \in \D$ and  $\rho_{AB}$ a state. Then, the following holds
\begin{equation}
-\log{\min\{\textup{rank}(\rho_A),\textup{rank}(\rho_B)\}} \leq H^\lambda_{\alpha,z}(A|B)_{\rho} \leq \log{\textup{rank}(\rho_A)} \,.
\end{equation}
Moreover, $H^\lambda_{\alpha,z}(A|B)_{\rho} \geq 0$ if $\rho_{AB}$ is separable.
\end{corollary}

\section{Limits}
\label{section limits}
The first limit we examine leads to the von Neumann conditional entropy. We show that for any path within the DPI region, the new conditional entropy converges to the von Neumann entropy as $\alpha \rightarrow 1$.

\begin{proposition}
\label{limit alpha to 1}
Let $\rho_{AB}$ be a quantum state. Let $J$ be an open interval of positive real numbers containing $1$, and let $f,g: J \rightarrow \mathbb{R}$ be two functions continuous at $\alpha=1$, such that $(\alpha,g(\alpha),f(\alpha))$ lies in $\D$ for all $\alpha \in J$. Then we have
 \begin{equation}
     \lim_{\alpha \rightarrow 1} H_{\alpha,g(\alpha)}^{f(\alpha)}(A|B)_\rho = H(A|B)_\rho.
 \end{equation}
\end{proposition}

\begin{proof}
  The proof is constructed by leveraging the duality relations and the monotonicity properties in $z$ and $\lambda$ established in the previous sections. By combining these results, we derive both lower and upper bounds involving the von Neumann conditional entropy. Let us start with the upper bound. The monotonicity in $\lambda$ established in Proposition~\ref{mono.lambda} implies that
   \begin{equation}
       \limsup_{\alpha \rightarrow 1} H^{f(\alpha)}_{\alpha,g(\alpha)}(A|B)_{\rho} \leq \limsup_{\alpha \rightarrow 1} H^\uparrow_{\alpha,g(\alpha)}(A|B)_{\rho} \leq \limsup_{\alpha \rightarrow 1} H^\uparrow_{\alpha,z_0}(A|B)_{\rho} = H(A|B)_{\rho} \,.
   \end{equation}
   In the second inequality, we used the fact that the \(\alpha\)-\(z\) Rényi relative entropy is monotone in \(z\) for any $z>0$~\cite[Proposition 6]{lin2015investigating}. Consequently, a fixed \(z_0>0\) exists that allows us to bound the conditional entropy using a path where \(z\) is held constant. 
   
   In the final inequality, we utilized a well-established argument based on the minimax theorem~\cite[Corollary A.2]{mosonyi2011quantum} (see, e.g.,~\cite[Appendix A]{rubboli2022new}). 
   Specifically, the limits \(\alpha \to 1^-\) and \(\alpha \to 1^+\) can be expressed as an infimum or supremum, due to the monotonicity of the \(\alpha\)-\(z\) Rényi relative entropy with respect to \(\alpha\) for any $z>0$~\cite[Theorem 3.1]{hiai2024log}. Moreover, the \(\alpha\)-\(z\) Rényi relative entropy is lower semicontinuous in its second argument (see, e.g.,~\cite[Lemma 18]{rubboli2022new}). Consequently, by combining these properties, the limit can be exchanged with the optimization over all states using the minimax theorem in~\cite[Corollary A.2]{mosonyi2011quantum}. Finally, the $\alpha$-$z$ R\'enyi relative entropy converges to the Umegaki relative entropy for $\alpha \rightarrow 1$ for any $z>0$~\cite[Theorem 4]{lin2015investigating}. Let us illustrate the argument with an example. Let us consider the limit $\alpha \rightarrow 1^-$. We then have
   \begin{align}
\lim \limits_{\alpha \rightarrow 1^-}\sup \limits_{\sigma_B} -D_{\alpha,z_0}(\rho_{AB} \| I_A \otimes \sigma_B) &= \inf_{\alpha < 1} \sup_{\sigma_B} -D_{\alpha,z_0}(\rho_{AB} \| I_A \otimes \sigma_B) \\
&=  \sup_{\sigma_B}\inf_{\alpha < 1} -D_{\alpha,z_0}(\rho_{AB} \| I_A \otimes \sigma_B) \\
&= \sup_{\sigma_B} -D(\rho_{AB} \| I_A \otimes \sigma_B) \\
&= -D(\rho_{AB} \| I_A \otimes \rho_B) \,.
\end{align}

For the lower bound, we apply the same reasoning in conjunction with the duality established in Theorem~\ref{Duality}. We have 
\begin{align}
       \liminf_{\alpha \rightarrow 1} H^{f(\alpha)}_{\alpha,g(\alpha)}(A|B)_{\rho} &= \liminf_{\alpha \rightarrow 1} -H^{w(\alpha)}_{u(\alpha),v(\alpha)}(A|C)_{\rho} \\
       &\geq \liminf_{\alpha \rightarrow 1} -H^\uparrow_{u(\alpha),v(\alpha)}(A|C)_{\rho} \\
       &\geq \liminf_{\alpha \rightarrow 1} -H^\uparrow_{u(\alpha),z_0}(A|C)_{\rho} \\
       &= -H(A|C)_{\rho} \\
       &= H(A|B)_{\rho} \,.
   \end{align}
In the first equality, we defined the functions $u,v,w$ induced by the duality relations (see equation~\eqref{explcit computation duality}). In the second inequality, we used the fact that $u(\alpha) \rightarrow 1$ as $\alpha\rightarrow1$, which follows from the duality relations~\eqref{explcit computation duality}.
\end{proof}

We observe that the condition requiring the function \( f \) to be continuous excludes the case \( f(1) = -\infty \). Indeed, in this scenario, the newly defined entropy does not converge to the von Neumann conditional entropy (see the next section).  
Similarly, the assumption that the path lies within \( \mathcal{D} \) excludes the case where \( z \to 0 \) as \( \alpha \to 1 \). For instance, it is well-known that in such cases, the \( \alpha \)-\( z \) Rényi relative entropy does not converge to the Umegaki relative entropy (see~\cite[Section 3]{audenaert13_alphaz}).

The second limits we consider are \(\lambda = 0\) and \(\lambda = 1\), which recover the quantities \(H^\downarrow_{\alpha,z}\) and \(H^\uparrow_{\alpha,z}\), respectively. 
The limit as \(\lambda \to 0\) is not immediately clear, as the exponent of the states \(\sigma_B\) in the definition of the new conditional entropy introduces an optimization over all projectors. Furthermore, the behavior of the limit depends on the direction of approach: for \(\lambda \to 0^-\) (from below) and \(\lambda \to 0^+\) (from above), the optimization involves either an infimum and a supremum or vice versa. As a result, it is not evident a priori which optimization should be adopted. 

In the following, we demonstrate that the limits \(\lambda \to 0^+\) and \(\lambda \to 0^-\) converge to the same value. This value coincides with the expression obtained by taking the supremum at \(\lambda = 0\) and minimizing over all projectors \(\sigma_B^0\), where the optimal choice for \(\sigma_B^0\) is clearly the identity operator. This observation justifies the choice of the supremum, rather than the infimum, in Definition~\ref{new entropy} for the new conditional entropy when $\lambda=0$.

\begin{proposition}
\label{limit lambda to zero and one}
    Let $(\alpha,z,\lambda) \in \D$ and $\rho_{AB}$ be a quantum state. Then, we have
    \begin{equation}
    \lim_{\lambda \rightarrow 0} H^{\lambda}_{\alpha,z}(A|B)_{\rho} = H^\downarrow_{\alpha,z}(A|B)_{\rho} \qquad
    \text{and} \qquad \lim_{\lambda \rightarrow 1} H^{\lambda}_{\alpha,z}(A|B)_{\rho} = H^\uparrow_{\alpha,z}(A|B)_{\rho}.
    \end{equation}

\end{proposition}
\begin{proof}
     The proof is similar to the one used to prove Proposition~\ref{limit alpha to 1} and it is a consequence of the minimax theorem~\cite[Corollary A.2]{mosonyi2011quantum} together with the variational form in Lemma~\ref{variational form with norms} and Lemma~\ref{monotonicity}. 
     
   Let us consider first the limit $\lambda \rightarrow 0$. The limits \(\lambda \to 0^-\) and \(\lambda \to 0^{+}\) can be expressed as an infimum or supremum. Indeed, according to Lemma~\ref{variational form with norms}
   \begin{equation}
       H^{\lambda}_{\alpha,z}(A|B)_{\rho}=\opt_{\tau_{AB}} F(\lambda,\tau_{AB}) \,,
   \end{equation}
   where we defined 
   \begin{equation}
		 F(\lambda,\tau_{AB}) := \frac{1}{\hat{z}(1-\alpha)} \log{\textup{Tr}\left( X_B(\tau_{AB}) \rho_B^{ \frac{(1-\lambda)(1-\alpha)}{z}} X_B(\tau_{AB}) \right)^{\hat{z}}} \,.
	\end{equation}
	Here, \( X_B(\tau_{AB}) := \left(\Tr_A\left(\rho_{AB}^{\frac{\alpha}{2z}} \tau_{AB}^{1-\frac{1}{z}} \rho_{AB}^{\frac{\alpha}{2z}}\right)\right)^{\frac{1}{2}} \) and $\hat{z} = z/(z+\lambda(\alpha-1))$. From the proof of Proposition~\ref{mono.lambda}, the function $F(\lambda,\tau_{AB})$ is monotone in $\lambda$. Moreover, using the same argument used to prove Lemma~\ref{lemma semicontintuity}, it follows that $\tau_{AB} \rightarrow F(\lambda,\tau_{AB})$ is lower/upper semicontinuous. Indeed, in the DPI region, it holds that $1-1/z \in [-1,1]$ and in this range the power is either operator monotone or antimonotone.
    Consequently, by combining these properties, the limit can be exchanged with the optimization over all states using the minimax theorem in~\cite[Corollary A.2]{mosonyi2011quantum}.  To illustrate the result, let us examine a specific case. Let us consider the case $\alpha<1$, $z>1$ and the limit $\lambda \rightarrow 0^{+}$. We then have
   \begin{align}
\lim \limits_{\lambda \rightarrow 0^+}\sup \limits_{\tau_{AB}} F(\lambda,\tau_{AB}) &= \inf_{\lambda > 0} \sup_{\tau_{AB}} F(\lambda,\tau_{AB}) \\
&=  \sup_{\tau_{AB}}\inf_{\lambda > 0} F(\lambda,\tau_{AB}) \\
&=  \sup_{\tau_{AB}}\lim_{\lambda \rightarrow 0^{+}} F(\lambda,\tau_{AB}) \\
&=  H_{\alpha,z}^\downarrow(A|B)_{\rho} \,.
\end{align}
where the last equality can be verified using the variational form of the norm as done in Lemma~\ref{variational form with norms}. The limit \(\lambda \to 0^-\) is more straightforward, as it involves two supremums and does not require the application of the minimax theorem.
    
The limit $\lambda \rightarrow 1$ is analogous to the limit $\lambda \rightarrow 0$ where the last step follows by noticing that $X_B(\tau_{AB}) \ll\rho_B^0$.
\end{proof}

Finally, we compute some limits for the classical conditional entropy derived in~\eqref{classical closed-form}.
The first limits we examine are \(\beta \to +\infty\) and \(\beta \to 0\), while keeping \(\alpha\) constant. In these limits, the weighted \(\beta\) quasi-norm yields 
\begin{align}
    &\lim_{\beta \rightarrow + \infty}H_{\alpha, \beta}(X|Y)_{\rho} = \min_{y} H_{\alpha} (X|Y=y)_{\rho} \,, \qquad \alpha>1 \\
    & \lim_{\beta \rightarrow 0}H_{\alpha, \beta}(X|Y)_{\rho} = \sum_{y}\rho(y) H_{\alpha}(X|Y=y)_{\rho} \,, \qquad 0<\alpha<1 \,,
\end{align}
where $H_{\alpha} (X|Y=y)_{\rho} = \frac{1}{1-\alpha}\log{\big(\sum_x\rho(x|y)^\alpha\big)}$. The second limit we consider is \(\beta \to \infty\) and \(\alpha \to 1\) within the region \(\alpha > 1\). To analyze this, we introduce the parameterization \(\alpha = 1 + x a\) and \(\beta = 1/x\), where \(a> 0\), and take the limit \(x \to 0\). Explicitly,
\begin{align}
&\lim_{x \rightarrow 0} H_{1+xa, \frac{1}{x}}(X|Y)_{\rho}= -\frac{1}{a} \log{\bigg(\sum \nolimits_{y} \rho(y) \text{exp}{\Big(-a H(X|Y=y)_{\rho}\Big)}\bigg)} \,,
\end{align}
where $H(X|Y=y)_{\rho} = -\sum_x \rho(x|y)\log{\rho(x|y)}$. The limit of the sum over \(x\) can be computed, for instance, by taking the logarithm and applying L'Hôpital's rule.

\section{Conclusion}
\label{sec:conclusion}

We have introduced a new family of quantum conditional R\'enyi entropies that encompasses various previously studied examples. Our quantity satisfies the DPI for the range of parameters defined by the three-dimensional region $\mathcal{D}$. While we conjecture that the DPI holds exclusively within this region and is violated outside of it, further proof is needed to confirm this. It is worth mentioning that this family does not appear to include quantum conditional entropies derived from maximal quantum R\'enyi divergences, such as those based on the Belavkin--Staszewski relative entropy~\cite{belavkin82_relent}. Although these entropies have yet to find clear operational significance, this highlights that a complete characterization of all R\'enyi-type conditional entropies satisfying DPI remains elusive.

Finally, we believe that it is worthwhile to investigate limiting behaviors of our quantity, especially as $\alpha \to 1$ while other parameters diverge. This is exemplified by the dual of the log-Euclidean conditional entropy ($z=\infty$ and $0 < \alpha < 1$). For a formal definition of the log-Euclidean divergence $D_{\alpha,+\infty} = D_\alpha^\flat$, see e.g.~\cite[Example III.34]{mosonyi2024geometric}. From the relation between the parameters in the duality relations~\ref{explicit computation}, we observe that as $z \rightarrow \infty$, the parameters of the dual behave such that $\hat{\alpha}, \hat{z} \rightarrow 1$ and  $\hat{\lambda}\rightarrow-\infty$.  For instance, in the arrow down case, we observe that the dual of the log-Euclidean conditional entropy is given by
\begin{equation}
    D^\flat_{\alpha}\left(\rho_B \middle\| \exp(\rho_B^{-\frac{1}{2}} \Tr_A(\rho_{AB}\log{\rho_{AB}})\rho_B^{-\frac{1}{2}}) \right).
\end{equation}
 Interestingly, the latter quantity behaves like the conditional von Neumann entropy for states of the form $ \rho_A \otimes \ketbra{\sigma}{\sigma}_B$ while it yields the R\'enyi entropy of the marginal $\rho_A$ for pure states $\ket{\rho}_{AB}$. Nonetheless, we leave this and other cases for future investigation.

\medskip
\paragraph*{\emph{\textbf{Acknowledgments.}}}
We thank Christoph Hirche for preliminary discussions on duality relations involving $\alpha$-$z$ R\'enyi relative entropies. This project is supported by the National Research Foundation, Singapore (NRF) through the National Quantum Office, hosted in A*STAR, under its Centre for Quantum Technologies Funding Initiative (S24Q2d0009) and by the NRF Investigatorship award (NRF-NRFI10-2024-0006). R.R. also acknowledges financial support from the ERC grant GIFNEQ 101163938.

\medskip
\paragraph*{\emph{\textbf{Data availability.}}}
We do not analyse or generate any data.

\medskip
\paragraph*{\emph{\textbf{Conflict of interest.}}}
The authors have no competing interests to declare.

\bibliographystyle{ultimate}
\bibliography{library}

\appendix

\section{Lower and upper semicontinuity and its consequences}
\label{lower semicontinuity}
In this appendix, we establish that the conditional entropy 
\( H_{\alpha,z}^\lambda(A|B)_\rho \) is continuous as a function of the state \( \rho \). 
This continuity property allows us to restrict several arguments in the main text 
to the case of full-rank states, thereby avoiding technical complications related 
to supports and operator powers that may otherwise be ill-defined.

We first start by proving that in the DPI region $\D$, the function $\Upsilon_{\alpha,z}^\lambda$ defined in Definition~\ref{trivariate} is lower or upper semicontinuous in $\sigma_B$, depending on the sign of $\lambda$. This argument is standard in the literature~\cite{mosonyi2011quantum,lami2024continuity} (see also \cite[Appendix A]{rubboli2022new}). 
We start by proving the monotonicity of $\Upsilon_{\alpha,z}^\lambda$ in the third argument.
\begin{lemma}
\label{decreasing}
Let $(\alpha,z, \lambda) \in \mathcal{D}$, $\rho \in \mathcal{S}(A)$, and $\tau,\sigma,\sigma' \in \mathcal{P}(A)$ such that $\rho \ll \tau$. If $\sigma \leq \sigma' $, then for $\lambda \geq 0$ we have the inequality
\begin{align}
&\Upsilon_{\alpha,z}^\lambda(\rho, \tau, \sigma) \geq \Upsilon_{\alpha,z}^\lambda(\rho, \tau, \sigma')  \,.
\end{align}
Moreover, if $\lambda < 0$, the inequality is reversed. 
\end{lemma}
\begin{proof}
We split the proof for the region $(\alpha,z, \lambda) \in \mathcal{D}$ with $\lambda > 0$ in the range $\alpha <1$ and $\alpha>1$. 

Let $\alpha <1$. We then have 
\begin{align}
 \sigma \leq \sigma'
&\implies  \sigma^{\lambda\frac{1-\alpha}{z}} \leq {\sigma'}^{\lambda \frac{1-\alpha}{z}} \\
& \implies  \Trm\left(\rho^\frac{\alpha}{2z} \tau^{\frac{1-\lambda}{2} \frac{1-\alpha}{z}} \sigma^{\lambda \frac{1-\alpha}{z}}  \tau^{\frac{1-\lambda}{2} \frac{1-\alpha}{z}}\rho^\frac{\alpha}{2z}\right)^z \leq \Trm\left(\rho^\frac{\alpha}{2z} \tau^{\frac{1-\lambda}{2} \frac{1-\alpha}{z}} {\sigma' }^{\lambda\frac{1-\alpha}{z}} \tau^{\frac{1-\lambda}{2} \frac{1-\alpha}{z}} \rho^\frac{\alpha}{2z}\right)^z  
\label{Impl2} \,.
\end{align}
In the first implication, we used that $\lambda (1-\alpha)/z \in (0,1]$ and that the power is operator monotone in the range $(0,1]$. In the last implication, we used that the trace functional $M \rightarrow \Tr(f(M))$ inherits the monotonicity from $f$~(see e.g.~\cite{carlen2010trace}). Therefore, in this range, we get $\Upsilon_{\alpha,z}^\lambda(\rho, \tau, \sigma) \geq \Upsilon_{\alpha,z}^\lambda(\rho, \tau, \sigma')$. 

The proof for $\alpha >1$ follows in the same way by noticing that $(1-\alpha)/z \in [-1,0)$ and that the power in this range is operator antimonotone. Finally, the proof for $\lambda < 0$ follows similarly, with the final inequality reversed.  
\end{proof}

\begin{lemma}
\label{lemma semicontintuity}
Let $(\alpha,z, \lambda)\in \mathcal{D}$, and let $\rho_{AB}$, $\sigma_B$ be states. The function
$
(\rho_{AB},\sigma_B) \rightarrow \Upsilon_{\alpha,z}^\lambda(\rho_{AB},I_A \otimes \rho_B,I_A \otimes \sigma_B)
$
is lower semicontinuous for $\lambda \geq 0$, and upper semicontinuous for $\lambda < 0$.   
\end{lemma}

\begin{proof}
Let us first consider the case where $\lambda \geq 0$. For any $\varepsilon>0$, we have the operator inequality $I_A \otimes \sigma_B \leq I_A \otimes \sigma_B+\varepsilon I_{AB}$.
From the previous lemma we have $ \Upsilon_{\alpha,z}^\lambda(\rho_{AB},I_A \otimes \rho_B,I_A \otimes \sigma_B) \geq \Upsilon_{\alpha,z}^\lambda(\rho_{AB},I_A \otimes \rho_B,I_A \otimes \sigma_B+\varepsilon I_{AB})$ and hence we can write
\begin{equation}
\label{sup}
\Upsilon_{\alpha,z}^\lambda(\rho_{AB},I_A \otimes \rho_B,I_A \otimes \sigma_B) = \sup \limits_{\varepsilon>0} \Upsilon_{\alpha,z}^\lambda(\rho_{AB},I_A \otimes \rho_B,I_A \otimes \sigma_B+\varepsilon I_{AB}) \,.
\end{equation} 
For any fixed $\varepsilon>0$, the functions $(\rho_{AB},\sigma_B) \rightarrow \Upsilon_{\alpha,z}^\lambda(\rho_{AB},I_A \otimes \rho_B,I_A \otimes \sigma_B+\varepsilon I_{AB})$ are continuous since the third argument of $\Upsilon$ has full support and $\rho_{AB} \ll I_A \otimes \rho_B$. Since the pointwise supremum of continuous functions is lower semicontinuous, the function  $(\rho_{AB},\sigma_B) \rightarrow \Upsilon_{\alpha,z}^\lambda(\rho_{AB},I_A \otimes \rho_B,I_A \otimes \sigma_B)$ is lower semicontinuous. The case $\lambda < 0$ is similar with $\sup$ replaced by $\inf$ and lower semicontinuity substituted with upper semicontinuity.
\end{proof}
A direct consequence of the above lemma is that, when $\lambda > 0$, we have $\inf_{\sigma_B}\Upsilon_{\alpha,z}^\lambda(\rho_{AB}, I_A \otimes \rho_B, I_A \otimes \sigma_B) = \min_{\sigma_B}\Upsilon_{\alpha,z}^\lambda(\rho_{AB}, I_A \otimes \rho_B, I_A \otimes \sigma_B)$, meaning the infimum is always attained.   It is important to note that the set of quantum states is compact, for instance, in the trace norm topology. Therefore, the statement follows from the property that a lower semicontinuous function defined on a compact set attains its minimum within that set. Finally, in contrast, when $\lambda < 0$, the supremum becomes a maximum.

In the following, we establish the continuity of the conditional entropy $H_{\alpha,z}^\lambda(A|B)_\rho$ over the state space. We leverage this property. for instance, when establishing additivity in Section~\ref{sec:additivity} and
chain rules in Section~\ref{sec:chainrules}.
Our proofs are inspired by the strategies used in \cite[Theorem 1]{lami2024continuity} and \cite[Theorem 5]{lami2023attainability}.
\begin{proposition}\label{Continuity}
    Let $(\alpha,z,\lambda) \in \mathcal{D}$. Then the function $\mathcal{S}(AB)\ni\rho\mapsto H_{\alpha,z}^\lambda(A|B)_\rho$ is continuous.
\end{proposition}
\begin{proof}
    Since by Corollary \ref{dimension bound}, $Q_{\alpha,z}^\lambda(A|B)_\rho$ is bounded away from zero and infinity, it suffices to prove that the function
    \begin{equation}
        \rho\mapsto Q_{\alpha,z}^\lambda(A|B)_\rho=\opt_{\sigma_B}Q_{\alpha,z}^\lambda(\rho_{AB}|\sigma_B)
    \end{equation}
    is continuous, where
    \begin{equation}\label{QQ} Q_{\alpha,z}^\lambda(\rho_{AB}|\sigma_B) = \Trm\left(\rho_{AB}^\frac{\alpha}{2z} I_A \otimes \rho_B^{\frac{(1-\lambda)(1-\alpha)}{2z}}\sigma_B^{\frac{\lambda(1-\alpha)}{z}} \rho_B^{\frac{(1-\lambda)(1-\alpha)}{2z}}\rho_{AB}^\frac{\alpha}{2z}\right)^z.
    \end{equation}
        
    To prove continuity, we first prove upper semicontinuity and then lower semicontinuity.
    
    First, we assume that $\lambda(1-\alpha)<0$. In this case, $Q_{\alpha,z}^\lambda(\rho_{AB}|\sigma_B)$ is convex in $\sigma_B$, and by definition
    \begin{equation} Q_{\alpha,z}^\lambda(A|B)_\rho=\inf_{\sigma_B}Q_{\alpha,z}^\lambda(\rho_{AB}|\sigma_B).
    \end{equation}
    Let $\rho_{AB}^{(n)}$ be a sequence in $\mathcal{S}(AB)$ converging to $\rho_{AB}^* \in \mathcal{S}(AB)$. Inspired by the proof of~\cite[Theorem 1]{lami2024continuity}, for $t \in (0,1)$ and nonnegative integer $n$, we define
    \begin{equation}
        Y_{t,n}:=\inf_{\sigma_B} Q_{\alpha,z}^\lambda\big(\rho_{AB}^{(n)}\big|(1-t)\sigma_B+t\pi_B\big),
    \end{equation}
    where $\pi_B$ is the maximally mixed state in $\mathcal{S}(B)$. This approach is justified by the fact that \((1-t)\sigma_B + t\pi_B\) remains full-rank for any \( t \in (0,1) \), which guarantees the continuity of the function with respect to the first argument \( \rho_{AB}^{(n)} \). Then, clearly $Q_{\alpha,z}^\lambda(A|B)_{\rho^{(n)}} \leq Y_{t,n}$ since the states of the form $(1-t)\sigma_B + t \pi_B$ form a subset of $\mathcal{S}(B)$. Moreover, by the convexity of $Q_{\alpha,z}^\lambda(\rho_{AB}|\sigma_B)$ in $\sigma_B$, we have
    \begin{align}
        Y_{t,n} &\leq (1-t)\inf_{\sigma_B} Q_{\alpha,z}^\lambda\big(\rho_{AB}^{(n)}\big|\sigma_B\big)+tQ_{\alpha,z}^\lambda\big(\rho_{AB}^{(n)}\big|\pi_B\big) \\
        &=(1-t)Q_{\alpha,z}^\lambda(A|B)_{\rho^{(n)}}+tQ^\lambda_{\alpha,z}\big(\rho_{AB}^{(n)}\big|\pi_B\big).
    \end{align}
    As a result,
    \begin{equation}\label{eq1continuity}
        Q_{\alpha,z}^\lambda(A|B)_{\rho^{(n)}}=\inf_{t \in (0,1)} Y_{t,n}, \qquad \text{for all $n$},
    \end{equation}
    since $Q_{\alpha,z}^\lambda(A|B)_{\rho^{(n)}}-Q^\lambda_{\alpha,z}\big(\rho_{AB}^{(n)},\pi_B\big) \leq 0$. Using the preceding equality, we can write
    \begin{align}
        \limsup_{n\rightarrow\infty} Q_{\alpha,z}^\lambda(A|B)_{\rho^{(n)}} &=\limsup_{n\rightarrow\infty} \inf_{t \in (0,1)} Y_{t,n} \\
        \label{eq2continuity}
        &\leq \inf_{t \in (0,1)} \limsup_{n\rightarrow\infty} Y_{t,n},
    \end{align}
    where we used the minimax inequality
    \begin{equation}
        \limsup_{n\rightarrow\infty} \inf_{t \in (0,1)} f_n(t) \leq \inf_{t \in (0,1)} \limsup_{n\rightarrow\infty} f_n(t),
    \end{equation}
    for any sequence of functions $f_n(t)$. Also note that
    \begin{align}
        \limsup_{n\rightarrow\infty} Y_{t,n} &=\limsup_{n\rightarrow\infty} \inf_{\sigma_B}Q^\lambda_{\alpha,z}\big(\rho_{AB}^{(n)}\big|(1-t)\sigma_B+t\pi_B\big) \\
        &\leq \inf_{\sigma_B} \limsup_{n\rightarrow\infty}Q^\lambda_{\alpha,z}\big(\rho_{AB}^{(n)}\big|(1-t)\sigma_B+t\pi_B\big) \\
        &= \inf_{\sigma_B} Q^\lambda_{\alpha,z}\big(\rho_{AB}^{*}\big|(1-t)\sigma_B+t\pi_B\big)  ,       \label{eq3conyinuity}
    \end{align}
    where we used the fact that $Q^\lambda_{\alpha,z}(\cdot|\sigma)$ is continuous when $\sigma$ is full-rank, implying that
    \begin{equation}\label{PP}
        \limsup_{n\rightarrow\infty} Q^\lambda_{\alpha,z}\big(\rho_{AB}^{(n)}\big|(1-t)\sigma_B+t\pi_B\big)=Q^\lambda_{\alpha,z}\big(\rho_{AB}^{*}\big|(1-t)\sigma_B+t\pi_B\big).
    \end{equation}
    Combining \eqref{eq1continuity}, \eqref{eq2continuity} and \eqref{eq3conyinuity} yields
    \begin{equation}
        \limsup_{n\rightarrow\infty} Q_{\alpha,z}^\lambda(A|B)_{\rho^{(n)}} \leq Q_{\alpha,z}^\lambda(A|B)_{\rho^{*}}.
    \end{equation}
    Since the sequence $\rho_{AB}^{(n)}$ was arbitrary, this means that the function $\rho\mapsto Q_{\alpha,z}^\lambda(A|B)_\rho$ is upper semicontinuous.
    
    We now turn to lower semicontinuity. After possibly extracting a subsequence, we may assume without loss of generality that the sequence \(\{Q_{\alpha,z}^\lambda(A|B)_{\rho^{(n)}}\}_n\) converges (note that by Corollary \ref{dimension bound}, $Q_{\alpha,z}^\lambda(A|B)_\rho$ is bounded away from infinity). Moreover, let $\{\omega_B^{(n)}\}_n$ be the sequence in $\mathcal{S}(B)$ of optimizers such that
    \begin{equation} Q_{\alpha,z}^\lambda(A|B)_{\rho^{(n)}}=Q_{\alpha,z}^\lambda\big(\rho_{AB}^{(n)}\big|\omega_B^{(n)}\big).
    \end{equation}
    Since $\mathcal{S}(B)$ is compact, there always exists a converging subsequence $\{\omega_B^{(n_k)}\}_k$ of the sequence $\{\omega_B^{(n)}\}_n$ converging to a state $\omega_B^{*}\in \mathcal{S}(B)$. Now, we observe that
    \begin{align}
        Q_{\alpha,z}^\lambda(A|B)_{\rho^{*}} &\leq Q_{\alpha,z}^\lambda\big(\rho_{AB}^{*}\big|\omega_B^{*}\big)\leq \liminf_{k\rightarrow\infty}Q_{\alpha,z}^\lambda\big(\rho_{AB}^{(n_k)}\big|\omega_B^{(n_k)}\big)\\ &=\liminf_{k\rightarrow\infty}Q_{\alpha,z}^\lambda(A|B)_{\rho^{(n_k)}}=\liminf_{n\rightarrow\infty}Q_{\alpha,z}^\lambda(A|B)_{\rho^{(n)}},
    \end{align}
    where the second inequality follows from the lower semicontinuity of the function $(\rho_{AB},\sigma_B)\mapsto Q_{\alpha,z}^\lambda(\rho_{AB}|\sigma_B)$ proved in Lemma~\ref{lemma semicontintuity} and the last equality follows from the fact that we picked a converging sequence \(\{Q_{\alpha,z}^\lambda(A|B)_{\rho^{(n)}}\}_n\). 

    \medskip
    The proof for the case where $\lambda(1-\alpha)>0$, is completely analogous. The crucial point is that in this case the function $Q_{\alpha,z}^\lambda(\rho_{AB}|\sigma_B)$ is concave in $\sigma_B$, and
    \begin{equation} Q_{\alpha,z}^\lambda(A|B)_\rho=\sup_{\sigma_B}Q_{\alpha,z}^\lambda(\rho_{AB}|\sigma_B).
    \end{equation}
    The rest of the proof is similar with obvious modifications.
    \end{proof}

\section{Auxiliary lemmas}
In this Appendix, we list several well-known facts that we repeatedly use throughout the text. 

The first result serves as a fundamental building block for establishing concavity and convexity statements across various contexts.
We have~\cite[Proposition 5]{carlen18_alphaz}.
\begin{lemma}
\label{primitive concavity/convexity}
Let $K \in \mathcal{L}(A)$ and $A\in \mathcal{P}(A)$. Then, 
\begin{equation}
A \mapsto \textup{Tr}\left(KA^pK^\dagger\right)^s
\end{equation}
\begin{enumerate}[itemsep=1pt]
\item is concave if $0 \leq p \leq 1$ and $0 < s \leq \frac{1}{p}$;
\item is convex if $-1 \leq p \leq 0$ and $s>0$;
\item is convex if $1 \leq p \leq 2$ and $s \geq \frac{1}{p}$.
\end{enumerate}
\end{lemma}
The primary proof of the data processing inequality (DPI) is based on variational characterizations. The earliest variational form of this nature was established for sandwiched Rényi divergences in \cite{frank13_sandwiched} (see also \cite{berta2017variational}), and was subsequently generalized to $\alpha$-$z$ R\'enyi relative entropies in \cite[Theorem 3.3]{zhang20_alphaz}, \cite[Theorem 2.4]{hiai2024alpha}, and ~\cite[Lemma III.23]{mosonyi2023strong} for $\alpha > 1$. In the following, we denote by $\mathcal{L}(A)^{\times}$ the set of invertible linear operators on $A$.

\begin{lemma}\label{variationalZhang}
	For $r_i>0,i=0,1,2$ such that $\frac{1}{r_0}=\frac{1}{r_1}+\frac{1}{r_2}$, we have for any $X,Y\in\mathcal{L}(A)^{\times}$ that
	\begin{equation}\label{equ:variational method min-concave}
	\Trm|XY|^{r_0}=\min_{Z\in\mathcal{L}(A)^{\times}}\left\{\frac{r_0}{r_1}\Trm|XZ|^{r_1}+\frac{r_0}{r_2}\Trm|Z^{-1}Y|^{r_2}\right\},
	\end{equation}
	and 
	\begin{equation}\label{equ:variational method max-convex}
	\Trm|XY|^{r_1}=\max_{Z\in\mathcal{L}(A)^{\times}}\left\{\frac{r_1}{r_0}\Trm|XZ|^{r_0}-\frac{r_1}{r_2}\Trm|Y^{-1}Z|^{r_2}\right\}.
	\end{equation}
\end{lemma}

Moreover, we have the additional variational form that involves the optimization over quantum states~\cite{lennert13_renyi} (see also~\cite[Lemma 3.4]{tomamichel16_book}).
\begin{lemma}
\label{optimization norm}
Let $Y\geq 0$ and $-1 \leq \mu \leq 1$. Then the solution of $\opt_{\sigma}\Tr(Y\sigma^\mu)$ is $\tau \propto Y^\frac{1}{1-\mu}$ which yields the value $\big(\textup{Tr}\big(Y^\frac{1}{1-\mu}\big)\big)^{1-\mu}$.
\end{lemma}

The next lemma demonstrates the equivalence between additive and multiplicative variational forms. This is a well-established result, which can be derived, for example, by requiring that all inequalities in the proof of~\cite[Theorem 3.3]{zhang20_alphaz} hold as equalities.

\begin{lemma}\label{additive/multiplicative}
 Let \( r_0, r_1, r_2 \) be positive numbers such that
 \begin{align}
     \frac{1}{r_0} = \frac{1}{r_1} + \frac{1}{r_2}.
 \end{align}
Moreover, let \( f, g \) be two positive matrix functions satisfying
\begin{align}
    f(\lambda X) = \lambda^{r_1} f(X) \quad \text{and} \quad g(\lambda X) = \lambda^{-r_2} g(X).
\end{align}
  Then, the following equality holds: 
  \begin{align}
     \inf_{X > 0} \left\{ \frac{r_0}{r_1} f(X) + \frac{r_0}{r_2} g(X) \right\} = \inf_{X > 0} \left\{ f(X)^{\frac{r_0}{r_1}} g(X)^{\frac{r_0}{r_2}} \right\}. 
  \end{align}
\end{lemma}

\begin{proof} 
   The weighted arithmetic-geometric mean inequality states that for \( \alpha \in (0,1) \), 
   \begin{align}
        \alpha x + (1 - \alpha) y \geq x^\alpha y^{1 - \alpha}.
   \end{align}
   Applying this with \( \alpha = \frac{r_0}{r_1} \) (noting that \( 1 - \alpha = \frac{r_0}{r_2} \)), we obtain  
   \begin{align}
        \frac{r_0}{r_1} f(X) + \frac{r_0}{r_2} g(X) \geq f(X)^{\frac{r_0}{r_1}} g(X)^{\frac{r_0}{r_2}}.
   \end{align}
   Taking the infimum over \( X > 0 \) preserves the inequality, yielding  
\begin{align}
    \inf_{X > 0} \left\{ \frac{r_0}{r_1} f(X) + \frac{r_0}{r_2} g(X) \right\} \geq \inf_{X > 0} \left\{ f(X)^{\frac{r_0}{r_1}} g(X)^{\frac{r_0}{r_2}} \right\}.
\end{align}
 Let \( X > 0 \) be any feasible point for the right-hand side. Define \( \lambda X \), where  
 \begin{align}
      \lambda = f(X)^{\frac{r_0 - r_1}{r_1^2}} g(X)^{\frac{r_0}{r_2 r_1}} = f(X)^{-\frac{r_0}{r_1 r_2}} g(X)^{-\frac{r_0 - r_2}{r_2^2}}.
 \end{align}
   Since \( \lambda X > 0 \), it is a valid point for the left-hand side. Substituting \( \lambda X \) into the left-hand expression, we find that it equals the right-hand side value. Therefore,
   \begin{align}
        \inf_{X > 0} \left\{ \frac{r_0}{r_1} f(X) + \frac{r_0}{r_2} g(X) \right\} \leq \inf_{X > 0} \left\{ f(X)^{\frac{r_0}{r_1}} g(X)^{\frac{r_0}{r_2}} \right\}.
   \end{align}
   Combining both inequalities completes the proof.  
\end{proof}

The following result relates to complex interpolation theory and offers a version of Hadamard's three-line theorem specifically for Schatten norms, as established in~\cite{beigi13_sandwiched}.
\begin{lemma}
\label{Hadamard}
Let $F:S \rightarrow L(\mathcal{H})$ be a bounded map that is holomorphic in the interior of $S$ and continuous up to the boundary. Assume that $1 \leq p_0,p_1 \leq \infty$ and for $0<\theta<1$ define $p_\theta$ by
\begin{equation}
\frac{1}{p_\theta} = \frac{1-\theta}{p_0}+\frac{\theta}{p_1}.
\end{equation}
For $k=0,1$, define 
\begin{equation}
M_k=\sup_{t \in \mathbb{R}}\|F(k+it)\|_{p_k}.
\end{equation}
Then, we have
\begin{equation}
\|F(\theta)\|_{p_{\theta}}\leq M_0^{1-\theta}M_1^\theta.
\end{equation}
\end{lemma}

The next two lemmas will help us derive a variational form using Schatten norms, which we use in Section~\ref{sec:chainrules} to prove chain rules for the new conditional entropy $H_{\alpha,z}^\lambda$. We refer to~\cite[Section 2.4]{Watrous} and~\cite[Lemma 11-12]{dupuis2015chain} for more details.
\begin{lemma}
\label{Dupuis Lemma 12}
    Let $|\psi\rangle_{AB}$  be a rank-one vector, possibly unnormalized. Then, 
    \begin{equation}
        \langle\psi\ket{\psi} = \Trm \left(\Op \!\!\,_{A\rightarrow B}(\ket{\psi})^\dagger \Op \!\!\,_{A\rightarrow B}(\ket{\psi})\right)
    \end{equation}
and therefore, 
\begin{equation}
    \|\ket{\psi}\| = \| \Op \!\!\,_{A\rightarrow B}(\ket{\psi})\|_2 \,.
\end{equation}
\end{lemma}
\begin{lemma}
\label{Dupuis Lemma 11}
    Let $|\psi\rangle_{AB}$ be a rank-one vector, possibly unnormalized and let $X_A \in \mathcal{L}(A)$ and $Y_B \in \mathcal{L}(B)$. Then, we have that
    \begin{equation}
        \Op\!\!\,_{A \rightarrow B}(X_A \otimes Y_B \ket{\psi}) = Y_B \Op\!\!\,_{A \rightarrow B}(\ket{\psi})X_A^\top \,.
    \end{equation}
\end{lemma}

\section{Stronger convexity/concavity statements}
Leveraging the DPI result from Theorem~\ref{DPItheorem}, and applying a standard argument that connects DPI with convexity, we derive additional results concerning the convexity and concavity properties.

\begin{proposition}
Let $(\alpha,z,\lambda) \in \D$. Then, the function 
\begin{align}
    \rho \mapsto \exp(\frac{1-\alpha}{1-\lambda(1-\alpha)} H^\lambda_{\alpha,z}(A|B)_{\rho} )
\end{align} 
is concave for $\alpha < 1$ and convex for $\alpha > 1$.
In particular, the function $\rho_{AB}\rightarrow H_{\alpha,z}^\lambda(A|B)_{\rho}$ is concave for $\alpha<1$ and quasi-concave for $\alpha>1$.
\end{proposition}
\begin{proof}
Let us consider the case $\alpha < 1$. From the DPI under partial trace on the conditioning system $B$', we have
\begin{align}
\opt_{\sigma_B}Q_{\alpha,z}^\lambda\Big(\sum_i p_i \rho_{AB}^i\Big|\sigma_B\Big)& \geq \opt_{\sigma_{BB'}}Q_{\alpha,z}^\lambda\Big(\sum_{i} p_i \rho^i_{AB} \otimes \ketbra{i}{i}_{B'}\Big|\sigma_{BB'}\Big) \\
&= \Bigg(\sum_i p_i \Big(\opt_{\sigma_B}Q_{\alpha,z}^\lambda(\rho^i_{AB}\big|\sigma_B)\Big)^{\frac{1}{1-\lambda(1-\alpha)}}\Bigg)^{1-\lambda(1-\alpha)},
\end{align}
where, in the last line, we solved the optimization on the classical register (see Proposition~\ref{classical conditioning}).
The case $\alpha > 1$ is analogous. The statement for \( H^\lambda_{\alpha,z} \) follows by the properties of the logarithm of a concave or convex positive function.
\end{proof}
\begin{remark}
Note that for the case $\lambda > 0$, the previous statement gives a stronger concavity/convexity compared to the usual concavity/convexity of $\rho_{AB} \rightarrow Q_{\alpha,z}^\lambda(A|B)_{\rho_{AB}}$. However, in the case $\lambda < 0$, the function $\rho_{AB} \rightarrow Q_{\alpha,z}^\lambda(A|B)_{\rho_{AB}}$ is neither concave nor convex, as demonstrated by examining its behavior under the mixture of states $\sum_{i} p_i \rho^i_{ABB'}$, where $\rho_{ABB'}= \rho^i_{AB} \otimes \ketbra{i}{i}_{B'}$. However, log-convexity/concavity still holds as established in Proposition~\ref{log-concavity}.
\end{remark}

\newpage
\section{List of symbols}

The table below presents the notation employed throughout this paper.

    \begin{table}[h!]
\centering
\renewcommand{\arraystretch}{1.3}
\begin{tabular}{c l}
\toprule
\textbf{Symbol} & \textbf{Definition} \\
\midrule
$A$, $B$, $C$ & \parbox[t]{12cm}{\raggedright Quantum systems (subsystems)} \\
$\rho_{AB}$ & \parbox[t]{12cm}{\raggedright Bipartite quantum state of systems $A$ and $B$} \\
$\sigma_B$, $\sigma_C$ & \parbox[t]{12cm}{\raggedright Quantum states on systems $B$ and $C$} \\
$D(\rho \|\sigma)$ & \parbox[t]{12cm}{\raggedright Umegaki relative entropy} \\
$\widetilde{D}_\alpha(\rho\|\sigma)$ & \parbox[t]{12cm}{\raggedright Sandwiched R\'enyi relative entropy} \\
$\widebar{D}_\alpha(\rho\|\sigma)$ & \parbox[t]{12cm}{\raggedright Petz--R\'enyi relative entropy} \\
$D_{\alpha,z}(\rho \|\sigma)$ & \parbox[t]{12cm}{\raggedright $\alpha$-$z$ R\'enyi relative entropy} \\
$H(A|B)_\rho$ & \parbox[t]{12cm}{\raggedright Conditional von Neumann entropy} \\
$\widetilde{H}_{\alpha}^\uparrow(A|B)_\rho$ & \parbox[t]{12cm}{\raggedright Sandwiched conditional R\'enyi entropy (arrow up) at $\rho_{AB}$} \\
$\widetilde{H}_{\alpha}^\downarrow(A|B)_\rho$ & \parbox[t]{12cm}{\raggedright Sandwiched conditional R\'enyi entropy (arrow down) at $\rho_{AB}$} \\
$\widebar{H}^\uparrow_{\alpha}(A|B)_\rho$ & \parbox[t]{12cm}{\raggedright Petz conditional R\'enyi entropy (arrow up) at $\rho_{AB}$} \\
$\widebar{H}^\downarrow_{\alpha}(A|B)_\rho$ & \parbox[t]{12cm}{\raggedright Petz conditional R\'enyi entropy (arrow down) at $\rho_{AB}$} \\
$H^\uparrow_{\alpha,z}(A|B)_\rho$ & \parbox[t]{12cm}{\raggedright $\alpha$-$z$ conditional R\'enyi entropy (arrow up) at $\rho_{AB}$} \\
$H^\downarrow_{\alpha,z}(A|B)_\rho$ & \parbox[t]{12cm}{\raggedright $\alpha$-$z$ conditional R\'enyi entropy (arrow down) at $\rho_{AB}$} \\
$H^\lambda_{\alpha,z}(A|B)_{\rho}$ & \parbox[t]{12cm}{\raggedright New conditional entropy at $\rho_{AB}$} \\
$\D$ & \parbox[t]{12cm}{\raggedright DPI region of the new conditional entropy $H^\lambda_{\alpha,z}$} \\
$(\hat{\alpha},\hat{z},\hat{\lambda})$ & \parbox[t]{12cm}{\raggedright Dual parameters to $(\alpha,z,\lambda)$ defined through entropic duality relations \eqref{duality2}} \\
$\opt$ & \parbox[t]{12cm}{\raggedright Optimization over quantum states; it can be either an $\inf$ or $\sup$, depending on whether the exponent of the optimized states is negative or positive.} \\
$\M_{A \rightarrow A'}$ & \parbox[t]{12cm}{\raggedright Quantum channel from system $A$ to system $A'$} \\
$\Op_{A \rightarrow B}$ & \parbox[t]{12cm}{\raggedright Operator-vector identification defined through \eqref{operator-vector}} \\
$Q_{\alpha,z}^\lambda(\rho_{AB}|\sigma_B)$ & \parbox[t]{12cm}{\raggedright Trace functional defined in \eqref{Q}} \\
$Q_{\alpha,z}^\lambda(A|B)_\rho$ & \parbox[t]{12cm}{\raggedright Optimized trace functional defined in \eqref{Q2}} \\
$\Upsilon_{\alpha,z}^\lambda(\rho,\tau,\sigma)$ & \parbox[t]{12cm}{\raggedright Trivariate function defined in Def.~\ref{trivariate}} \\
$\Trm(\cdot)$ & \parbox[t]{12cm}{\raggedright Trace operation} \\
$\Trm_A(\cdot)$ & \parbox[t]{12cm}{\raggedright Partial trace operation with respect to subsystem $A$} \\
$A \ll B$ & \parbox[t]{12cm}{\raggedright $A$ is dominated by $B$, i.e. the kernel of $A$ contains the kernel of $B$} \\
$|A|$ & \parbox[t]{12cm}{\raggedright The absolute value of $A$; $(AA^\dagger)^\frac{1}{2}$} \\
$\|A\|_p$ & \parbox[t]{12cm}{\raggedright Schatten $p$-norm; $\left(\Trm|A|^p\right)^\frac{1}{p}$} \\
\bottomrule
\end{tabular}
\caption{Symbols and definitions}
\end{table}
\end{document}